\documentclass[11pt]{article}
\usepackage{hyperref}
\usepackage{amssymb}
\usepackage{amsmath}
\usepackage{mdframed}
\usepackage{subfig}
\usepackage{cancel}
\usepackage{enumerate}
\usepackage{color}
\usepackage{mathtools}
\usepackage{pdfpages}
\usepackage{mathrsfs}
\usepackage{enumitem}
\usepackage{mathtools}
\usepackage{graphicx}
\usepackage{mdframed}
\usepackage{amsfonts}
\usepackage{cancel}
\usepackage{placeins}
\usepackage{amsthm}
\usepackage{tikz}
\usetikzlibrary{shapes,arrows,automata}
\usepackage{xcolor}
\usetikzlibrary{arrows.meta}

\definecolor{goldenpoppy}{rgb}{0.99, 0.76, 0.0}
\definecolor{richblack}{rgb}{0.06, 0.05, 0.03}
\definecolor{cadmiumred}{rgb}{0.89, 0.0, 0.13}
\definecolor{fuchsia}{rgb}{0.3, 0.0, 0.3}
\definecolor{green(ncs)}{rgb}{0.0, 0.52, 0.32}
\tikzstyle{species_T} = [circle,radius=0.1cm, text centered, draw=black, fill=goldenpoppy]
\tikzstyle{species_R} = [circle,radius=0.1cm, text centered, draw=black, fill=white]
\tikzstyle{species_C} = [circle,radius=0.1cm, text centered, draw=black, fill=fuchsia]
\tikzstyle{dots} = [circle,radius=0.1cm,text centered]
\tikzstyle{arrowA} = [-{Latex[length=2mm]},white,dashed]
\tikzstyle{arrowB} = [-{Latex[length=2mm]},green(ncs)]
\tikzstyle{arrowC} = [-{Latex[length=2mm]},cadmiumred]
\tikzstyle{inhibit} = [thick,-|,black!100]
\tikzstyle{loosely dashed}= [dash pattern=on 3pt off 6pt]

\newtheorem{definition}{Definition}[section]
\newtheorem{theorem}{Theorem}[section]

\newtheorem{example}{Example}[section]

\let\Item\item

\begin{document}
\vspace*{-1cm}

\centerline{{\huge Robust control of biochemical reaction networks via}}

\medskip

\centerline{{\huge stochastic morphing}}
 
\medskip
\bigskip

\centerline{
\renewcommand{\thefootnote}{$*$}
{\Large Tomislav Plesa\footnote{
Department of Bioengineering, Imperial College London,
Exhibition Road, \\ London, SW7 2AZ, UK;
e-mail: t.plesa@ic.ac.uk}
\qquad 
Guy-Bart Stan$^*$
\qquad 
Thomas E. Ouldridge$^*$
\qquad 
Wooli Bae$^*$
}}

\medskip
\bigskip

\noindent
{\bf Abstract}:
Synthetic biology is an interdisciplinary field aiming to
design biochemical systems with desired behaviors. 
To this end, molecular controllers have been developed which, 
when embedded into a pre-existing ambient biochemical network,
control the dynamics of the underlying target molecular species.
When integrated into smaller compartments, such as 
biological cells in vivo, or vesicles in vitro, 
controllers have to be calibrated to factor in the intrinsic noise. 
In this context, molecular controllers 
put forward in the literature have focused on manipulating the mean (first moment), 
and reducing the variance (second moment), of the target species. 
However, many critical biochemical processes
are realized via higher-order moments, particularly 
the number and configuration of the modes (maxima) of the probability distributions.
To bridge the gap, a controller called \emph{stochastic morpher} is
put forward in this paper, inspired by gene-regulatory networks, 
which, under suitable time-scale separations, morphs the 
probability distribution of the target species 
into a desired predefined form. The morphing can be 
performed at the lower-resolution, allowing one to achieve
desired multi-modality/multi-stability, and at the higher-resolution, 
allowing one to achieve arbitrary probability distributions.
Properties of the controller, such as robust perfect adaptation
and convergence, are rigorously established, and demonstrated on 
various examples. Also proposed is a blueprint for an experimental implementation 
of stochastic morpher.

\section{Introduction} \label{sec:intro}
Synthetic biology is a growing interdisciplinary field 
of science and engineering which aims to design 
biochemical systems with predefined behaviors~\cite{SynthBio1}.
In the sub-field of nucleic-acid-based synthetic biology
(also called DNA and/or RNA computing), biochemical
systems are engineered using nucleic acids 
(DNA and/or RNA molecules). 
An advantage of this approach lies in the fact that nucleic acids
have relatively well-understood biophysical properties, 
and their production is systematic and cost-effective.
The key mechanism behind the excellent programmability
properties of DNA and RNA, allowing for a controllable and dynamic
change in the structure of these molecules, is the 
\emph{toehold-mediated strand-displacement mechanism}~\cite{Experiment5,RNAComputing}.
 The strand-displacement mechanism involves a single-stranded nucleic acid displacing 
another one from a duplex, as a consequence of the Watson-Crick base-pairing principle,
and allows one to realize dynamical systems~\cite{Sulc,Tom,Experiment1,DNAComputing1}.
In particular, a large class of abstract mass-action biochemical 
reaction networks (see also Appendix~\ref{app:CRNs}
for a background on reaction networks) can be physically realized
using strand-displacement DNA computing~\cite{DNAComputing1}. A proof-of-concept
is the displacillator - a purely DNA-based synthetic oscillator implemented in vitro~\cite{Displacillator}.
Nucleic acids play some of the key roles inside living systems, 
involving storage and transfer of information, catalysis and a 
variety of regulatory functions. 
Consequently, nucleic-acid-based strand-displacement 
synthetic systems are desirable, as they can be more readily interfaced 
with a variety of key biochemical processes in living systems. 
Let us note that DNA and RNA strand-displacement
is also hypothesized to take place in a number of 
native cellular processes, including
 the genetic recombination process~\cite{Experiment2} 
(see also Section~\ref{sec:experiments}),
CRISPR-Cas systems~\cite{CRISPR}
and co-transcriptional folding of RNA~\cite{Sulc}. 

Depending on applications, synthetic systems can be 
implemented in a variety of different 
environments, each generally requiring different 
engineering approaches. In particular, synthetic systems may be integrated into
larger-volume compartments (such as test-tubes in vitro), 
or smaller-volume compartments (such as biological cells in vivo, 
or cell-like vesicles in vitro).
When integrated into larger-volume compartments, owning to the higher species copy-numbers, 
the task is to design reaction networks with 
desired deterministic dynamics, described by the reaction-rate equations~\cite{Feinberg}.
Mathematical methods for achieving such goals have been 
developed in~\cite{Me_Homoclinic,Me_Limitcycles}.
On the other hand, when integrated into smaller-volume compartments, 
owning the the lower species copy-numbers,
intrinsic noise becomes an important dynamical 
feature~\cite{CellCycle,Circadian,Vesicles1,Vesicles2,Vesicles3}, 
and the synthetic systems have to be constructed via a more-detailed
stochastic approach~\cite{GillespieDerivation} (see also Appendix~\ref{app:stochastic_model}
for a background on the stochastic model of reaction networks).
To this end, so-called \emph{noise-control algorithm} has been developed in~\cite{Me_NAA}, 
which systematically re-designs a given reaction network to 
arbitrarily manipulate its intrinsic-noise profile and reshape its probability distribution, 
while preserving the desired underlying deterministic dynamics. 
Biochemical networks have been successfully implemented in vitro, 
displaying both desirable deterministic dynamics in test-tubes~\cite{Displacillator,Enzymes},
and stochastic dynamics in vesicles~\cite{Vesicles1,Vesicles2,Vesicles3}.

Another important feature of the compartments, aside from their volumes,
is whether the compartments are biochemically active, i.e. infused with pre-existing biochemical 
processes (such as the native molecular machinery inside biological cells), 
or otherwise biochemically inactive (such as suitable test-tubes). 
When biochemically active environments are considered, 
of interest may be isolated synthetic systems, 
i.e. systems which, in an ideal case, execute predefined
dynamics without altering their biochemical environment.
On the other hand, one may also be interested 
in so-called controller networks, which are 
designed to couple to their active surroundings in order 
to manipulate the dynamics of some of the
underlying ambient biochemical species. 
Controller networks, based on the RNA strand-displacement, have
 been successfully engineered to manipulate 
 transcription and post-transcription stages of gene expression inside living cells, 
e.g. controlling the behavior of RNA polymerases~\cite{PreTranscription}
and editing the structure of messenger RNA (mRNA) molecules~\cite{PostTranscription}, 
respectively. Let us note that compiling abstract reaction networks
into physical ones, and subsequent integrations into 
biochemically active environments, involves overcoming
a number of challenges~\cite{Control3,Retroactivity}, 
including undesirable cross-reactions inside individual, 
and between multiple, synthetic networks, and between 
the synthetic systems and their environments.

\begin{figure}[!htbp]
\vskip -1.2 cm
\centerline{
\begin{tikzpicture}[thick,scale=0.9, every node/.style={scale=0.9}]
 \begin{scope}[fill opacity=0.75]
\filldraw[fill={richblack}, draw=richblack,line width=0.01cm] (0,0) rectangle (5.83,6);
\end{scope}
\node at (3,6.4) (A) {\LARGE $\color{richblack} \mathcal{R}_{\alpha}$};
\node (X_r1) [species_R, xshift= 2 cm,yshift= 5.2 cm] {};
\node (X_r2) [species_R, xshift= 2 cm,yshift= 4 cm] {};
\node (X_r3) [dots, xshift= 2 cm,yshift= 2.9 cm] {\textcolor{white}{$\Huge\vdots$}};
\node (X_r4) [species_R, xshift= 2 cm,yshift=1.6 cm] {};
\node at (1.2,5.3) () {\large {\color{white} $X_{n+1}$}};
\node at (1.2,4) () {\large {\color{white} $X_{n+2}$}};
\node at (1.3,1.6) () {\large {\color{white} $X_{N}$}};
\node at (1.5,0.7) () {\footnotesize {\color{white} \textbf{Residual species}}};
\node at (1.5,0.35) () {\footnotesize {\color{white} $\mathcal{X}_{\rho}$}};
\node (X_t1) [species_T, xshift= 5 cm,yshift= 5.2 cm] {};
\node (X_t2) [species_T, xshift= 5 cm,yshift= 4 cm] {};
\node (X_t3) [dots, xshift= 5 cm,yshift= 2.9 cm] {\textcolor{goldenpoppy}{$\Huge\vdots$}};
\node (X_t4) [species_T, xshift= 5 cm,yshift=1.6 cm] {};
\node at (4.3,5.3) () {\large {\color{goldenpoppy} $X_{1}$}};
\node at (4.3,4) () {\large {\color{goldenpoppy} $X_{2}$}};
\node at (4.3,1.6) () {\large {\color{goldenpoppy} $X_{n}$}};
\node at (4.5,0.7) () {\footnotesize {\color{goldenpoppy} \textbf{Target species}}};
\node at (4.5,0.35) () {\footnotesize {\color{goldenpoppy} $\mathcal{X}_{\tau}$}};
\draw [arrowA] (X_r2) -- (X_t1); \draw [arrowA] (X_t1) -- (X_r2);
 \begin{scope}[fill opacity=0.35]
\filldraw[fill={green(ncs)}, draw=green(ncs),line width=0.01cm] (7.8,0) rectangle (11.2,6);
\end{scope}
\node at (9.5,6.4) (B) {\LARGE $\color{green(ncs)} \mathcal{R}_{\beta}$};
\node (X_c1) [species_C, xshift= 9 cm,yshift= 5.2 cm] {};
\node (X_c2) [species_C, xshift= 9 cm,yshift= 4 cm] {};
\node (X_c3) [dots, xshift= 9 cm,yshift= 2.9 cm] {\textcolor{fuchsia}{$\Huge\vdots$}};
\node (X_c4) [species_C, xshift= 9 cm,yshift=1.6 cm] {};
\node at (9.6,5.3) () {\large {\color{fuchsia} $Y_{1}$}};
\node at (9.6,4) () {\large {\color{fuchsia} $Y_{2}$}};
\node at (9.7,1.6) () {\large {\color{fuchsia} $Y_{M}$}};
\node at (9.5,0.7) () {\footnotesize {\color{fuchsia} \textbf{Controlling species}}};
\node at (9.5,0.35) () {\footnotesize {\color{fuchsia} $\mathcal{Y}$}};
\draw [arrowB] (X_c1) -- (X_c2); \draw [arrowB] (X_c2) -- (X_c1);
\draw [arrowB] (X_c2) -- (9,3.1); \draw [arrowB] (9,3.1) -- (X_c2);
\draw [arrowB] (9,2.43) -- (X_c4); \draw [arrowB] (X_c4) -- (9,2.43);
 \begin{scope}[fill opacity=0.075]
 \filldraw[fill={cadmiumred}, draw=cadmiumred,line width=0.01cm] (3.4,1.15) rectangle (10.5,5.75);
\end{scope}
\node at (7,6.1) (C) {\LARGE $\color{cadmiumred} \mathcal{R}_{\gamma}$};
\draw [arrowC] (8.84,5.3) -- (5.16,5.3); \draw [arrowC] (5.16,5.3) -- (8.84,5.3); 
\draw [arrowC] (X_c1) -- (5.15,4.12);  \draw [arrowC] (5.15,4.12) -- (X_c1); 
\draw [arrowC] (X_c1) -- (5.12,1.74); \draw [arrowC] (5.12,1.74) -- (X_c1); 
\draw [arrowC] (X_c2) -- (X_t2); \draw [arrowC] (X_t2) -- (X_c2); 
\draw [arrowC] (8.89,4.15) -- (5.18,5.175); \draw [arrowC] (5.18,5.175) -- (8.89,4.15); 
\draw [arrowC] (X_c2) -- (5.19,1.6); \draw [arrowC] (5.19,1.6)--(X_c2); 
\draw [arrowC] (X_t1) -- (8.87,1.73);  \draw [arrowC] (8.87,1.73) -- (X_t1);  
\draw [arrowC] (X_t2) -- (8.81,1.6); \draw [arrowC] (8.81,1.6) -- (X_t2);
\draw [arrowC]  (5.158,1.5) -- (8.84,1.5);  \draw [arrowC] (8.84,1.5) -- (5.158,1.5);  
\end{tikzpicture}
}
\vskip 0cm 
\caption{Schematic representation of biochemical control.
{\it A black-box \emph{input} network, $\mathcal{R}_{\alpha} = 
\mathcal{R}_{\alpha} (\mathcal{X})$, is shown in black. 
The input species $\mathcal{X} = \{X_1, X_2, \ldots, X_N\}$ are divided
into two mutually-exclusive sets: 
the \emph{target} species $\mathcal{X}_{\tau} = \{X_1, X_2, \ldots, X_n\}$, and 
the \emph{residual} species $\mathcal{X}_{\rho} = \{X_{n+1}, X_{n+2}, \ldots, X_N\}$, 
which are shown in white and yellow, respectively. The target and residual species
are generally coupled, but the nature of the coupling is at most partially known.   
A known coupling between the species $X_{1}$ and $X_{n+2}$
is depicted as a white dashed double-arrow. 
A \emph{controller} is shown, consisting of the
sub-networks $\mathcal{R}_{\beta} = \mathcal{R}_{\beta}(\mathcal{Y})$
and $\mathcal{R}_{\gamma} = \mathcal{R}_{\gamma}(\mathcal{X}_{\tau}, \mathcal{Y})$, 
displayed in green and red, respectively. The network $\mathcal{R}_{\beta}$
specifies how the \emph{controlling} species $\mathcal{Y} = \{Y_1, Y_2, \ldots, Y_M\}$, 
shown in purple, interact among themselves, which is depicted as the green double-arrows. 
On the other hand, the network $\mathcal{R}_{\gamma}$ specifies how the controlling
species are interfaced with the target species, which is displayed as the red
double-arrows. Embedding a controller $\mathcal{R}_{\beta,\gamma} = \mathcal{R}_{\beta} \cup 
\mathcal{R}_{\gamma}$ into an input network
$\mathcal{R}_{\alpha}$ gives rise to an \emph{output} network
$\mathcal{R}_{\alpha} \cup \mathcal{R}_{\beta} \cup 
\mathcal{R}_{\gamma}$.
}
}
\label{fig:Control_Theory}
\end{figure}
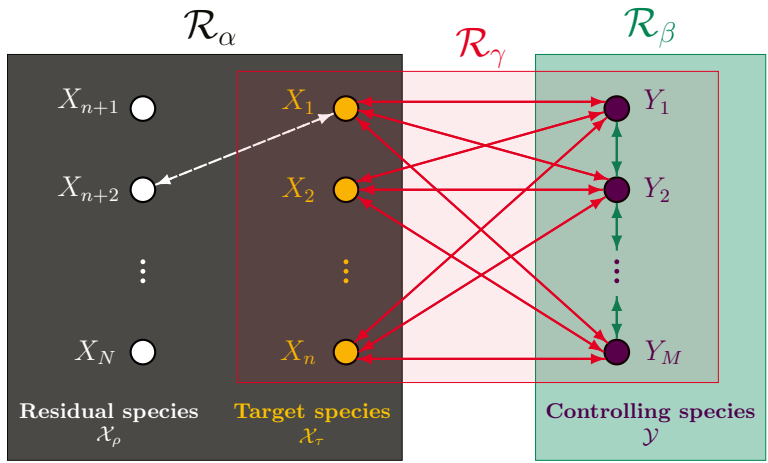

In this paper, we focus on controlling stochastic biochemical reaction networks,
such as those forming biochemically active environments inside 
small-volume compartments, both in vitro (e.g. inside cell-like vesicles)
and in vivo (e.g. inside living cells). 
We call a fixed ambient network, which we wish to control, 
an \emph{input} (uncontrolled) network, shown in black in Figure~\ref{fig:Control_Theory}. 
A critical assumption is that the input network 
is a \emph{black-box}: its fixed structure and the induced dynamics
are at most partially known. For example, if noisy time-series (sample paths) generated 
by a black-box network are experimentally available, dynamical features such as
the average species abundance, time-scales of fluctuations
and bifurcation structures may be inferrable~\cite{Liao,Zoltan,Wooli}.
A key goal of biochemical control is to embed a suitable auxiliary
network, called a \emph{controller},
into a black-box input network, resulting in an 
\emph{output} (controlled) network, which ensures
that a desired subset of the input biochemical species have controlled
stochastic dynamics in the output network. 
A controller consists of a sub-network governing its internal dynamics, 
and a sub-network specifying how the controller
is intefaced with an input network, which
are shown in green and red in Figure~\ref{fig:Control_Theory}, respectively.
We divide the input species
into two mutually-exclusive sets: species 
which are explicitly (directly) 
targeted by the controller are called the \emph{target} species, 
while the remaining species, which may be implicitly 
(indirectly) affected by the controller, are called the \emph{residual} species.
The target and residual species are shown in white and yellow in
Figure~\ref{fig:Control_Theory}, respectively, while the \emph{controlling} species,
introduced by the controller, are shown in purple.
Control may be sought
over the probability distributions of the input biochemical species 
(which we call \emph{weak control}), 
or at the level of the underlying sample paths 
(which we call \emph{strong control}).
Of practical importance is weak control involving the long-time (stationary)
statistics (expectations) of the input species. 
See Appendix~\ref{app:biochemical_control} for more details
on biochemical control.

A controller network must satisfy a set of constraints
in order to be useful and experimentally realizable. 
Structurally, it is desirable that the controller consists of up-to second-order (bi-molecular) 
reactions, i.e. reactions involving at most two reactant 
molecules~\cite{Me_Bimolecular}, in order to be experimentally realizable~\cite{DNAComputing1}.
Let us note that the composition of the bi-molecular reactions is arbitrary 
when implemented via the strand-displacement mechanism, i.e.
the reactions may involve arbitrary product and reactant species. In particular, 
catalytic reactions, involving the same species as both a product and reactant, 
are allowed, as such non-elementary reactions are expanded into suitable 
elementary counterparts when compiled into physical 
networks~\cite{DNAComputing1,Nonelementary,Catalysis1,Catalysis2}.
Kinetically, the rate coefficients of the reactions from a controller, 
realized via the strand-displacement mechanism,
can be varied over at least six orders of magnitude~\cite{Experiment5,Mismatches,RNAComputing}, 
allowing one to achieve time-scale separations.
Dynamically, when the controller is embedded into an input network, 
it is desirable that long-time dynamics of the corresponding
 output network satisfy the following two properties.
Firstly, the stationary probability distribution of the target species
 is independent of the initial conditions for all of the species in the output network 
(a property known as ergodicity in mathematical literature). 
And, secondly, the controlled stationary statistic of the target species does not explicitly 
depend on the parameters (rate coefficients) from the input network. A controller
satisfying these two properties is said to be \emph{robust}, 
see also Appendix~\ref{app:biochemical_control}.

The first property satisfied by a robust controller ensures that
 the dynamics of the target species reach a unique desired probability
distribution in the long-run, even when the precise initial conditions 
for the biochemical species in the output network are unknown 
or experimentally difficult to control. For example, controlled
synthetic systems may be split from a test-tube 
into a large number of cell-like vesicles~\cite{Vesicles1,Vesicles2,Vesicles3}, 
with each vesicle initially containing in general a different
abundance of the underling species. A similar effect occurs during the
division of living cells, which induces an extrinsic noise into the initial
conditions of the corresponding daughter cells~\cite{CellCycle}. 
A robust controller guarantees that, despite the difficult-to-control initial conditions, 
the long-time behavior of the system remains controlled. 

The second property satisfied by a robust controller 
is necessary for systematically controlling the desired statistic, 
as the control may be reached by fine-tuning the parameters 
appearing in the controller, which can be experimentally manipulated, 
without the need to factor in the generally unknown 
and uncontrollable parameters from the underlying black-box input network.
An output statistic which does not explicitly depend on the input parameters is said to
display \emph{robust perfect adaptation}~\cite{Adaptation,CellSignal}.
If such a property holds only in an asymptotic limit of
some of the parameters from the controller (i.e. 
under suitable time-scale separation between the controller and the input network), 
then we say that the statistic displays \emph{asymptotic robust perfect adaptation}.
Perfect and near-perfect adaptations have been hypothesized 
to play important roles in systems biology, e.g. in cell signaling, glycolysis
and chemokinesis~\cite{CellSignal,Glycolic,Chemotaxis1,Chemotaxis2,Chemotaxis3}.
These biochemical processes, and many others, involving 
multi-stability, oscillations, and bifurcations
at both the deterministic and stochastic levels~\cite{CellCycle,Circadian,Bistability,Kepler,
Me_Homoclinic,Me_Limitcycles,Me_NAA,Me_Mixing}, 
 also often display time-scale 
separations (fast-slow dynamics)~\cite{Multiscale}. 
In some of the phenomenological models of such processes, 
the fast dynamics are eliminated (averaged out).
Furthermore, the distinction between perfect and near-perfect adaptations
may be difficult to infer, since the underlying experimental 
biochemical data is subjected to measurement errors.
As a consequence, asymptotic
robust perfect adaptation may play a greater role than its 
non-asymptotic counterpart in biochemical settings. 
In fact, robust perfect adaptation 
is a purely structural property of a reaction network and, as such, 
is fragile under network perturbations (e.g. addition of new reactions).
An instance of the structural fragility has been investigated in~\cite{LeakyControl1}, where it has been shown
that robustness of the integral-feedback controller put forward in~\cite{Khammash}, 
under the addition of degradation reactions in the controlling species, 
is recovered only under an appropriate time-scale separation. 

\begin{table*}[!htbp]
\vskip -2.5 cm
 \renewcommand\tablename{Algorithm}
\hrule
\vskip 2.5 mm
\textbf{Input}: Let the mass-action kinetics input network be given by
\begin{align}
\mathcal{R}_{\alpha} & = \mathcal{R}_{\alpha}(\mathcal{X}), 
\label{eq:input_algorithm}
\end{align}
where $\mathcal{X} = \{X_1, X_2, \ldots, X_N\}$, 
with the rate coefficients $\boldsymbol{\alpha} = (\alpha_1, \alpha_2, \ldots)$.
Assume explicit control is sought over the stochastic dynamics of the
target species $\mathcal{X}_{\tau} = \{X_1, X_2, \ldots, X_n\}
\subseteq \mathcal{X}$, $1 < n \le N$. 

\textbf{Stochastic morpher}: Consider the controller, called the stochastic morpher, given by  
\begin{align}
\mathcal{R}_{\beta,\gamma}(\mathcal{X}_{\tau}, \mathcal{Y}, \mathcal{Z})  & = 
\mathcal{R}_{\beta}(\mathcal{Y}) \cup 
\mathcal{R}_{\gamma}^{\varepsilon}(\mathcal{X}_{\tau}, \mathcal{Z}; \, \mathcal{Y}), 
\label{eq:SM_algorithm}
\end{align}
depending on the target species $\mathcal{X}_{\tau}$, 
and two additional sets of species: the controlling and mediating species,
$\mathcal{Y}$ and $\mathcal{Z}$, respectively.  
The sub-network $\mathcal{R}_{\beta}(\mathcal{Y})$ is given by
\begin{align}
\mathcal{R}_{\beta}(\mathcal{Y}): \; \; 
& & 2 Y_1 & \xrightarrow[]{\beta_{1,1}} Y_1
                      \xrightarrow[]{\beta_{1,2}} Y_2 
                     \xrightarrow[]{\beta_{2,3}} Y_3 \xrightarrow[]{\beta_{3,4}} \ldots 
	                 \xrightarrow[]{\beta_{M-1,M}} Y_{M}
					 \xrightarrow[]{\beta_{M,1}} Y_{1},
\label{eq:R_beta_alg}
\end{align}
with the $M$ controlling species $\mathcal{Y} = \{Y_1, Y_2, \ldots, Y_M\}$, 
whose sum of copy-numbers is assumed to be initially non-zero, $\sum_{i = 1}^M Y_i(0) \ne 0$.
Consider two choices for $\mathcal{R}_{\gamma}^{\varepsilon}(\mathcal{X}_{\tau}, \mathcal{Z}; \, \mathcal{Y})$:
\begin{enumerate}
\item[{\rm (i)}] \textbf{Lower-resolution control}. 
$\mathcal{R}_{\gamma}^{\varepsilon} = \mathcal{R}_{\gamma}^{\mathcal{P}}(\mathcal{X}_{\tau}; \, \mathcal{Y})
= \mathcal{R}_{\gamma_0}^{\varepsilon}(\mathcal{X}_{\tau}; \, \varnothing) 
\cup_{i = 1}^M \mathcal{R}_{\gamma_i}^{\varepsilon}(\mathcal{X}_{\tau}; \, Y_i)$:
\begin{align}
\mathcal{R}_{\gamma_0}^{\varepsilon}(\mathcal{X}_{\tau}; \, \varnothing) : \; & & X_j & 
\xrightarrow[]{\gamma_{0,j}/\varepsilon}  \varnothing, 
\; \; \; \; \; \; \; \; \; \; \; \;  \;
\textrm{for } j \in \{1, 2, \ldots, n\}, \nonumber \\
\mathcal{R}_{\gamma_i}^{\varepsilon}(\mathcal{X}_{\tau}; \, Y_i): \; & & Y_i & 
\xrightarrow[]{\gamma_{i,j}/\varepsilon}  Y_i + X_j, 
\; \; \; \; 
\textrm{for } 
i \in \{1, 2, \ldots, M\},  
\; 
j \in \{1, 2, \ldots, n\}, 
\;  
0 < \varepsilon \ll 1,
\label{eq:Poisson_control_alg}
\end{align} 
with $\gamma_{0,j} > 0$ for $j \in \{1, 2, \ldots, n\}$.
\item[{\rm (ii)}] \textbf{Higher-resolution control}. 
$\mathcal{R}_{\gamma}^{\varepsilon} = \mathcal{R}_{\gamma}^{\delta}(\mathcal{X}_{\tau}, \mathcal{Z}; \, \mathcal{Y})
= \mathcal{R}_{\gamma_0}^{\mu,\varepsilon,\sigma}(\mathcal{X}_{\tau}, \mathcal{Z}; \, \varnothing) 
\cup_{i = 1}^M \mathcal{R}_{\gamma_i}^{\mu,\varepsilon,\sigma}(\mathcal{Z}; \, Y_i)$:
\begin{align}
\mathcal{R}_{\gamma_0}^{\mu,\varepsilon,\sigma}(\mathcal{X}_{\tau}, \mathcal{Z}; \, \varnothing) : \; 
& & \varnothing & \xrightarrow[]{1/\varepsilon}  X_j, 
\; \; \; \; X_j  \xrightleftharpoons[1/\mu]{\gamma_{0,j,1}} Z_{j,1}, \nonumber \\
& & X_j + Z_{j,l} &  \xrightleftharpoons[1/\mu]{\gamma_{0,j,l+1}} Z_{j,l+1},
\; \; \; \; \; \; \; 
\textrm{for } j \in \{1, 2, \ldots, n\},  
\;
l \in \{1, 2, \ldots, c_j-1\}, \nonumber \\
\mathcal{R}_{\gamma_i}^{\mu,\varepsilon,\sigma}(\mathcal{Z}; \, Y_i): \; 
& & Y_i + Z_{j,x_{i,j} + 1} &  \xrightarrow[]{\gamma_{i,j}}  Y_i + Z_{j,x_{i, j}}, 
\; \; \; \; 
\textrm{for } i \in \{1, 2, \ldots, M\},  
\;
 j \in \{1, 2, \ldots, n\}, \nonumber \\
&&&
\; \; \; \; \; \; \; \; \; \; \; 
\; \; \; \; \; \; \; \; \; \; \; 
\; \; \; \; \; \; \; \; \; \; \;  
\; \; \; \; \; \; \; \; \; 
 \{x_{i,j}\}_{i = 1}^M \in \{0, 1,\ldots, c_j-1\},
\label{eq:Kronecker_control_alg}
\end{align} 
\end{enumerate}
where $\mathbf{c} = (c_1, c_2, \ldots, c_n) \in \mathbb{Z}_{>}^{n}$ is the truncation vector,
$\mathcal{Z} =\{\{Z_{j,l}\}_{j = 1}^{n}\}_{l =0}^{c_j}$ are the auxiliary species, 
with $Z_{j,0} \equiv \varnothing$. 
The rate coefficients from~(\ref{eq:Kronecker_control_alg}) are assumed to satisfy
 the kinetic conditions given by~(\ref{eq:kinetic_conditions_alg})
in Appendix~\ref{app:analysis}, with $0 < \mu \ll \varepsilon, \sigma \ll 1$. 

\textbf{Output}: Embedding the stochastic morpher~(\ref{eq:SM_algorithm})--(\ref{eq:Kronecker_control_alg})
into the input network~(\ref{eq:input_algorithm}), gives an output network
\begin{align}
\mathcal{R}_{\alpha,\beta,\gamma}(\mathcal{X}, \mathcal{Y}, \mathcal{Z}) & = 
 \mathcal{R}_{\alpha}(\mathcal{X}) \cup \mathcal{R}_{\beta}(\mathcal{Y}) \cup 
\mathcal{R}_{\gamma}^{\varepsilon}(\mathcal{X}_{\tau}, \mathcal{Z}; \, \mathcal{Y}), \label{eq:output_algorithm}
\end{align}
whose species $\mathcal{X}_{\tau}$, under suitable assumptions, have controlled stochastic dynamics.
In particular, the stationary PMF of $\mathcal{X}_{\tau}$ is 
a linear combination of Poisson distributions centered at the points
$(x_1, x_2, \ldots, x_n) = (\gamma_{i,1}/\gamma_{0,1},  \gamma_{i,2}/\gamma_{0,2}, 
\ldots, \gamma_{i,n}/\gamma_{0,n})$,
 if $\mathcal{R}_{\gamma} =  \mathcal{R}_{\gamma}^{\mathcal{P}}$, 
and Kronecker-delta distributions centered at the points 
$(x_1, x_2, \ldots, x_n) = (x_{i,1}, x_{i,2}, 
\ldots, x_{i,n})$, with $x_{i,j} \in  [0, c_j-1]$,
 if $\mathcal{R}_{\gamma} =  \mathcal{R}_{\gamma}^{\delta}$,
 for $i \in \{1, 2, \ldots, M\}$,  
see Theorem~\ref{theorem:Poisson_Kronecker} in Appendix~\ref{app:analysis}. 
\smallskip
\hrule
\smallskip
\caption{{\it \noindent The algorithm for control of biochemical networks using the stochastic morpher.}}
\label{alg:SM_algorithm}
\end{table*}

Biochemical controllers developed in the literature 
have been largely focused on manipulating the stationary
mean (first-moment) of the target species~\cite{Control1,Control2,Khammash}, 
and reducing their stationary variance (second-moment)~\cite{Khammash2,Cardelli}. 
Such an approach may be seen as a step forward from controlling the 
deterministic dynamics of the
target species, to which the underlying stochastic dynamics converge in
 the thermodynamic limit~\cite{Kurtz}. 
However, many biochemical phenomena in systems biology, 
including cellular differentiation and memory, 
quorum sensing, bacterial chemokinesis and antibiotic resistance, are realized
via higher-order moments of the underlying probability 
distributions~\cite{Switching,Bistability,Swarming,Kepler,Me_Mixing,Me_NAA}. 
Particularly important is the number and configuration of
the modes (maxima) of the probability distributions, 
and the timing and pattern of stochastic 
switching in the underlying sample paths.
Such dynamically exotic and biochemically important phenomena cannot be achieved
using controllers which target only the mean and variance. 
To bridge the gap, in this paper, we develop a robust controller
called \emph{stochastic morpher}, presented in Algorithm~\ref{alg:SM_algorithm}, 
which is inspired by the stochastic phenomenon called 
\emph{noise-induced mixing}~\cite{Me_Mixing}, 
which some gene-regulatory networks utilize for dynamical control in vivo~\cite{Kepler}.
As appropriate reactions from the stochastic morpher fire faster,
 overriding the reactions from the underlying black-box input network, 
the probability distribution of the target species, from the corresponding output network, 
 gradually transforms (morphs)
into a desired predefined form.
More precisely, control may be achieved, under suitable time-scale separations, 
at two different levels of resolution:
at the lower-resolution level, and at a lower biochemical cost, 
one may control the number and configuration
of the modes in the target multi-modal probability distribution (weak control), 
and the mean timing and pattern of stochastic switching
 in the underlying multi-stable sample paths (strong control).
At the higher-resolution level, and at a higher biochemical cost, 
one may achieve arbitrary 
target stationary distributions on bounded state-spaces 
(control over all of the stationary moments). 
The achieved probability distributions, and hence all of the underlying moments, 
display asymptotic robust perfect adapation.

The rest of the paper is organized as follows. 
In Section~\ref{sec:example1}, we introduce 
the lower- and higher-resolution control from
Algorithm~\ref{alg:SM_algorithm} by applying it 
on the one-species first-order (uni-molecular)
production-degradation test network~(\ref{eq:input_1}). 
In Section~\ref{sec:example2}, we focus on the
lower-resolution control in greater detail, by applying it
 on the three-species second-order (bi-molecular) 
test network~(\ref{eq:input_2}). We explicitly jointly control 
 two input biochemical species (target species), and 
outline how the remaining (residual) species is implicitly affected. 
In Section~\ref{sec:example3}, we apply 
Algorithm~\ref{alg:SM_algorithm} on the gene-expression 
network~(\ref{eq:input_3}), and demonstrate how 
implicit control may be achieved. In particular, 
we explicitly influence the mRNA (target species) in 
a suitable way, ensuring that the translated 
protein (residual species) is implicitly controlled. 
In Section~\ref{sec:experiments}, we put forward a blueprint for an experimental
realization of the stochastic morpher, using strand-displacement
DNA nanotechnology. Finally, we conclude by presenting a summary and discussion 
 in Section~\ref{sec:discussion}.
The notation and background theory utilized in the paper are introduced as needed,
and are summarized in Appendix~\ref{app:background}.
General properties and convergence of the stochastic morpher, outlined via
specific examples in Sections~\ref{sec:example1}--\ref{sec:example3}, 
are rigorously established in Appendix~\ref{app:analysis}.

\section{Production-degradation input network} \label{sec:example1}
Consider the one-species uni-molecular
 input network $\mathcal{R}_{\alpha}^1 = \mathcal{R}_{\alpha}^1(X)$,
under mass-action kinetics, given by
\begin{align}
\mathcal{R}_{\alpha}^1:
\hspace{1.0cm} \varnothing  \xrightleftharpoons[\alpha_{2}]{\alpha_{1}} X,
\label{eq:input_1}
\end{align}
where we adopt the convention of denoting two 
irreversible reactions (in this case, 
$\varnothing \to X$ and $X \to \varnothing$)
jointly as a single reversible reaction (in this case, 
$\varnothing \xrightleftharpoons[]{} X$), 
for notational convenience. 
In this paper, biochemical species, and their copy-numbers
as a function of time $t$, are represented with
upper-case letters (such as $X$, and $X(t)$, respectively), 
while the copy-number values are denoted by the corresponding lower-case letters (such as $x$), 
with the latter being elements of the set of non-negative integers, 
denoted by $\mathbb{Z}_{\ge}$.
Symbol $\varnothing$ denotes biochemical species which 
are not explicitly taken into an account. 
Furthermore, for simplicity, we assume
the non-negative rate coefficients, displayed above the reaction arrows,
 are dimensionless, and we denote them using the same letter as the sub-script 
of the corresponding reaction network. 
See also Appendix~\ref{app:background} for a summary of the 
notation used in this paper. 

In what follows, we fix the (dimensionless) rate coefficients
of the input network $\mathcal{R}_{\alpha}^1(X)$ to 
$\boldsymbol{\alpha} = (\alpha_1, \alpha_2) = (1, 1/15)$. 
The stationary probability-mass function (PMF) 
of~(\ref{eq:input_1}), describing the long-time dynamics of the input network 
and denoted by $p_{\infty}(x)$,
is given by the Poisson distribution with mean $\alpha_1/\alpha_2$ 
(in this paper, we also say that the Poisson distribution is centered at 
$\alpha_1/\alpha_2$), denoted by 
$p_{\infty}(x) = \mathcal{P}(x; \, \alpha_1/\alpha_2)$. 
For the particular choice of the rate coefficients, the Poisson PMF
is centered at $x = \alpha_1/\alpha_2 = 15$
and is shown as the black dots, interpolated with solid black lines for visual clarity,
in Figures~\ref{fig:Poisson}--\ref{fig:Kronecker}. 
In the rest of this section, we embed different controllers
into the input network~(\ref{eq:input_1}), in order to desirably
influence the dynamics of the species $X$ and showcase
the capabilities of Algorithm~\ref{alg:SM_algorithm}. 
Network~(\ref{eq:input_1}) may be interpreted as a simplified
model of genetic transcription, with $X$ representing an mRNA species, 
being transcribed and degraded, see also Section~\ref{sec:example3}
for a more-detailed model. 
Despite the simplicity of~(\ref{eq:input_1}), 
it serves as a test network for biochemical control theory, 
outlining some of the advantages and disadvantages a controller may have. 
For example, the controller put forward in~\cite{Khammash}, when
embedded into the test network~(\ref{eq:input_1}), is unable to guide
the stationary mean of the species $X$ below the value $\alpha_1/\alpha_2$.

\begin{figure}[!htbp]
\vskip -1.0 cm
\centerline{
\hskip 0mm
\includegraphics[width=0.35\columnwidth]{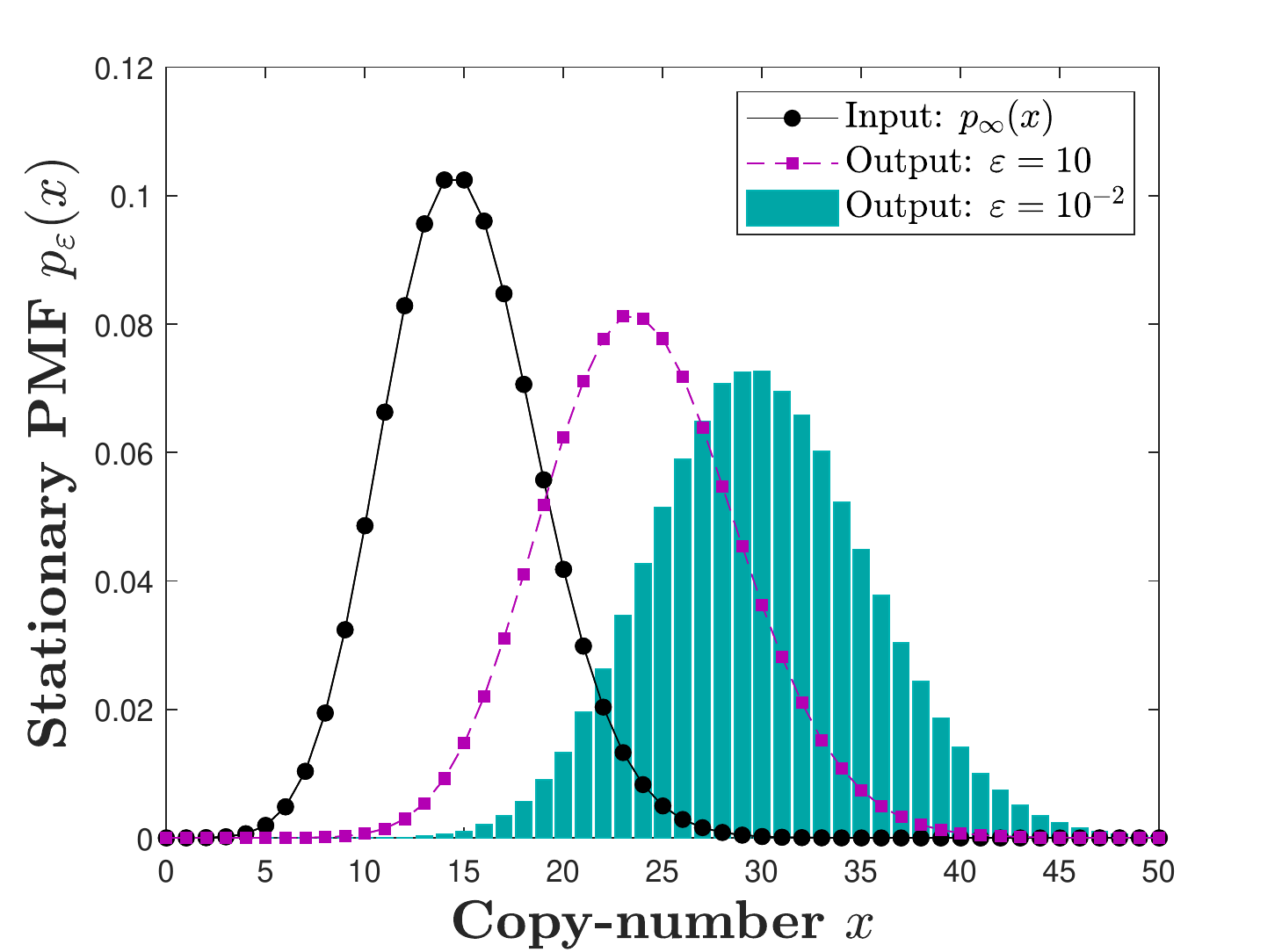}
\hskip 1.5cm
\includegraphics[width=0.35\columnwidth]{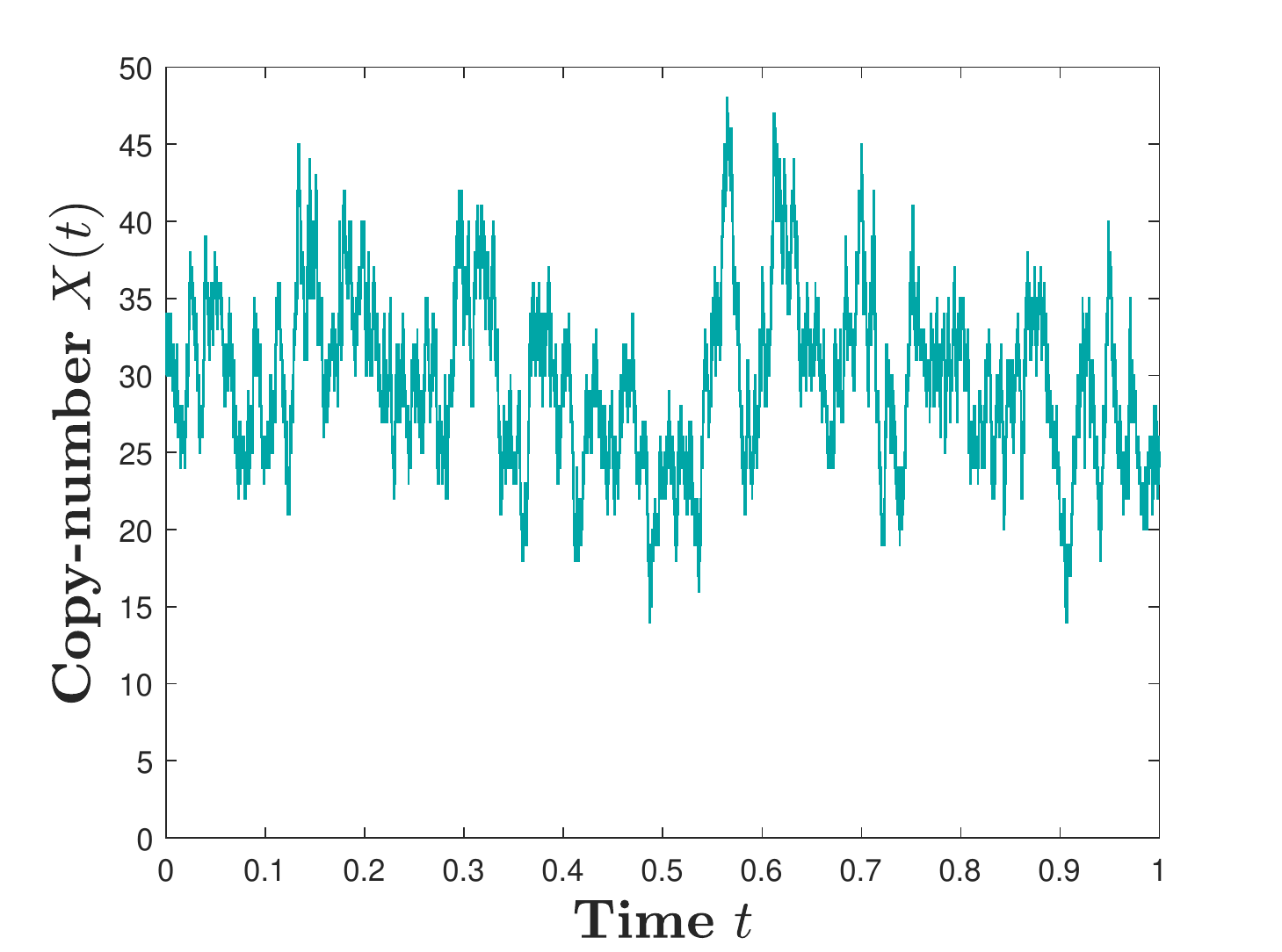}
}
\vskip -4.5cm
\leftline{\hskip 1.5cm (a) \hskip 6.7cm (b)} 
\vskip 4.0cm
\centerline{
\hskip 0mm
\includegraphics[width=0.35\columnwidth]{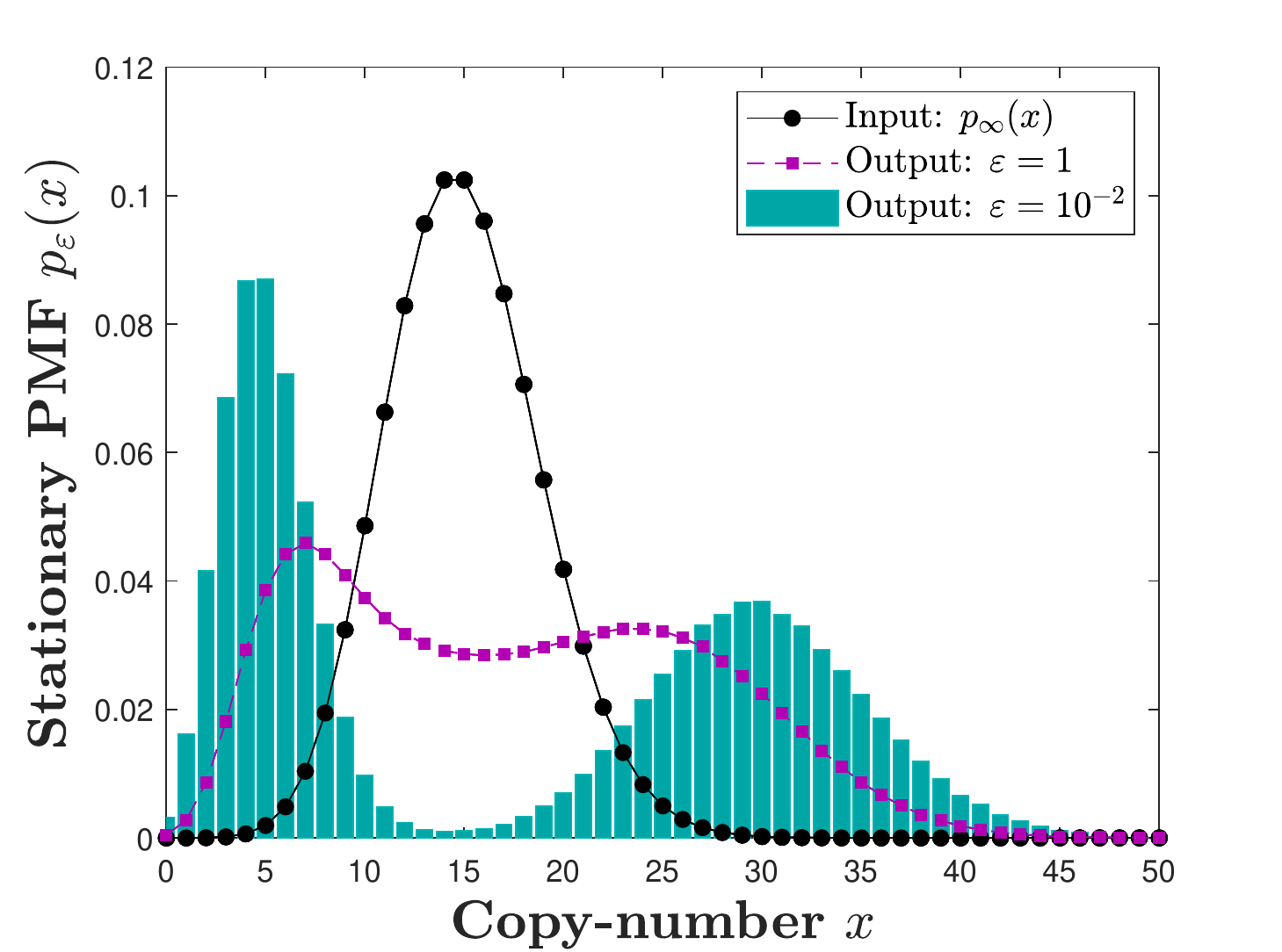}
\hskip 1.5cm
\includegraphics[width=0.35\columnwidth]{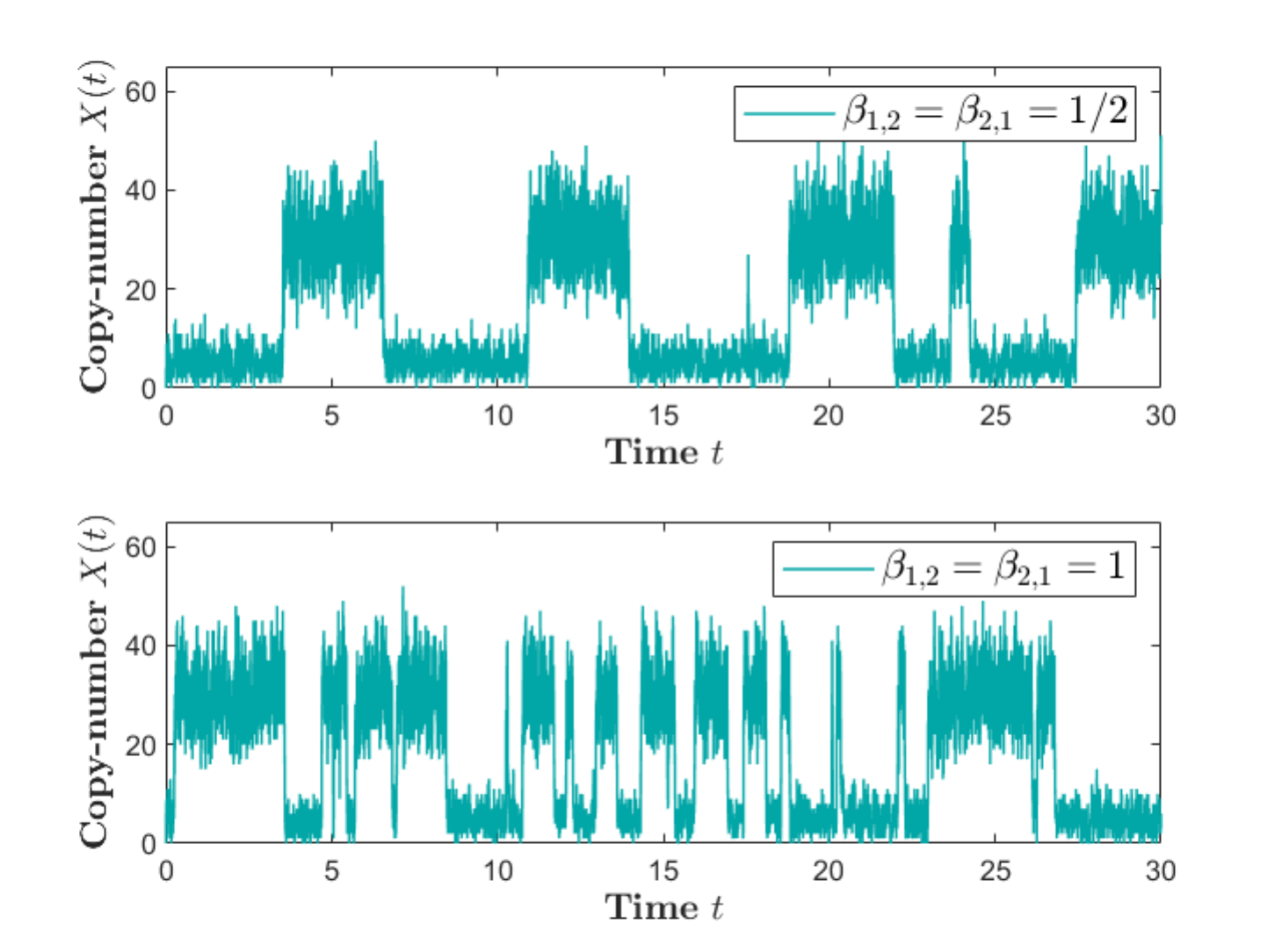}
}
\vskip -4.5cm
\leftline{\hskip 1.5cm (c) \hskip 6.7cm (d)} 
\vskip 4.0cm
\centerline{
\hskip 0mm
\includegraphics[width=0.35\columnwidth]{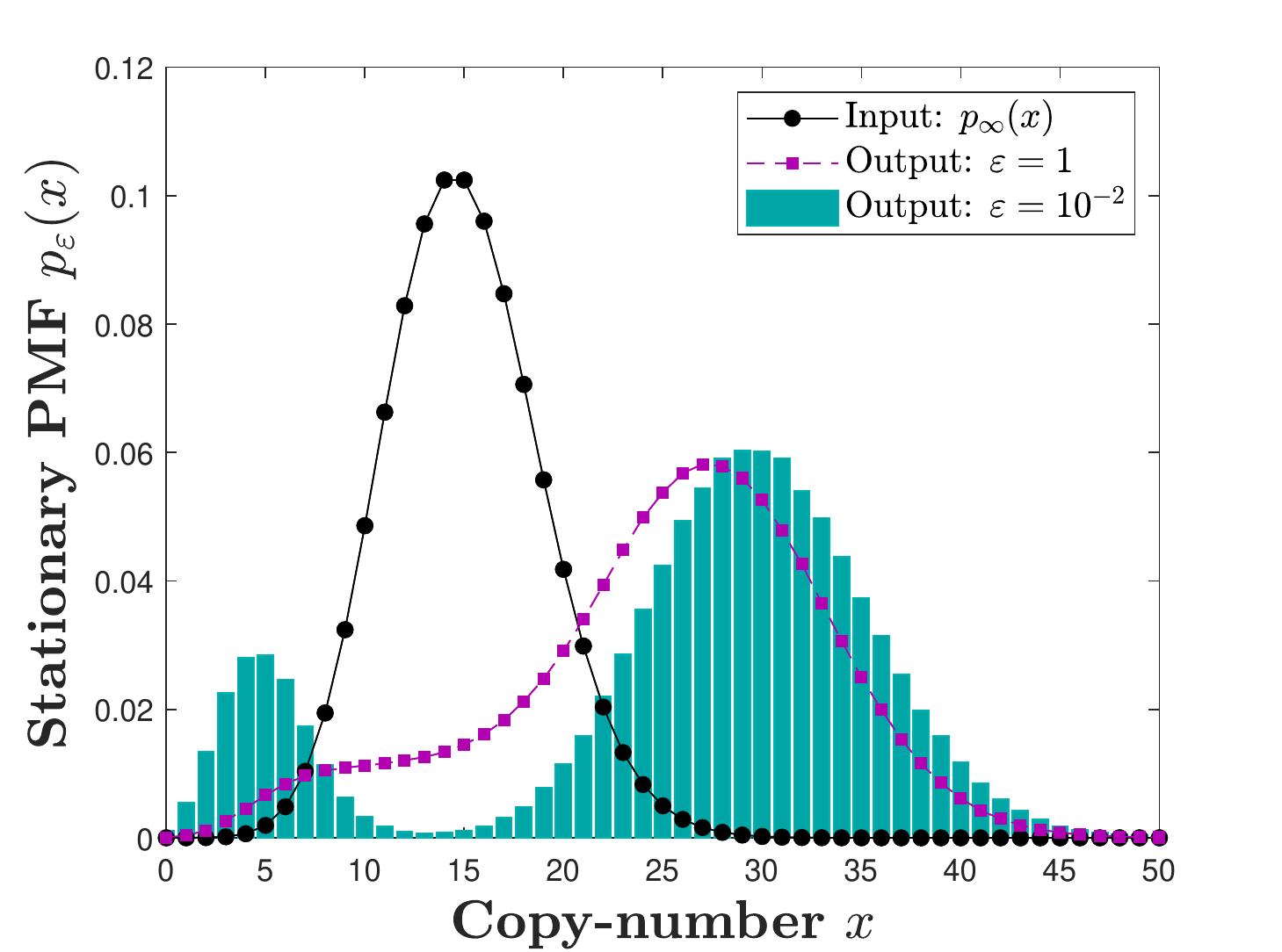}
\hskip 1.5cm
\includegraphics[width=0.35\columnwidth]{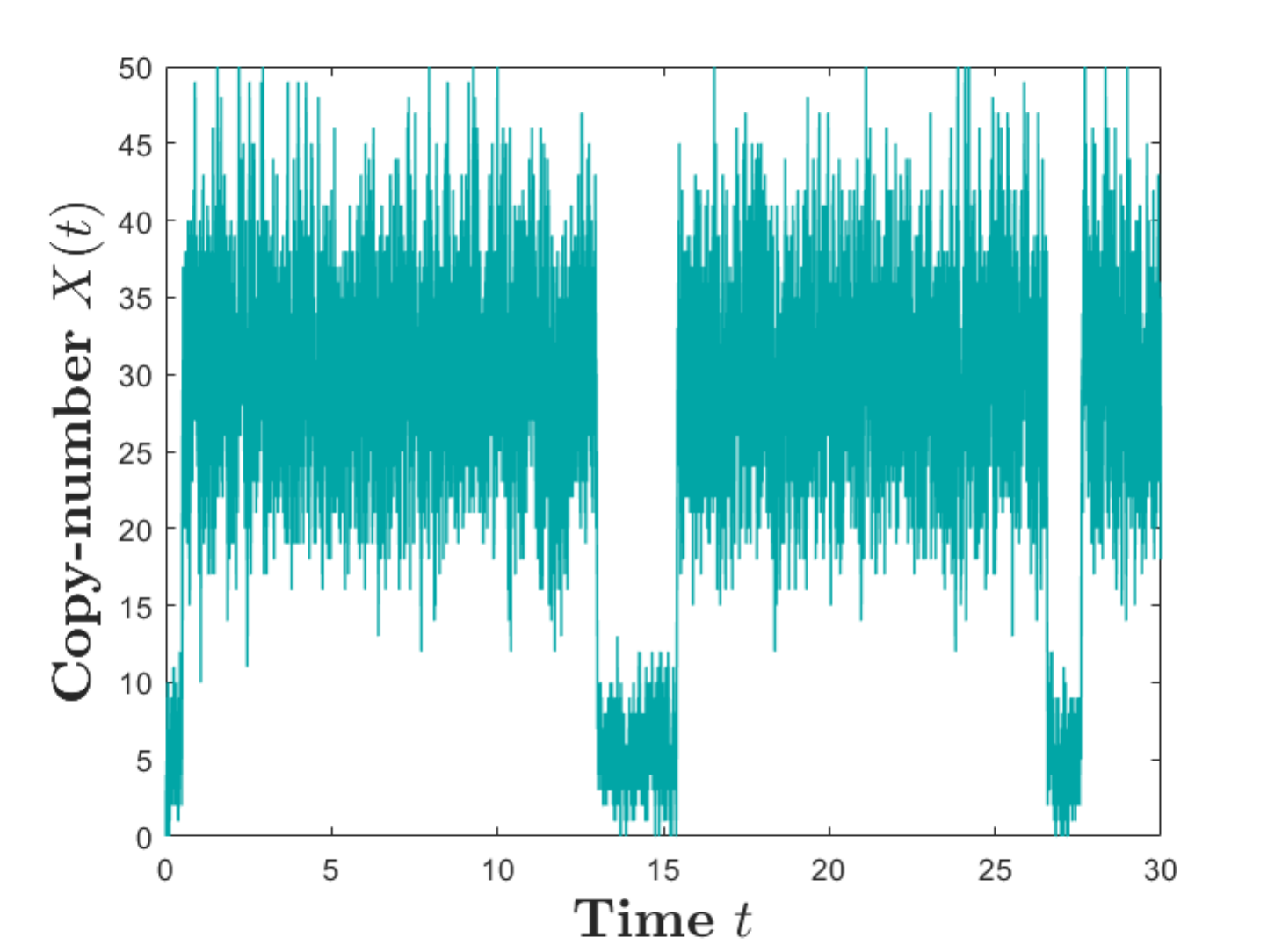}
}
\vskip -4.5cm
\leftline{\hskip 1.5cm (e) \hskip 6.7cm (f)} 
\vskip 4.0cm
\centerline{
\hskip 0mm
\includegraphics[width=0.35\columnwidth]{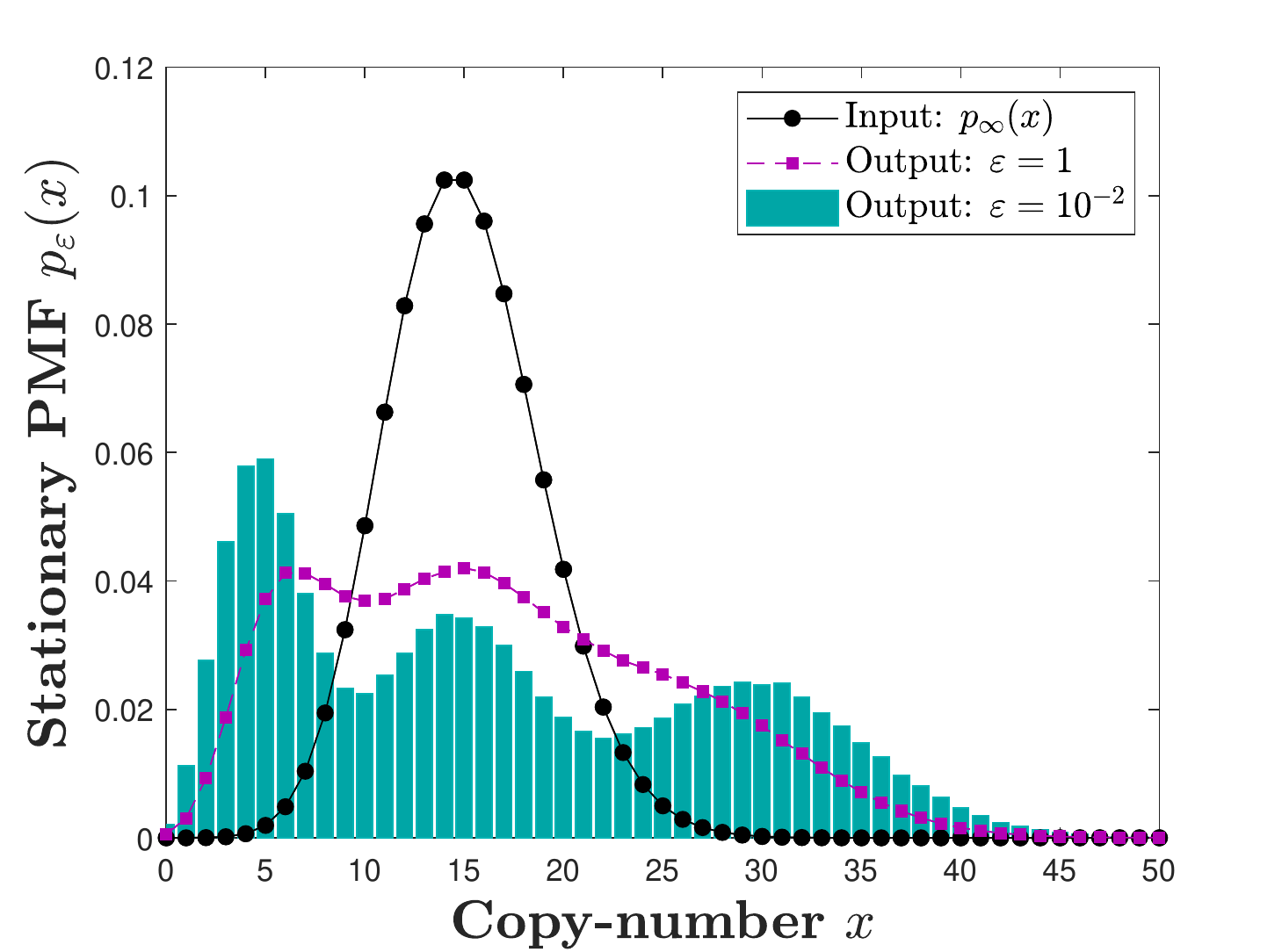}
\hskip 1.5cm
\includegraphics[width=0.35\columnwidth]{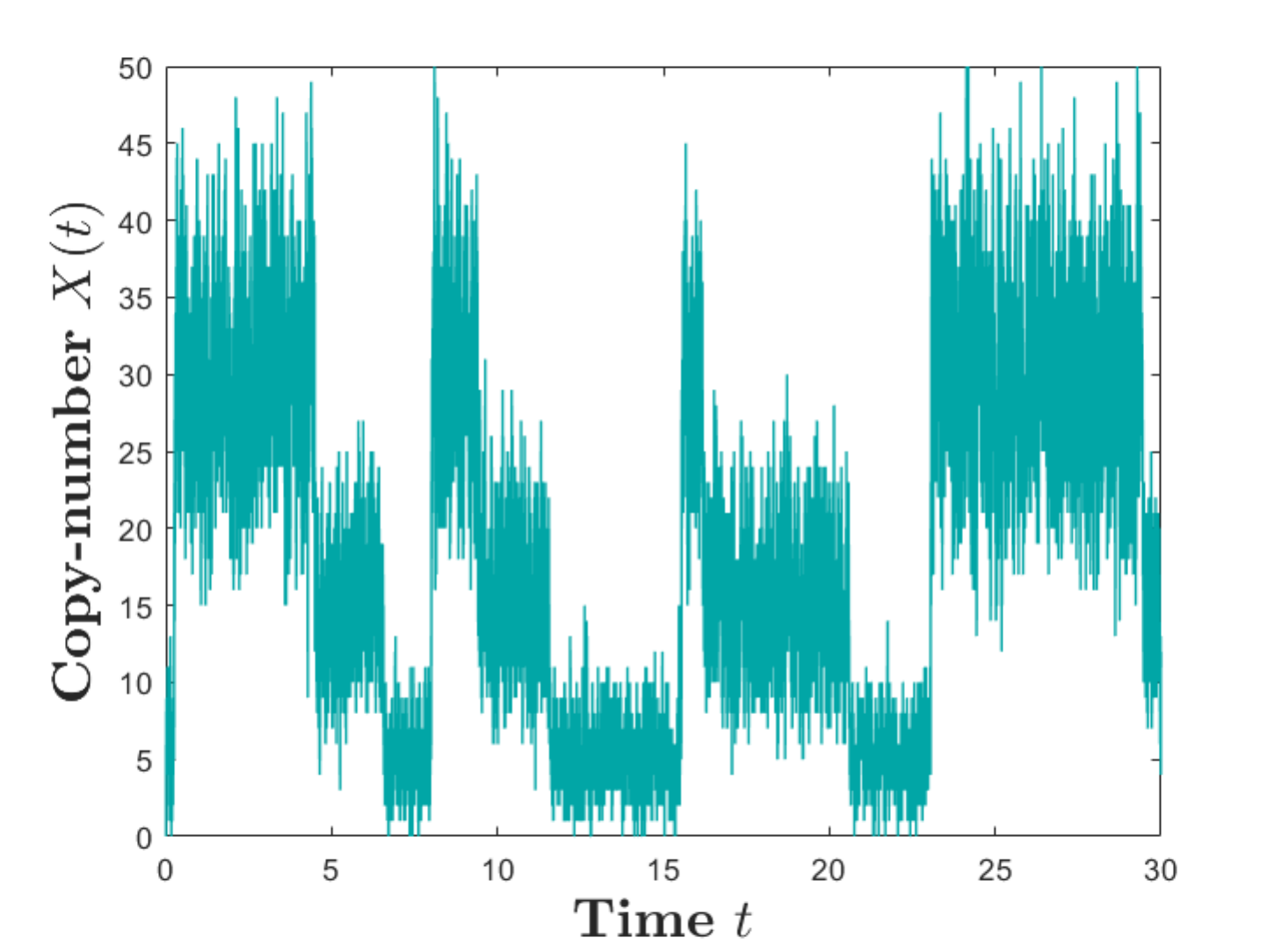}
}
\vskip -4.5cm
\leftline{\hskip 1.5cm (g) \hskip 6.7cm (h)} 
\vskip 3.8cm
\caption{ 
 Application of the lower-resolution control from 
{\rm Algorithm~\ref{alg:SM_algorithm}} on the input network~{\rm (\ref{eq:input_1})}
with $(\alpha_1, \alpha_2) = (1, 1/15)$. 
{\it The stationary {\rm PMF} of the input network is displayed
as the interpolated black dots. 
Panel {\rm (a)} shows the stationary $x$-marginal {\rm PMF}
of the output network~{\rm (\ref{eq:input_1})$\cup$(\ref{eq:R_gamma_Poisson_1})}, 
with $\beta_{1,1} = 1$ and $(\gamma_0, \gamma_1) = (1,30)$, 
obtained numerically for two different values of the asymptotic parameter $\varepsilon$, 
shown as the interpolated purple squares and the cyan histogram. 
Panel {\rm (b)} displays a representative sample path, obtained by applying the
Gillespie algorithm, corresponding to the cyan
histogram from panel {\rm (a)}.
Analogous plots are shown for the output 
network~{\rm (\ref{eq:input_1})$\cup$(\ref{eq:R_gamma_Poisson_2})}
with $(\gamma_0, \gamma_1,\gamma_2) = (1,5,30)$
and $(\beta_{1,1}, \beta_{1,2}, \beta_{2,1}) = (1, 1/2, 1/2)$ (as well as 
$(\beta_{1,1}, \beta_{1,2}, \beta_{2,1}) = (1, 1, 1)$) in panels {\rm (c)}--{\rm (d)}, 
while with $(\beta_{1,1}, \beta_{1,2}, \beta_{2,1}) = (1, 5/6, 1/6)$
in panels {\rm (e)}--{\rm (f)}.
Finally, panels {\rm (g)}--{\rm (h)} show the plots for the
output network~{\rm (\ref{eq:input_1})$\cup$(\ref{eq:R_gamma_Poisson_3})}
with $(\gamma_0, \gamma_1,\gamma_2,\gamma_3) = (1,5,30,15)$
and $(\beta_{1,1}, \beta_{1,2}, \beta_{2,3}, \beta_{3,1}) = (1, 1/3, 1/3,1/3)$. 
For simplicity, the sample paths have been generated with the controlling species
initially satisfying $\sum_{i = 1}^M Y_i(0) = 1$.
}
}
\label{fig:Poisson}
\end{figure}

\subsection{Lower-resolution control} \label{sec:low_resolution}
\emph{Uni-modality}.  
Consider the controller $\mathcal{R}_{\beta} \cup \mathcal{R}_{\gamma}^{\mathcal{P}} = 
\mathcal{R}_{\beta}(Y_1) \cup \mathcal{R}_{\gamma}^{\mathcal{P}}(X; \, Y_1)$, 
called a stochastic morpher, given by
\begin{align}
\mathcal{R}_{\beta}:
& & 2 Y_1 & \xrightarrow[]{\beta_{1,1}} Y_1, \nonumber \\
\mathcal{R}_{\gamma}^{\mathcal{P}}: 
& & \mathcal{R}_{\gamma_0}^{\varepsilon}: \; \; \; \; X & \xrightarrow[]{\gamma_{0}/\varepsilon} \varnothing, \nonumber \\
& & \mathcal{R}_{\gamma_1}^{\varepsilon}:  \; \; \; \; Y_1 & \xrightarrow[]{\gamma_{1}/\varepsilon} Y_1 + X, 
\; \; \; \; 0 < \varepsilon \ll 1,
\label{eq:R_gamma_Poisson_1}
\end{align}
where $X$ is the target species, while $Y_1$ is the controlling species. 
The controller~(\ref{eq:R_gamma_Poisson_1}) consists of two sub-networks:
$\mathcal{R}_{\beta}(Y_1)$, describing a bi-molecular
degradation of the controlling species $Y_1$, and
$\mathcal{R}_{\gamma}^{\mathcal{P}}(X; \, Y_1) = \mathcal{R}_{\gamma_0}^{\varepsilon}(X; \, \varnothing) 
\cup  \mathcal{R}_{\gamma_1}^{\varepsilon}(X; \, Y_1)$ 
describing a degradation of the target species $X$, and a 
production of $X$ catalyzed by $Y_1$. 
To emphasize the catalytic role of $Y_1$ in
$\mathcal{R}_{\gamma_1}^{\varepsilon}$, we write 
$\mathcal{R}_{\gamma_1}^{\varepsilon} = \mathcal{R}_{\gamma_1}^{\varepsilon}(X; \, Y_1)$ 
and, since $\mathcal{R}_{\gamma_0}^{\varepsilon}$ is not catalyzed by $Y_1$, 
we write $\mathcal{R}_{\gamma_0}^{\varepsilon} = 
\mathcal{R}_{\gamma_0}^{\varepsilon}(X; \, \varnothing)$. 
The super-script $\mathcal{P}$ appearing in $\mathcal{R}_{\gamma}^{\mathcal{P}}$ 
stands for the Poisson distribution, as motivated shortly. 
See also Figure~\ref{fig:Control_Theory}, and Appendices~\ref{app:biochemical_control} 
and~\ref{app:analysis}, for more details on the notation. 

In what follows, we analyze the output network 
$\mathcal{R}_{\alpha}^1 \cup \mathcal{R}_{\beta} 
\cup \mathcal{R}_{\gamma}^\mathcal{P}$, obtained by embedding the 
stochastic morpher~(\ref{eq:R_gamma_Poisson_1})
into the input network~(\ref{eq:input_1}), which we compactly denote
by~(\ref{eq:input_1})$\cup$(\ref{eq:R_gamma_Poisson_1}).
Assuming the copy-number of $Y_1$, denoted by $y_1 \in \mathbb{Z}_{\ge}$, is non-zero initially,
the sub-network $\mathcal{R}_{\beta}(Y_1)$ 
fires until the stationary value $y_1 = 1$ is reached.
On the other hand, the stationary marginal PMF of the target species $X$ from the output 
network~(\ref{eq:input_1})$\cup$(\ref{eq:R_gamma_Poisson_1}), denoted by
$p_{\varepsilon}(x)$, reads
$
p_{\varepsilon}(x) = \mathcal{P} \left(x; \, (\gamma_{1} + \varepsilon \alpha_1)/
(\gamma_{0} + \varepsilon \alpha_2) \right),
$
from which it follows that
\begin{align}
p_{\varepsilon}(x) =
\begin{cases}
\mathcal{P} \left(x; \, \frac{\alpha_1}{\alpha_2} \right), 
\; \; \textrm{as } \varepsilon \to \infty, \\[0.3cm]
\mathcal{P} \left(x; \, \frac{\gamma_{1}}{\gamma_{0}} \right), 
\; \; \textrm{as } \varepsilon \to 0.
\label{eq:R_gamma_Poisson_1_PMF}
\end{cases}
\end{align}
In words, as the sub-network $\mathcal{R}_{\gamma}^\mathcal{P}$ from the 
stochastic morpher $\mathcal{R}_{\beta} \cup \mathcal{R}_{\gamma}^\mathcal{P}$, 
given by~(\ref{eq:R_gamma_Poisson_1}),
fires faster, the input network $\mathcal{R}_{\alpha}^1$, 
given by~(\ref{eq:input_1}), is over-ridden, and
the stationary $x$-marginal PMF of the corresponding output network
$\mathcal{R}_{\alpha}^1 \cup \mathcal{R}_{\beta} \cup \mathcal{R}_{\gamma}^\mathcal{P}$
 is gradually transformed (morphed) from 
the Poisson PMF centered at $x = \alpha_1/\alpha_2$ to
the Poisson PMF centered at $x = \gamma_{1}/\gamma_{0}$.
Note that such a uni-modal morphing also controls
the first-moment (mean) of the output network. 
This is numerically confirmed in
 Figure~\ref{fig:Poisson}(a), where we display the stationary  $x$-marginal PMFs
of the output network~(\ref{eq:input_1})$\cup$(\ref{eq:R_gamma_Poisson_1}) for different values of $\varepsilon$, 
with the coefficients from 
$\mathcal{R}_{\beta}(Y_1)$ and $\mathcal{R}_{\gamma}^{\mathcal{P}}(X; \, Y_1)$
fixed to $\beta_{1,1} = 1$ and $\boldsymbol{\gamma}= (\gamma_{0}, \gamma_{1}) = (1,30)$, respectively. 
In particular, the stationary PMF is a Poisson distribution centered at $x = 24$ when $\varepsilon = 10$, shown 
as the purple squares, which, in accordance with~(\ref{eq:R_gamma_Poisson_1_PMF}), 
converges close to the Poisson distribution
centered at $x = \gamma_{1}/\gamma_{0} = 30$ when $\varepsilon = 10^{-2}$, 
shown as the cyan histogram in Figure~\ref{fig:Poisson}(a). 
A representative sample path, corresponding
to the histogram from Figure~\ref{fig:Poisson}(a), is displayed in Figure~\ref{fig:Poisson}(b), over a 
relatively short time-interval, allowing for the time-scale of the underlying fluctuations
around the mean to be more readily visually discernible.

We say that the stochastic morpher~(\ref{eq:R_gamma_Poisson_1}), 
when embedded into~(\ref{eq:input_1}), is a \emph{robust} controller, 
due to the fact that the stationary $x$-marginal PMF of the resulting output network, 
given by~(\ref{eq:R_gamma_Poisson_1_PMF}), 
satisfies the following two properties:
(a) it is independent of the initial conditions for $X$ and non-zero initial conditions for $Y_1$,  
and
(b) it is independent of the 
rate coefficients $\boldsymbol{\alpha}$ from the input network 
in the limit $\varepsilon \to 0$. 
See also Definition~\ref{app:biochemical_control}
from Appendix~\ref{app:biochemical_control}, 
as well as Theorem~\ref{theorem:Poisson_Kronecker} from 
Appendix~\ref{app:long_time} for
a more general result.
Note that, without the reaction $2 Y_1 \to Y_1$ from~(\ref{eq:R_gamma_Poisson_1}),
condition (a) would be violated, as
in this case the stationary $x$-marginal PMF would depend 
on the initial condition for $Y_1$. However, under suitable experimental implementations of
the stochastic morpher inside cell-like vesicles, 
one can achieve exactly one copy-number of $Y_1$ initially inside the vesicles ~\cite{Vesicles5},
analogous to having one copy-number of a gene inside a living cell,
hence eliminating the need for the reaction $2 Y_1 \to Y_1$, 
see also Section~\ref{sec:experiments}.
Note also that property (b) leads to asymptotic robust
perfect adaptation for \emph{all} of the underlying stationary statistics
of the target species $X$. 

The output network~{\rm (\ref{eq:input_1})$\cup$(\ref{eq:R_gamma_Poisson_1})}
is obtained as a particular application of
Algorithm~\ref{alg:SM_algorithm} on the input network~(\ref{eq:input_1}), 
with $\mathcal{R}_{\beta}(Y_1)$ and 
$\mathcal{R}_{\gamma}^\mathcal{P}(X; \, Y_1)$ obtained
by taking $M = 1$ controlling species $\mathcal{Y} = \{Y_1\}$, and $N = 1$ target species 
$\mathcal{X}_{\tau} = \{X \equiv X_1\}$, in~(\ref{eq:R_beta_alg})--(\ref{eq:Poisson_control_alg}).
Before proceeding to further applications of Algorithm~\ref{alg:SM_algorithm}, 
let us note that the degradation reaction $X \to\varnothing$,
introduced by the controller~(\ref{eq:R_gamma_Poisson_1}),
may be seen as an approximation of the reaction 
$Y_0 + X \to Y_0$, where $Y_0$ is a suitable additional 
controller (buffer) species, assumed to be maintained at a constant copy-number. 
One may also replace $X \to\varnothing$ from~(\ref{eq:R_gamma_Poisson_1}) with 
(a possibly experimentally less elegant, see Section~\ref{sec:experiments})
reaction $Y_1 + X \to Y_1$, without changing the conclusions made in this section.
More generally, in this paper, reactions present inside controllers, which 
depend explicitly only on the target species $\mathcal{X}_{\tau}$, are assumed to 
implicitly depend on suitable additional auxiliary (buffer) species, 
see also Appendix~\ref{app:analysis} for a further discussion.
Such simplifications have been employed for the purpose
of exposition, and do not limit experimental implementations of the stochastic morpher.
In fact, an introduction of suitable buffer species is a critical step in
 experimentally realizing biochemical reaction networks~\cite{DNAComputing1}.

\emph{Bi-modality}. 
Let us now apply Algorithm~\ref{alg:SM_algorithm} on the input network~(\ref{eq:input_1}),
with the stochastic morpher $\mathcal{R}_{\beta} \cup \mathcal{R}_{\gamma}^\mathcal{P} = 
\mathcal{R}_{\beta}(Y_1, Y_2) \cup \mathcal{R}_{\gamma}^\mathcal{P}(X; \, Y_1, Y_2)$, 
given by
\begin{align}
\mathcal{R}_{\beta}:
& & 2 Y_1 & \xrightarrow[]{\beta_{1,1}} Y_1 \xrightleftharpoons[\beta_{2,1}]{\beta_{1,2}} Y_2,
\nonumber \\
\mathcal{R}_{\gamma}^{\mathcal{P}}: 
& & \mathcal{R}_{\gamma_0}^{\varepsilon}:  \; \; \; \; X & \xrightarrow[]{\gamma_{0}/\varepsilon} \varnothing, \nonumber \\
& & \mathcal{R}_{\gamma_1}^{\varepsilon}:  \; \; \; \; Y_1 & \xrightarrow[]{\gamma_{1}/\varepsilon} Y_1 + X, \nonumber \\
& & \mathcal{R}_{\gamma_2}^{\varepsilon}:  \; \; \; \;  Y_2 & \xrightarrow[]{\gamma_{2}/\varepsilon} Y_2 + X, 
\; \; \; \; 0 < \varepsilon \ll 1.
\label{eq:R_gamma_Poisson_2}
\end{align}
Here, the sub-network $\mathcal{R}_{\beta}(Y_1,Y_2)$
describes first-order conversion between the two controlling species $Y_1$ and $Y_2$, 
with the reaction $2 Y_1 \to Y_1$ ensuring that the species $Y_1$ and $Y_2$
 satisfy the conservation law $(y_1 + y_2) = 1$ in the long-run,
independent of the non-zero initial conditions, 
i.e. the stationary (long-time) state-space is given by $(y_1,y_2) \in \{ (1,0), (0,1) \}$.
On the other hand, the sub-network $\mathcal{R}_{\gamma}^\mathcal{P}(X; \, Y_1, Y_2)$
 involves two production reactions for the species $X$, one catalyzed by
$Y_1$ and the other by $Y_2$. Ignoring the reaction $2 Y_1 \to Y_1$, 
note that~(\ref{eq:R_gamma_Poisson_2}) may be interpreted
as describing a gene, which switches between two different states $Y_1$ and $Y_2$, 
and produces an mRNA species $X$ at different rates, 
depending on the gene state~\cite{Kepler,Me_Mixing}.

When  $(y_1,y_2) = (1,0)$, reaction $\mathcal{R}_{\gamma_2}^{\varepsilon}$
from the sub-network $\mathcal{R}_{\gamma}^\mathcal{P}(X; \, Y_1, Y_2)$ cannot fire, 
and the remaining faster reactions
$\mathcal{R}_{\gamma_0}^{\varepsilon} \cup \mathcal{R}_{\gamma_1}^{\varepsilon}$ 
generate the Poisson PMF centered at $x = \gamma_{1}/\gamma_{0}$, 
while when $(y_1,y_2) = (0,1)$, the reaction $\mathcal{R}_{\gamma_1}^{\varepsilon}$
is switched off, and the active faster reactions 
$\mathcal{R}_{\gamma_0}^{\varepsilon} \cup \mathcal{R}_{\gamma_2}^{\varepsilon}$ 
 induce the Poisson PMF centered at $x = \gamma_{2}/\gamma_{0}$.
As the controlling species $Y_1$ and $Y_2$ convert between themselves, 
they mix the two Poisson PMFs from the faster network
$\mathcal{R}_{\gamma}^\mathcal{P}(X; \, Y_1, Y_2)$, 
which over-ride the PMF of the input network $\mathcal{R}_{\alpha}^1(X)$. 
In the limit $\varepsilon \to 0$, the resulting stationary $x$-marginal PMF of the 
output network~(\ref{eq:input_1})$\cup$(\ref{eq:R_gamma_Poisson_2}) 
is given by
\begin{align}
p_{0}(x) & = \left(1 + \frac{\beta_{1,2}}{\beta_{2,1}} \right)^{-1}
\mathcal{P} \left(x; \,  \frac{\gamma_{1}}{\gamma_{0}} \right) 
+ \left(1 + \frac{\beta_{2,1}}{\beta_{1,2}} \right)^{-1}
\mathcal{P} \left(x; \,  \frac{\gamma_{2}}{\gamma_{0}} \right), 
\label{eq:R_gamma_Poisson_2_PMF}
\end{align}
see also Theorem~\ref{theorem:Poisson_Kronecker} in Appendix~\ref{app:analysis}
for a general result.
Therefore, the controller~(\ref{eq:R_gamma_Poisson_2}) allows one to 
morph the input PMF into a bi-modal output one, which is a linear combination 
of two Poisson distributions, whose modes are controlled
with the rate coefficients $\boldsymbol{\gamma}$ from 
the faster sub-network $\mathcal{R}_{\gamma}^\mathcal{P}(X; \, Y_1, Y_2)$, 
while the PMF values at the modes (weights in~(\ref{eq:R_gamma_Poisson_2_PMF})) 
are determined by the rate coefficients $\boldsymbol{\beta}$ from 
the slower sub-network $\mathcal{R}_{\beta}(Y_1,Y_2)$. 

More precisely, the stationary PMF~(\ref{eq:R_gamma_Poisson_2_PMF}) 
is independent of the asymptotic parameter $\varepsilon$,
 and depends on $\beta_{1,2}$ and $\beta_{2,1}$ only via the ratio $\beta_{1,2}/\beta_{2,1}$,
which determines the PMF values at the two modes,
which in turn depend on the ratios $\gamma_1/\gamma_0$
and $\gamma_2/\gamma_0$.
However, note that the underlying sample paths do depend on
 $\varepsilon$, which determines the
 time-scale of the fluctuations near each of the two modes. 
Furthermore, the parameters $\beta_{1,2}$ and $\beta_{2,1}$ 
influence the sample paths independently, and not only via their ratio. 
In particular, for a fixed ratio $\beta_{1,2}/\beta_{2,1}$,
 the value of $\beta_{1,2}$ determines the time-scale of stochastic 
switching between the two modes.
More precisely, the time spent near $\gamma_1/\gamma_0$, 
given the system has started near the corresponding mode, 
is an exponentially distributed random variable 
with mean $1/\beta_{1,2}$,
after which the system deterministically moves to a neighborhood of
$\gamma_2/\gamma_0$, and vice-versa for the other mode.
These observations are instances of the fact that PMFs do not
uniquely capture time-parametrizations of the underlying sample paths, 
and, if desired, can be exploited for gaining further biochemical control.
In particular, one may control 
the stationary marginal-PMF of the target species via appropriate ratios of the underlying rate coefficients
(weak control). Furthermore, some of the properties of the underlying sample paths 
may also be controlled via a more-detailed fine-tuning of (the order of magnitude of)
 the rate coefficients (strong control).
Given a fixed ratio $\beta_{1,2}/\beta_{2,1}$, 
the precise values of the coefficients may be fixed
with the constraint $\beta_{1,2} + \beta_{2,1} = c > 0$.
In what follows, when we do not wish to explicitly control the mean switching time, 
we arbitrarily set $c = 1$.

Let us fix two modes of the output network~(\ref{eq:input_1})$\cup$(\ref{eq:R_gamma_Poisson_2})
 to $x = \gamma_1/\gamma_0 = 5$ and $x = \gamma_1/\gamma_0 = 30$, 
which may be achieved by choosing $\boldsymbol{\gamma} = 
 (\gamma_{0}, \gamma_{1}, \gamma_2) = (1,5,30)$. 
In Figure~\ref{fig:Poisson}(c), we display the corresponding stationary
$x$-marginal PMFs for the 
output network for different values of $\varepsilon$,
with $\boldsymbol{\beta} = (\beta_{1,1}, \beta_{1,2}, \beta_{2,1})
= (1, 1/2, 1/2)$ chosen so that the two Poisson distributions 
from~(\ref{eq:R_gamma_Poisson_2_PMF}) have equal weights,
i.e. we take $\beta_{1,2}/\beta_{2,1} = 1$.  
One can notice that the uni-modal input PMF 
is morphed into the bi-modal
output one, as predicted by~(\ref{eq:R_gamma_Poisson_2_PMF}), 
shown as the cyan histogram in Figure~\ref{fig:Poisson}(c).
A corresponding representative sample path is shown in 
the top sub-panel of Figure~\ref{fig:Poisson}(d), over a suitable time-interval, 
where the time-scale of the noise-induced switching between the two modes
is observable, with the mean time spent near each of the two modes
 given by $1/\beta_{1,2} = 1/\beta_{2,1} = 2$ time-units. 
Note that the time-scale of the fluctuations near 
each of the modes (determined by the parameter $\varepsilon$) matches
the one shown magnified in Figure~\ref{fig:Poisson}(b).
Also shown, in the bottom sub-panel of Figure~\ref{fig:Poisson}(d), 
is a sample path, over the same time-interval as in the top sub-panel, 
when $\boldsymbol{\beta} = (\beta_{1,1}, \beta_{1,2}, \beta_{2,1})
= (1, 1, 1)$, which also corresponds to the stationary PMF shown as the
 histogram in Figure~\ref{fig:Poisson}(c), but whose 
mean switching time is halved, $1/\beta_{1,2} = 1/\beta_{2,1} = 1$. 
More generally, instead of balancing the two Poisson PMFs
by choosing $\beta_{1,2}/\beta_{2,1} = 1$, as in 
Figure~\ref{fig:Poisson}(c)--(d), one may control the 
weights of each of the two
Poisson PMFs from~(\ref{eq:R_gamma_Poisson_2_PMF})
in a number of desirable ways. 
For example, in Figure~\ref{fig:Poisson}(e)--(f), we 
set $\beta_{1,2}/\beta_{2,1} = 2 \mathcal{P}(\gamma_1/\gamma_0; \, \gamma_1/\gamma_0)/
\mathcal{P}(\gamma_2/\gamma_0; \, \gamma_2/\gamma_0) \approx 5$, 
ensuring that the value of the stationary PMF
at the mode $x = \gamma_2/\gamma_0 = 30$ 
is approximately twice the value at the mode $x = \gamma_1/\gamma_0 = 5$, 
which may be achieved by taking $\boldsymbol{\beta} = 
(\beta_{1,1}, \beta_{1,2}, \beta_{2,1}) = (1, 5/6, 1/6)$.
Note that the intermediate PMFs from 
Figures~\ref{fig:Poisson}(c) and (e), shown as the interpolated purple squares, 
and obtained when $\varepsilon  = 1$,  
still partially achieve the goal of the control, demonstrating that the 
stochastic morpher may be useful even when not firing much faster than
the input network.

\emph{Tri-modality}. 
Algorithm~\ref{alg:SM_algorithm} may be utilized
to achieve multi-modality beyond bi-modality at the PMF level, 
and multi-stability and a controlled switching pattern at the underlying sample path level.
For example, let us morph the stationary PMF of
the input network~(\ref{eq:input_1}) into a tri-modal one,
with the modes $x \in \{5, 15, 30\}$. Furthermore, let
the underlying sample paths spend on average $3$ time-units
in the neighborhood of each of the modes, with the
switching order $5 \to 30 \to 15$, i.e.
after being close to the mode $x = 5$, the system
should jump near $x = 30$, then close to $x = 15$
and, finally, return back to $x = 5$. To this end,
consider embedding into the input network $\mathcal{R}_{\alpha}^1(X)$ 
the stochastic morpher
$\mathcal{R}_{\beta} \cup \mathcal{R}_{\gamma}^\mathcal{P} = 
\mathcal{R}_{\beta}(Y_1, Y_2,Y_3) \cup \mathcal{R}_{\gamma}^{\varepsilon}(X; \, Y_1, Y_2,Y_3)$,
given by
\begin{align}
\mathcal{R}_{\beta}:
& & 2 Y_1 & \xrightarrow[]{\beta_{1,1}} 
Y_1 \xrightarrow[]{\beta_{1,2}} 
Y_2 \xrightarrow[]{\beta_{2,3}}
Y_3 \xrightarrow[]{\beta_{3,1}} Y_1, 
\nonumber \\
\mathcal{R}_{\gamma}^{\mathcal{P}}: 
& & \mathcal{R}_{\gamma_0}^{\varepsilon}:  \; \; \; \; X & \xrightarrow[]{\gamma_{0}/\varepsilon} \varnothing, \nonumber \\
& & \mathcal{R}_{\gamma_1}^{\varepsilon}:  \; \; \; \; Y_1 & \xrightarrow[]{\gamma_{1}/\varepsilon} Y_1 + X, \nonumber \\
& & \mathcal{R}_{\gamma_2}^{\varepsilon}:  \; \; \; \;  Y_2 & \xrightarrow[]{\gamma_{2}/\varepsilon} Y_2 + X, \nonumber \\
& & \mathcal{R}_{\gamma_3}^{\varepsilon}:  \; \; \; \;  Y_3 & \xrightarrow[]{\gamma_{3}/\varepsilon} Y_3 + X, 
\; \; \; \; 0 < \varepsilon \ll 1.
\label{eq:R_gamma_Poisson_3}
\end{align}
Analogous to Figure~\ref{fig:Poisson}(a)--(f), in Figure~\ref{fig:Poisson}(g)--(h)
we display the stationary $x$-marginal PMF, and a representative sample path, of the 
output network~(\ref{eq:input_1})$\cup$(\ref{eq:R_gamma_Poisson_3}), 
with $\boldsymbol{\gamma} = (\gamma_0, \gamma_1, \gamma_2, \gamma_3)
= (1, 5, 30, 15)$ and $\boldsymbol{\beta} = (\beta_0, \beta_{1,2}, \beta_{2,3}, \beta_{3,1})
= (1, 1/3, 1/3, 1/3)$. In the asymptotic limit $\varepsilon \to 0$,
 the stationary PMF is a linear combination of the three Poisson
distributions centered at $x = \gamma_1/\gamma_0 = 5$, 
$x = \gamma_2/\gamma_0 = 30$ and $x = \gamma_3/\gamma_0 = 15$, 
each with equal weights (see also Theorem~\ref{theorem:Poisson_Kronecker} 
in Appendix~\ref{app:analysis}), which is in excellent agreement with the histogram 
from Figure~\ref{fig:Poisson}(g), where the asymptotic parameter $\varepsilon$ is two orders of magnitude
larger than the rate coefficients from the networks $\mathcal{R}_{\alpha}^1(X)$ and 
$\mathcal{R}_{\beta}(Y_1, Y_2,Y_3)$. Note that the switching order of the sample path from
 Figure~\ref{fig:Poisson}(h) mirrors the conversion
$Y_1 \to Y_2 \to Y_3 \to Y_1$ from~(\ref{eq:R_gamma_Poisson_3}).

\subsection{Higher-resolution control}
In Section~\ref{sec:low_resolution}, we have applied the lower-resolution
control from Algorithm~\ref{alg:SM_algorithm}, which consists of
the networks $\mathcal{R}_{\beta}$ and $\mathcal{R}_{\gamma}^\mathcal{P}$
given by~(\ref{eq:R_beta_alg}) and~(\ref{eq:Poisson_control_alg}), respectively, 
and may be used to achieve multi-modality/multi-stability. In this section, 
we replace the uni-molecular lower-resolution (Poisson-based) 
control network $\mathcal{R}_{\gamma}^\mathcal{P}$
with its bi-molecular higher-resolution (Kronecker-delta-based) 
counterpart $\mathcal{R}_{\gamma}^{\delta}$, given 
by~(\ref{eq:Kronecker_control_alg}) in Algorithm~\ref{alg:SM_algorithm}, 
which may be used to morph input PMFs to arbitrary probability distributions on
bounded state-spaces. 

\begin{figure}[!htbp]
\centerline{
\hskip 0mm
\includegraphics[width=0.35\columnwidth]{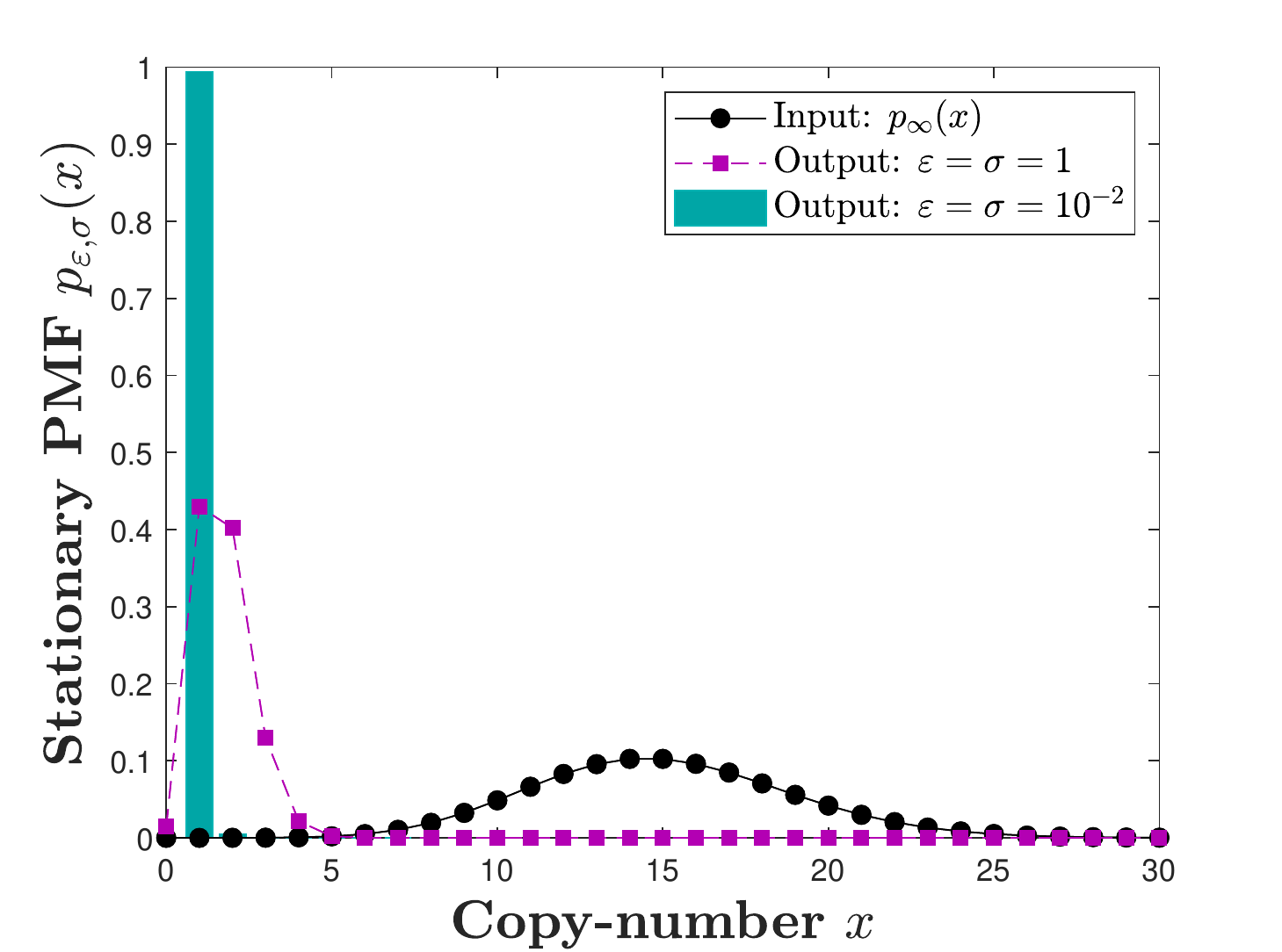}
\hskip 1.5cm
\includegraphics[width=0.35\columnwidth]{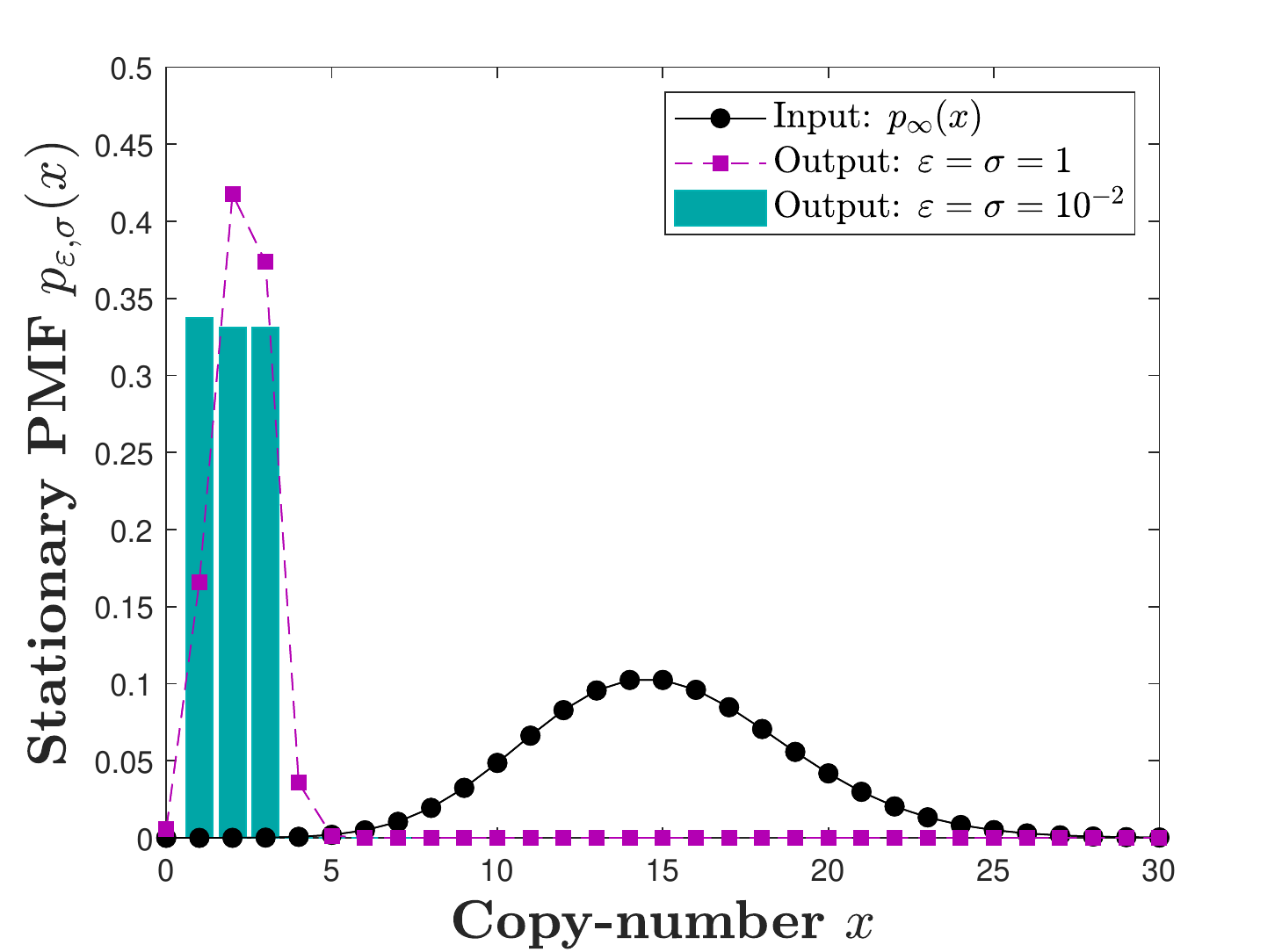}
}
\vskip -4.5cm
\leftline{\hskip 1.5cm (a) \hskip 6.9cm (b)} 
\vskip 4.0cm
\centerline{
\hskip 0mm
\includegraphics[width=0.35\columnwidth]{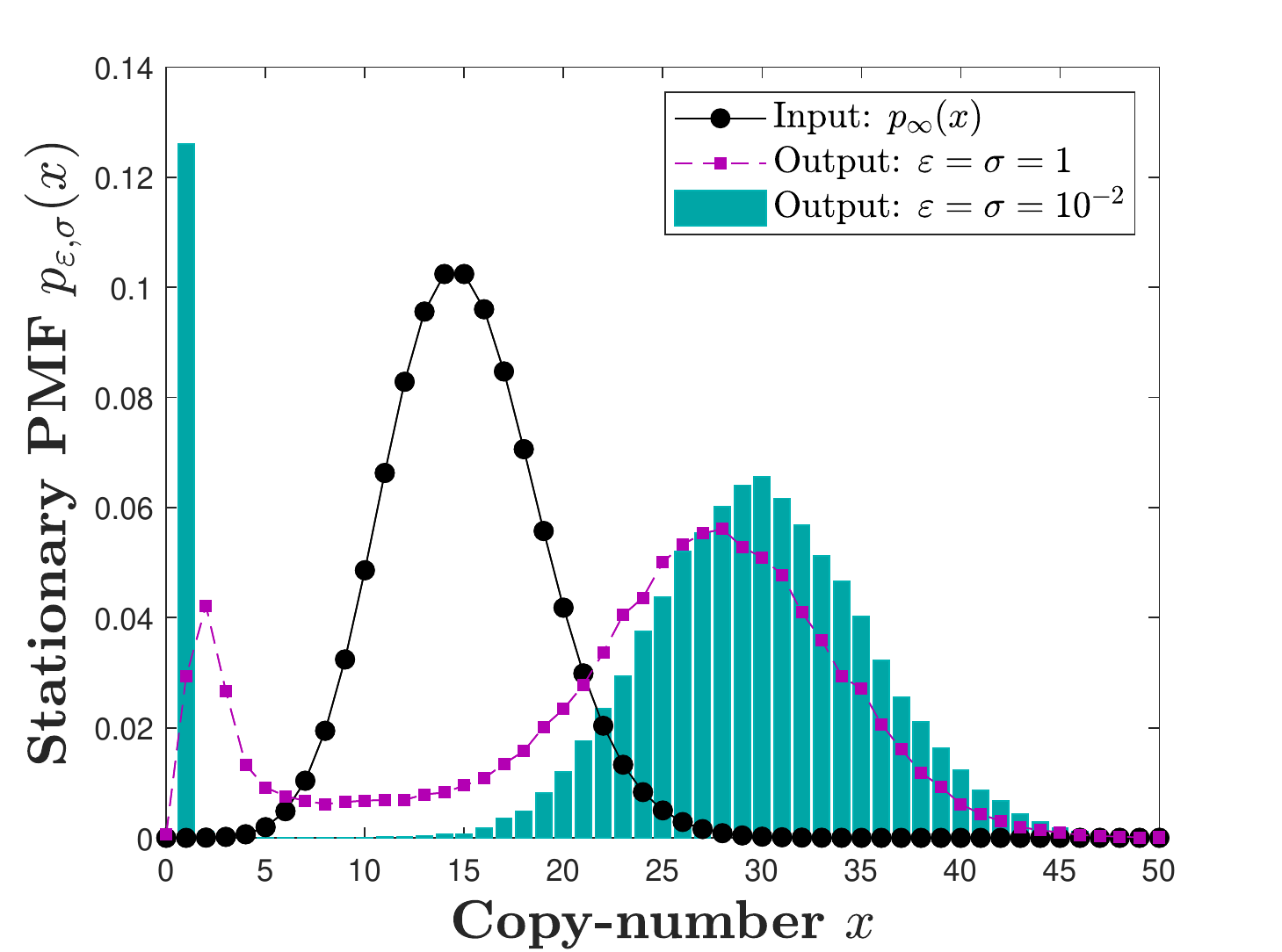}
\hskip 1.5cm
\includegraphics[width=0.35\columnwidth]{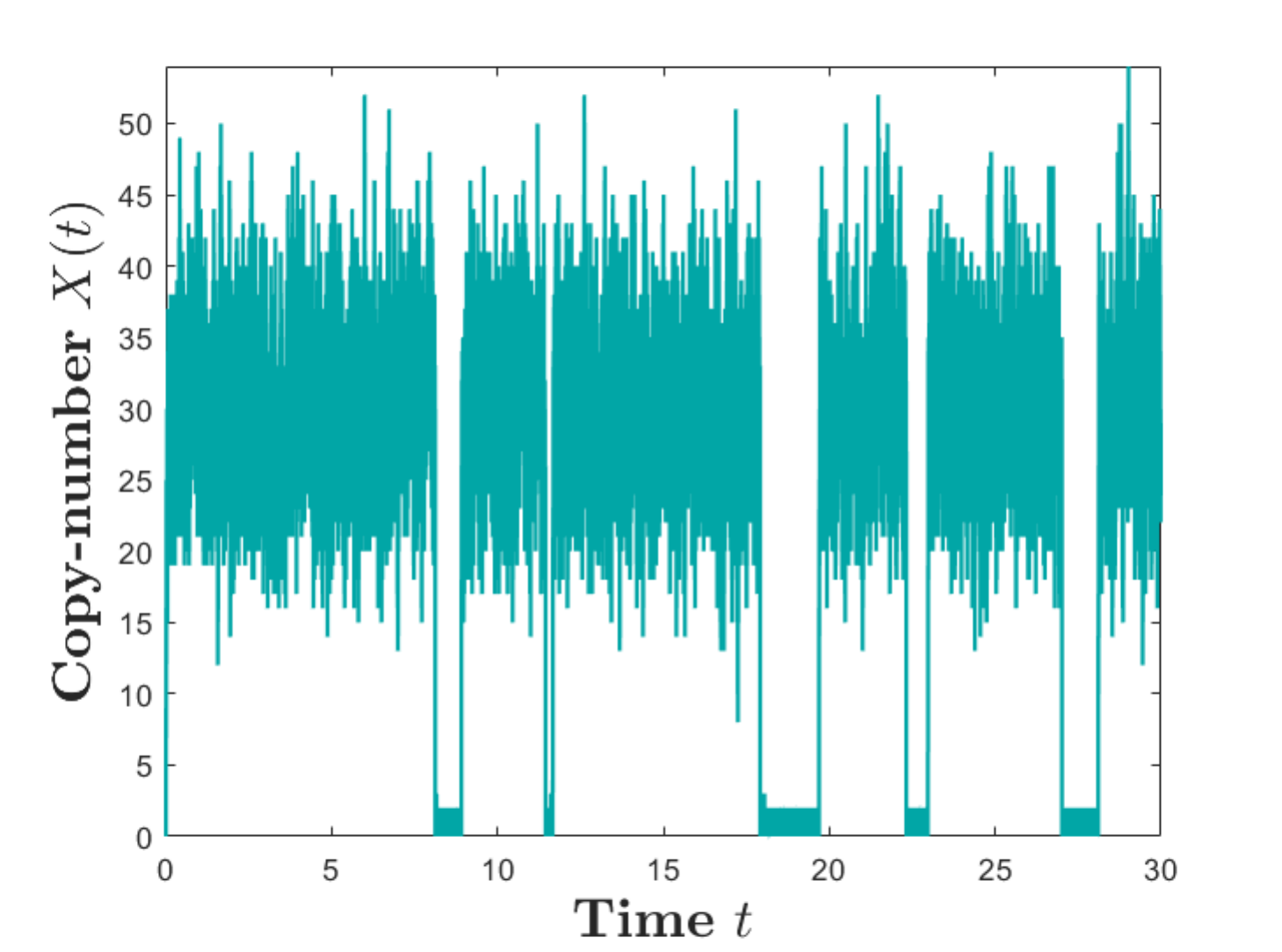}
}
\vskip -4.5cm
\leftline{\hskip 1.5cm (c) \hskip 6.9cm (d)} 
\vskip 3.8cm 
\caption{ 
Application of the higher-resolution control from 
{\rm Algorithm~\ref{alg:SM_algorithm}} on the input network~{\rm (\ref{eq:input_1})}
with $(\alpha_1, \alpha_2) = (1, 1/15)$. 
{\it The stationary {\rm PMF} of the input network is displayed
as the interpolated black dots. 
Panel {\rm (a)} shows the stationary $x$-marginal {\rm PMF}
of the output network~{\rm (\ref{eq:input_1})$\cup$(\ref{eq:R_gamma_Kronecker_1})}, 
taken in the limit $\mu \to 0$ for computational 
efficiency (see also {\rm Theorem~\ref{theorem:Kronecker_mu}} in
 {\rm Appendix~\ref{app:analysis}}), with $\beta_{1,1} = 1$ and 
two different values of the asymptotic parameters $(\varepsilon, \sigma)$, 
shown as the interpolated purple squares and the cyan histogram. 
Panel {\rm (b)} displays an analogous plot for the output 
network~{\rm (\ref{eq:input_1})$\cup$(\ref{eq:R_gamma_Kronecker_2})}
with $(\beta_{1,1}, \beta_{1,2}, \beta_{2,3}, \beta_{3,1}) = (1, 1/3, 1/3, 1/3)$. 
Panel {\rm (c)} shows the stationary $x$-marginal {\rm PMF}
of the output network~{\rm (\ref{eq:input_1})$\cup$(\ref{eq:R_gamma_Kronecker_3})}
with $(\beta_{1,1}, \beta_{1,2}, \beta_{2,1}) = (1, 1/8, 7/8)$, 
$(\gamma_{1}^{\mathcal{P}}, \tilde{\gamma}_{1}^{\mathcal{P}}) = (30,1)$, 
$(\gamma_{0,1}^{\delta}, \gamma_{0,2}^{\delta}, \gamma_2^{\delta}) 
= (\mu^2 \varepsilon \sigma)^{-1/3} (\mu^{1/3}, \mu^{-1/6}, \mu^{-1/6})$, 
$\mu = 10^{-10}$ and two different values for the pair $(\varepsilon, \sigma)$, 
while panel {\rm (d)} displays a representative sample path
 corresponding to the histogram from panel {\rm (c)}.
}
}
\label{fig:Kronecker}
\end{figure}

\emph{Kronecker-delta distribution}. 
Consider the stochastic morpher $\mathcal{R}_{\beta} \cup \mathcal{R}_{\gamma}^{\delta} = 
\mathcal{R}_{\beta}(Y_1) \cup \mathcal{R}_{\gamma}^{\delta}(X, Z_1, Z_2; \, Y_1)$, given by
\begin{align}
\mathcal{R}_{\beta}:
& & 2 Y_1 & \xrightarrow[]{\beta_{1,1}} Y_1, \nonumber \\
\mathcal{R}_{\gamma}^{\delta}: 
 & \hspace{1cm} \mathcal{R}_{\gamma_0}^{\mu,\varepsilon,\sigma}:  \hspace{-1.5cm}
&  \varnothing & \xrightarrow[]{1/\varepsilon} X, \nonumber \\
&&  X  &\xrightleftharpoons[1/\mu]{\gamma_{0,1}} Z_1, \nonumber \\
&& X + Z_1 & \xrightleftharpoons[1/\mu]{\gamma_{0,2}} Z_2, \nonumber \\
& \hspace{1cm} \mathcal{R}_{\gamma_1}^{\mu,\varepsilon,\sigma}: \hspace{-1.5cm}
& Y_1 + Z_2 & \xrightarrow[]{\gamma_{1}} Y_1 + Z_1.
\label{eq:R_gamma_Kronecker_1}
\end{align}
Network $\mathcal{R}_{\gamma}^{\delta} = \mathcal{R}_{\gamma_0}^{\mu,\varepsilon,\sigma} 
\cup \mathcal{R}_{\gamma_1}^{\mu,\varepsilon,\sigma}$ consists of two sub-networks: 
$\mathcal{R}_{\gamma_0}^{\mu,\varepsilon,\sigma}$ describes a production of
$X$, a reversible conversion of $X$ into an auxiliary species $Z_1$, 
and a reversible conversion of $X$ and $Z_1$ into another auxiliary species $Z_2$. 
On the other hand, $\mathcal{R}_{\gamma_1}^{\mu,\varepsilon,\sigma}$
describes an irreversible conversion of $Z_2$ into $Z_1$, catalyzed by $Y_1$. 
Note that the controlling species $Y_1$ does not react directly with the 
target species $X$. Instead, $Y_1$ acts on $X$ indrectly, via the species $Z_1$
and $Z_2$. For this reason, we call $\mathcal{Z} = \{Z_1, Z_2 \}$ the \emph{mediating} species, 
as they propagate the action of the controlling species $Y_1$ onto
the target species $X$. 

The dynamics of the mediating species are assumed to be sufficiently fast.
More precisely, it is assumed that the coefficients $\gamma_{0,1}$, $\gamma_{0,2}$
and $\gamma_{1}$ from~(\ref{eq:R_gamma_Kronecker_1}) 
satisfy the \emph{kinetic condition} $\mu^2 \gamma_{0,1} \gamma_{0,2} \gamma_1 = 
(\sigma \varepsilon)^{-1}$ with the asymptotic parameters $0 < \mu \ll \varepsilon, \sigma \ll 1$.
This ensures that the network $\mathcal{R}_{\gamma}^{\delta}$
fires sufficiently fast in a balanced way, see Theorem~\ref{theorem:Kronecker_mu} 
in Appendix~\ref{app:analysis}, 
and~\cite{Me_Bimolecular}. Under the kinetic condition, 
species $Z_1$ and $Z_2$ formally satisfy $Z_1 = X$ and $Z_2 = X + Z_1 = 2 X$ in the limit $\mu \to 0$,
and the network $\mathcal{R}_{\gamma}^{\delta}$ from~(\ref{eq:R_gamma_Kronecker_1}) 
 reduces to
\begin{align} 
 & \hspace{1cm} \mathcal{R}_{\gamma_0}^{\varepsilon}:  \hspace{-1.5cm}
&  \varnothing & \xrightarrow[]{1/\varepsilon} X, \nonumber \\
& \hspace{1cm} \mathcal{R}_{\gamma_1}^{\varepsilon,\sigma}: \hspace{-1.5cm}
& Y_1 + 2 X  & \xrightarrow[]{1/(\sigma \varepsilon)} Y_1 + X, 
\; \; \; \; 0 <\varepsilon, \sigma \ll 1.
\label{eq:R_Kronecker_reduced}
\end{align}
The first reaction from~(\ref{eq:R_Kronecker_reduced}) provides a strong positive drift, 
which is overpowered by an even stronger negative
drift, induced by the second reaction from~(\ref{eq:R_Kronecker_reduced}), 
when the copy-number of the target species
satisfies $x > 1$. As a consequence, the target species 
$X$ from the output network~(\ref{eq:input_1})$\cup$(\ref{eq:R_gamma_Kronecker_1})
spends most of the time at the single state $x = 1$, 
i.e. the stationary $x$-marginal PMF is a Kronecker-delta distribution centered at $x = 1$, 
which we denote by $\delta_{x,1}$. In Figure~\ref{fig:Kronecker}(a), 
we display the stationary $x$-marginal PMF of the output 
network~(\ref{eq:input_1})$\cup$(\ref{eq:R_gamma_Kronecker_1}),
taken in the limit $\mu \to 0$ for computational 
efficiency. The PMF is shown in purple when the remaining two asymptotic parameters
are fixed to $\varepsilon = \sigma = 1$, while as the cyan histogram
when $\varepsilon = \sigma = 10^{-2}$, which is in excellent agreement
with the Kronecker-delta distribution $\delta_{x,1}$. 

Before proceeding to further applications of the higher-resolution control, 
let us explain briefly why the network $\mathcal{R}_{\gamma}^{\delta}$ 
from~(\ref{eq:R_gamma_Kronecker_1}), 
involving the mediating species $Z_1$ and $Z_2$, has been put forward, 
as opposed to the dynamically similar network~(\ref{eq:R_Kronecker_reduced}), 
which has been put forward to implement 
Kronecker-delta distributions in~\cite{David2}. 
The former network is bi-molecular,
and hence experimentally implementable in principle~\cite{DNAComputing1}. 
On the other hand, network~(\ref{eq:R_Kronecker_reduced})
 contains a third-order (tri-molecular) reaction 
which may not be directly experimentally implementable. 
As exemplified shortly, encoding Kronecker-delta distributions
centered at higher values of $x$ is achieved in our framework
by simply adding more auxiliary species, which participate
in up-to second-order reactions. On the other hand, this
is achieved in networks of the form~(\ref{eq:R_Kronecker_reduced}) 
by further increasing the order of some of the underlying reactions, 
thus making such networks experimentally less desirable.

\emph{Uniform distribution}. 
Using multiple Kronecker-delta distributions, by applying 
the higher-resolution control from Algorithm~\ref{alg:SM_algorithm}, 
one may achieve arbitrary probability distributions on bounded
domains. For example, having achieved a probability distribution
concentrated at a single point, utilizing the 
controller~(\ref{eq:R_gamma_Kronecker_1}), 
let us now morph the stationary PMF
of the input network~(\ref{eq:input_1}) into a 
uniform distribution on the state-space $x \in \{1,2,3\}$, 
via the controller
\begin{align}
\mathcal{R}_{\beta}:
& & 2 Y_1 & \xrightarrow[]{\beta_{1,1}} 
Y_1 \xrightarrow[]{\beta_{1,2}} 
Y_2 \xrightarrow[]{\beta_{2,3}}
Y_3 \xrightarrow[]{\beta_{3,1}} Y_1, 
\nonumber \\
\mathcal{R}_{\gamma}^{\delta}: 
 & \hspace{1cm} \mathcal{R}_{\gamma_0}^{\mu,\varepsilon,\sigma}:  \hspace{-1.5cm}
&  \varnothing & \xrightarrow[]{1/\varepsilon} X, \nonumber \\
&&  X  &\xrightleftharpoons[1/\mu]{\gamma_{0,1}} Z_1, \nonumber \\
&& X + Z_1 & \xrightleftharpoons[1/\mu]{\gamma_{0,2}} Z_2, \hspace{0.3cm}
      X + Z_2     \xrightleftharpoons[1/\mu]{\gamma_{0,3}} Z_3, \hspace{0.3cm}
      X + Z_3     \xrightleftharpoons[1/\mu]{\gamma_{0,4}} Z_4,
\nonumber \\
& \hspace{1cm} \mathcal{R}_{\gamma_1}^{\mu,\varepsilon,\sigma}: \hspace{-1.5cm}
& Y_1 + Z_2 & \xrightarrow[]{\gamma_{1}} Y_1 + Z_1, \nonumber \\
& \hspace{1cm} \mathcal{R}_{\gamma_2}^{\mu,\varepsilon,\sigma}: \hspace{-1.5cm}
& Y_2 + Z_3 & \xrightarrow[]{\gamma_{2}} Y_2 + Z_2, \nonumber \\
& \hspace{1cm} \mathcal{R}_{\gamma_3}^{\mu,\varepsilon,\sigma}: \hspace{-1.5cm}
& Y_3 + Z_4 & \xrightarrow[]{\gamma_{3}} Y_3 + Z_3, 
\; \; \; \; 0 < \mu \ll \varepsilon, \sigma \ll 1.
\label{eq:R_gamma_Kronecker_2}
\end{align}
Network~(\ref{eq:R_gamma_Kronecker_2}) involves four mediating species, which, 
under suitable kinetic conditions (see Theorem~\ref{theorem:Kronecker_mu} 
in Appendix~\ref{app:analysis}), 
 formally satisfy $Z_1 = X$, $Z_2 = 2 X$, $Z_3 = 3 X$ and $Z_4 = 4 X$. 
Consequently, the fast production reaction  
$\varnothing \xrightarrow[]{1/\varepsilon} X$ and
$\mathcal{R}_{\gamma_i}^{\mu,\varepsilon,\sigma}$
generate the Kronecker-delta distributions centered at $x = i$, 
for $i \in \{1, 2, 3\}$, and the weight of each of the three
Kronecker-delta distributions is controlled with 
the rate coefficients from the sub-network 
$\mathcal{R}_{\beta}$, in the same manner as in 
network~(\ref{eq:R_gamma_Poisson_3}). 
In particular, a uniform distribution may be achieved 
by taking $\boldsymbol{\beta} = (\beta_{1,1}, 
\beta_{1,2}, \beta_{2,3}, \beta_{3,1}) = (1, 1/3, 1/3, 1/3)$. 
Analogous to Figure~\ref{fig:Kronecker}(a), 
 in Figure~\ref{fig:Kronecker}(b) we display the $x$-marginal PMF of the output 
network~(\ref{eq:input_1})$\cup$(\ref{eq:R_gamma_Kronecker_2}).
Let us note that, while the weights of the Kronecker-delta distributions
are encoded kinetically, in the rate coefficients from the sub-network 
$\mathcal{R}_{\beta}$, the centers of the distributions are encoded
stoichiometrically, i.e. they are determined by which mediating species
is catalyzed by the controlling species. 
 
\emph{Hybrid control}. 
One may also wish to combine the lower- and higher-resolution
networks $\mathcal{R}_{\gamma}^{\mathcal{P}}$ and $\mathcal{R}_{\gamma}^{\delta}$ 
from Algorithm~\ref{alg:SM_algorithm}, respectively, 
into a composite hybrid scheme for biochemical control.
For example, one may wish to obtain a more-detailed control 
over regions of the state-space where the target species are
in lower copy-numbers, while a less-detailed control may be sought over the state-space where 
the target species are in higher copy-numbers. 
Such a hybrid approach may be experimentally desirable, as
biochemical realizations of the Kronecker-delta PMFs 
centered at lower copy-numbers of the species
are less expensive to engineer, since a smaller number
of the mediating species $\mathcal{Z}$ is required.
For example, in Figure~\ref{fig:Kronecker}(c), we
embed a hybrid controller into the input network~(\ref{eq:input_1}),
and morph the PMF into a mixture of a Kronecker-delta distribution
at $x = 1$ and a Poisson distribution at $x = 30$, 
with the underlying sample path shown 
in Figure~\ref{fig:Kronecker}(d). The hybrid controller is 
given as the network~(\ref{eq:R_gamma_Kronecker_3})
in Appendix~\ref{app:hybrid}. 

\section{Bi-stable input network} \label{sec:example2}
In Section~\ref{sec:example1}, we have applied Algorithm~\ref{alg:SM_algorithm}
in order to control the one-species uni-molecular 
input network~(\ref{eq:input_1}).
In this section, we apply Algorithm~\ref{alg:SM_algorithm}
to a more complicated reaction network, involving 
bi-molecular reactions and multiple biochemical species.
In particular, we highlight how
 Algorithm~\ref{alg:SM_algorithm} may be applied to simultaneously
control multiple species, instead of only one species at a time. 
Furthermore, we analyze the dynamics
of the species which are not explicitly controlled.

To this end, we put forward the three-species bi-molecular
 input network $\mathcal{R}_{\alpha}^2 = \mathcal{R}_{\alpha}^2(X_1, X_2, X_3)$,
under mass-action kinetics, given by
\begin{align}
\mathcal{R}_{\alpha}^2:
& & \varnothing & \xrightleftharpoons[\alpha_{2}]{\alpha_{1}} X_1, \hspace{0.3cm}
\varnothing \xrightarrow[]{\alpha_{3}} X_2, \hspace{0.3cm}
\varnothing \xrightarrow[]{\alpha_{4}} X_3, \hspace{0.3cm}
X_3 \xrightarrow[]{\alpha_{5}} X_1, \nonumber \\
& & 2 X_1 &\xrightarrow[]{\alpha_{6}} 2 X_1 + X_3, \hspace{0.3cm}
X_1 + X_2 \xrightarrow[]{\alpha_{7}} 2 X_1, \hspace{0.3cm}
X_1 + X_3 \xrightarrow[]{\alpha_{8}} X_3.
\label{eq:input_2}
\end{align}
 For illustrative purposes, in what follows, we focus on controlling the target species
$\mathcal{X}_{\tau} = \{X_1, X_2\}$, while the remaining species, 
called the residual species and denoted by
$\mathcal{X}_{\rho} = \{X_3\}$, is implicitly influenced, but not explicitly controlled
(see also Figure~\ref{fig:Control_Theory} in Section~\ref{sec:intro} 
for a visualization of the target and residual species in a general setting).
For a particular choice of the rate coefficients $\boldsymbol{\alpha}$, 
the stationary $(x_1,x_2)$-marginal PMF of the input network
$\mathcal{R}_{\alpha}^2$ is shown in Figure~\ref{fig:Schlogl}(a), 
while the underlying sample paths for $X_1$ and $X_2$ 
are shown in cyan and red in Figure~\ref{fig:Schlogl}(b), 
respectively. The $(x_1,x_2)$-marginal PMF is bi-modal, 
with the modes approximately given by $(x_1,x_2) = (10,40)$
and $(x_1,x_2) = (40,10)$, and with the species $X_1$ and $X_2$
being negatively correlated. 

\begin{figure}[!htbp]
\centerline{
\hskip 0mm
\includegraphics[width=0.35\columnwidth]{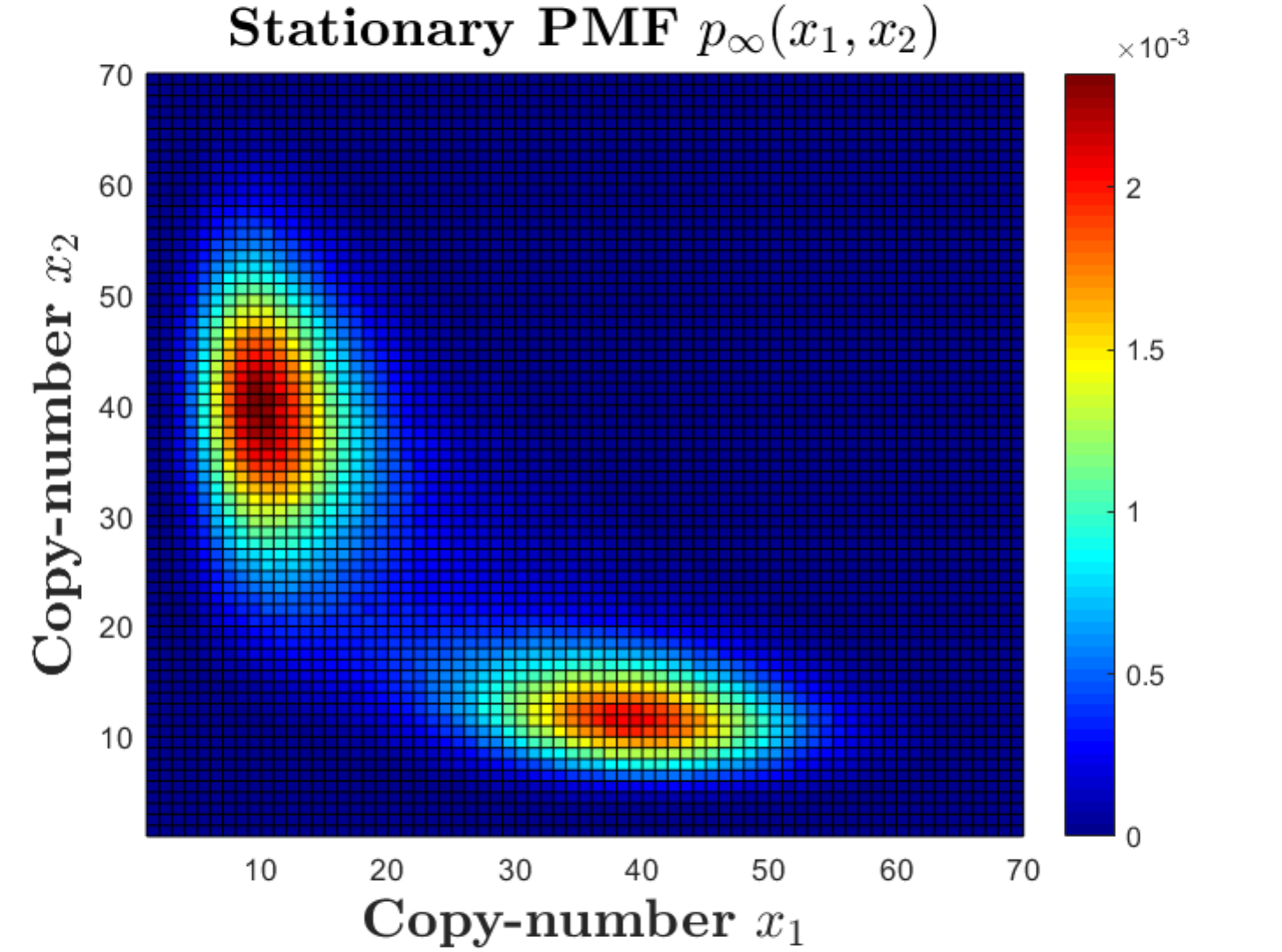}
\hskip 0.05cm
\includegraphics[width=0.35\columnwidth]{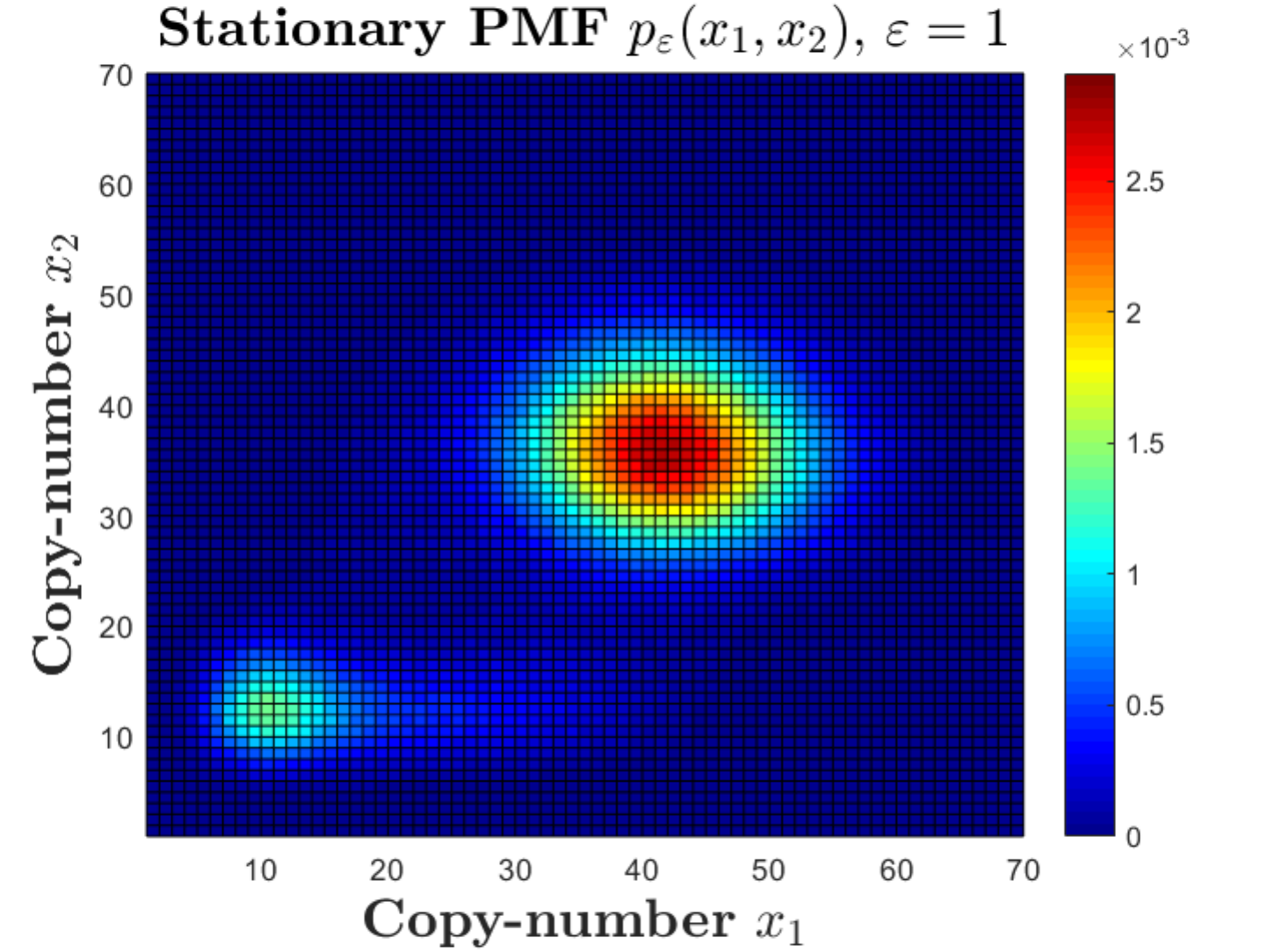}
\hskip 0.05cm
\includegraphics[width=0.35\columnwidth]{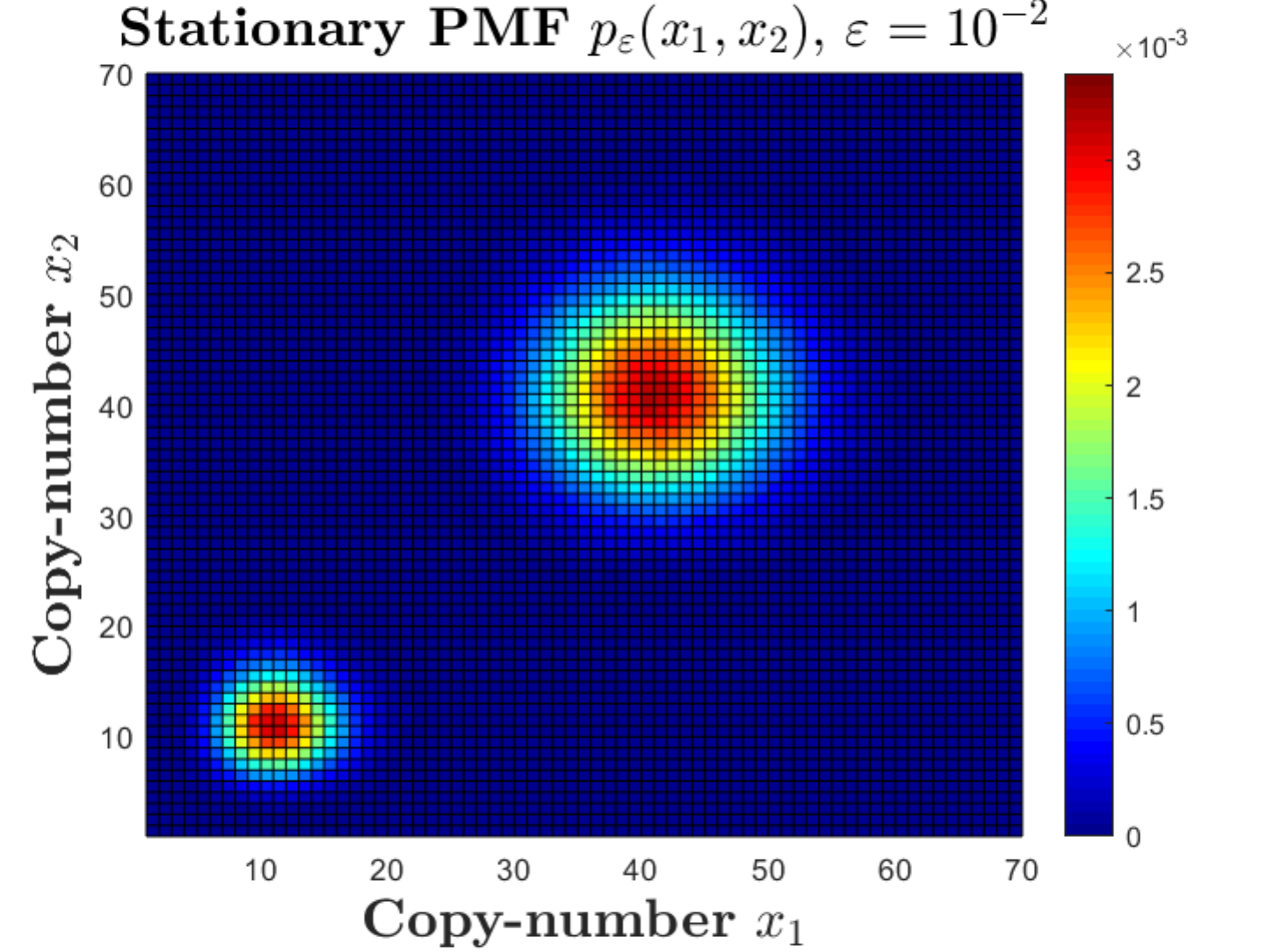}
}
\vskip -4.5cm
\leftline{\hskip -0.8cm (a) \hskip 5.3cm (c) \hskip 5.4cm (e)} 
\vskip 4.0cm
\centerline{
\hskip 0mm
\includegraphics[width=0.35\columnwidth]{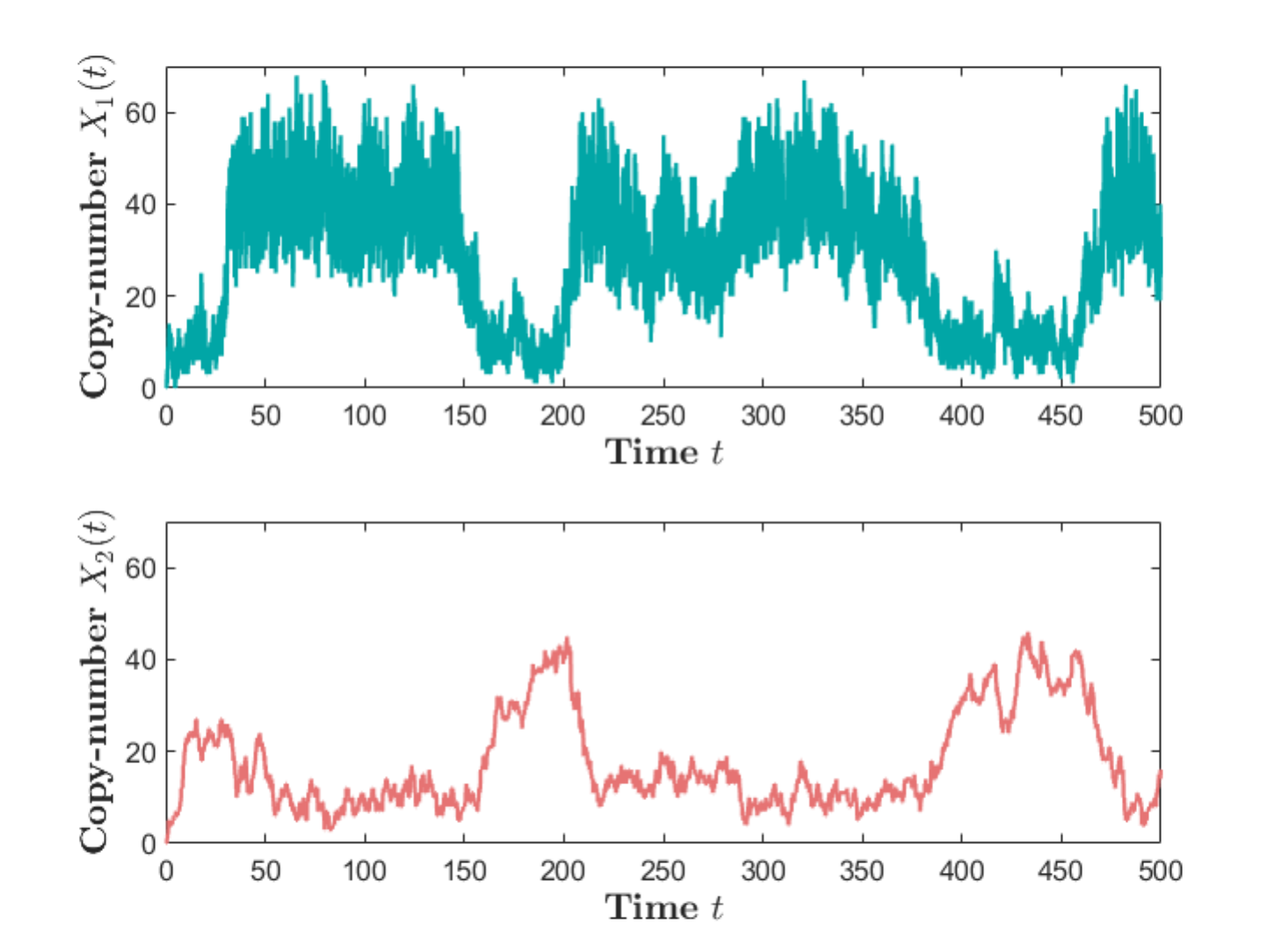}
\hskip 0.05cm
\includegraphics[width=0.35\columnwidth]{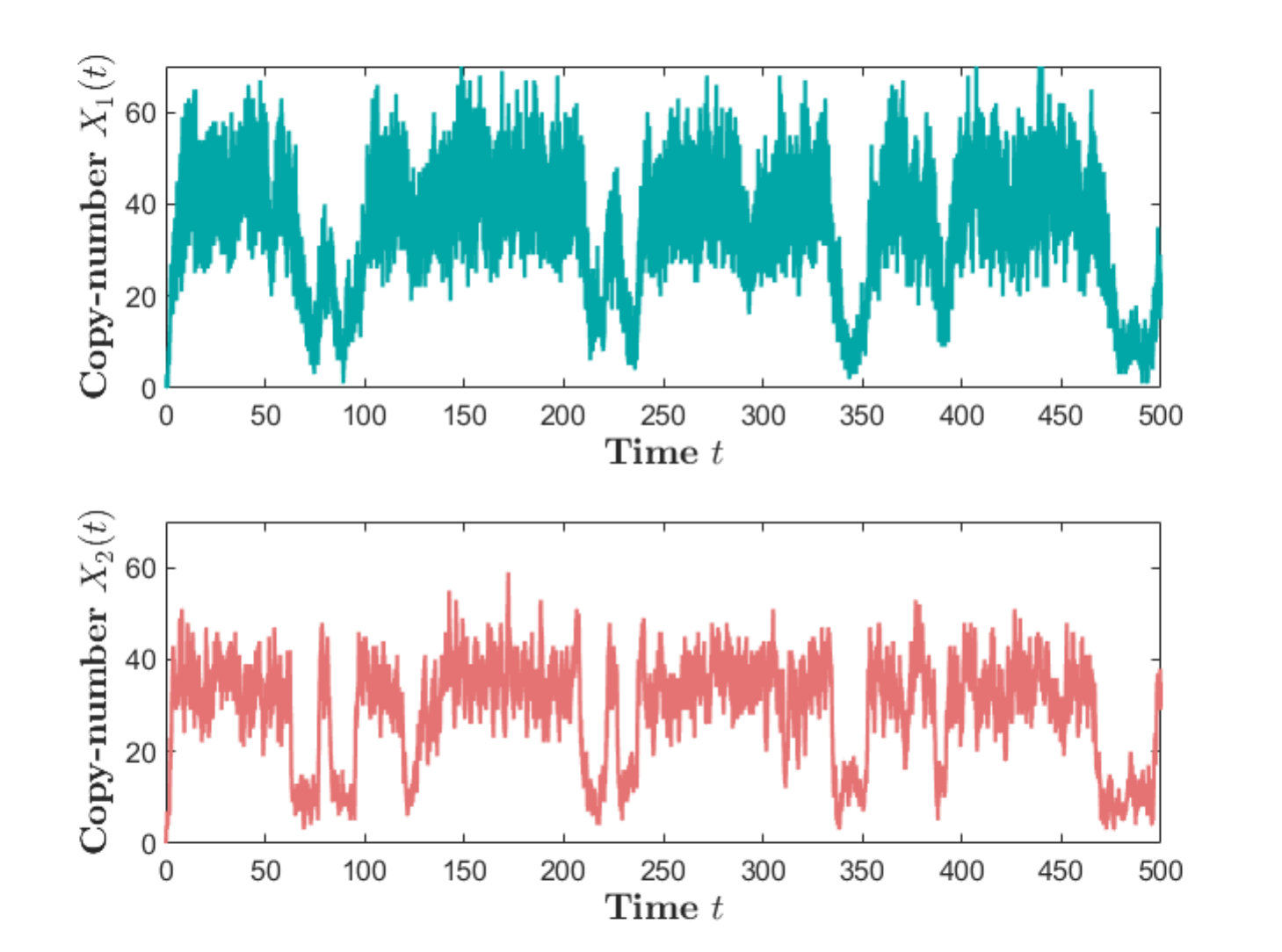}
\hskip 0.05cm
\includegraphics[width=0.35\columnwidth]{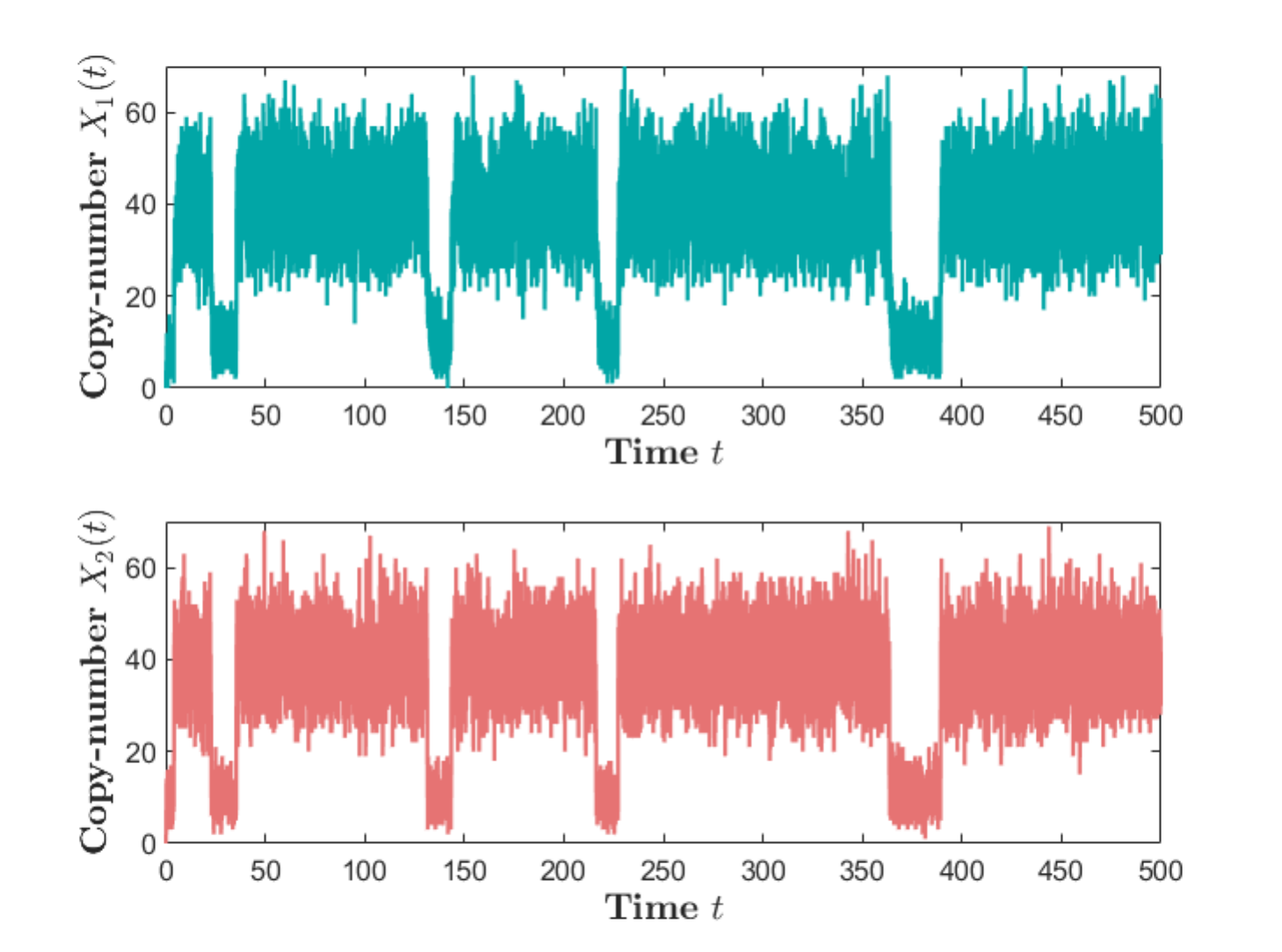}
}
\vskip -4.5cm
\leftline{\hskip -0.8cm (b) \hskip 5.3cm (d) \hskip 5.4cm (f)} 
\vskip 4.2cm
\centerline{
\hskip 0mm
\includegraphics[width=0.35\columnwidth]{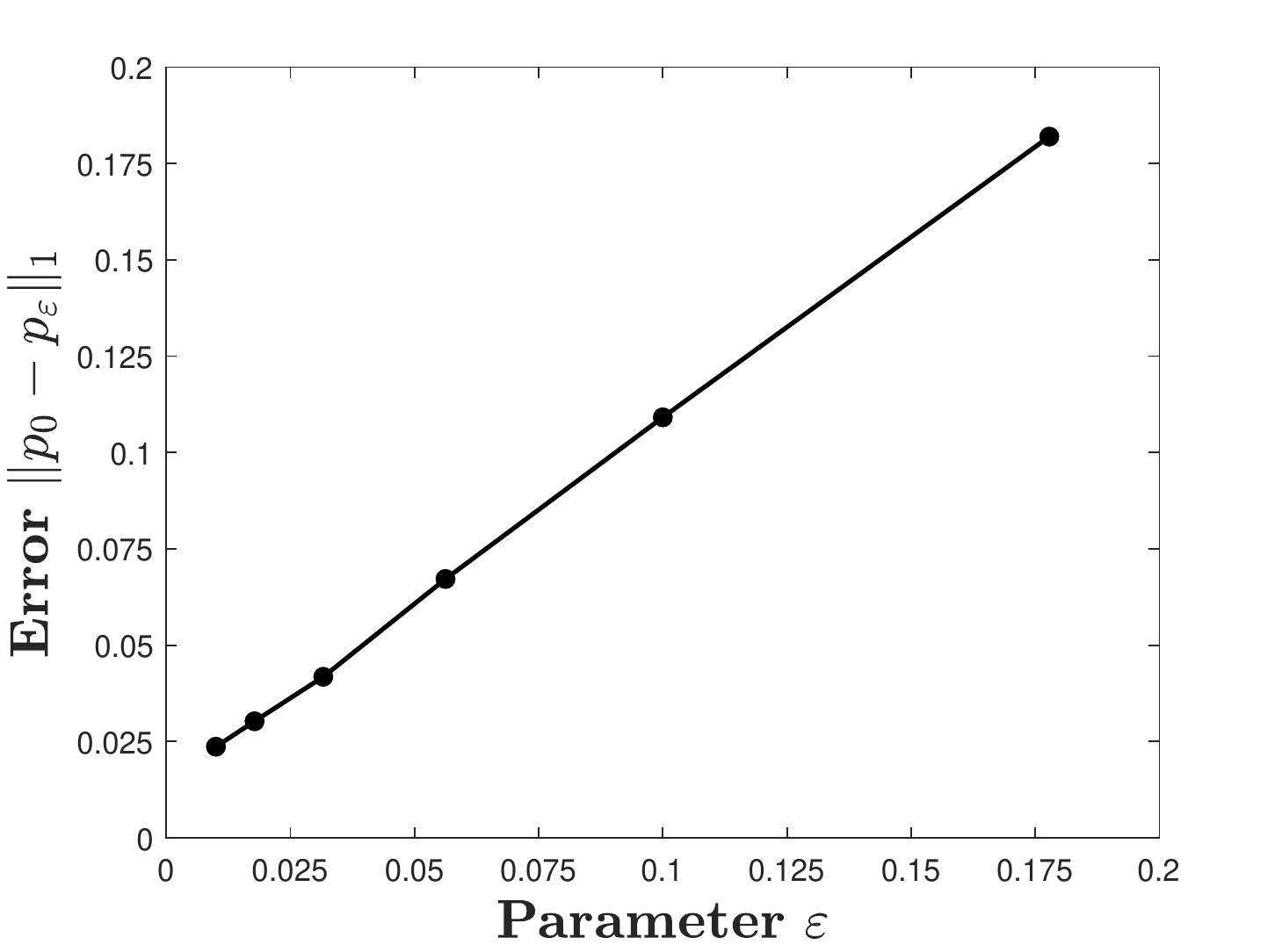}
\hskip 0.1cm
\includegraphics[width=0.35\columnwidth]{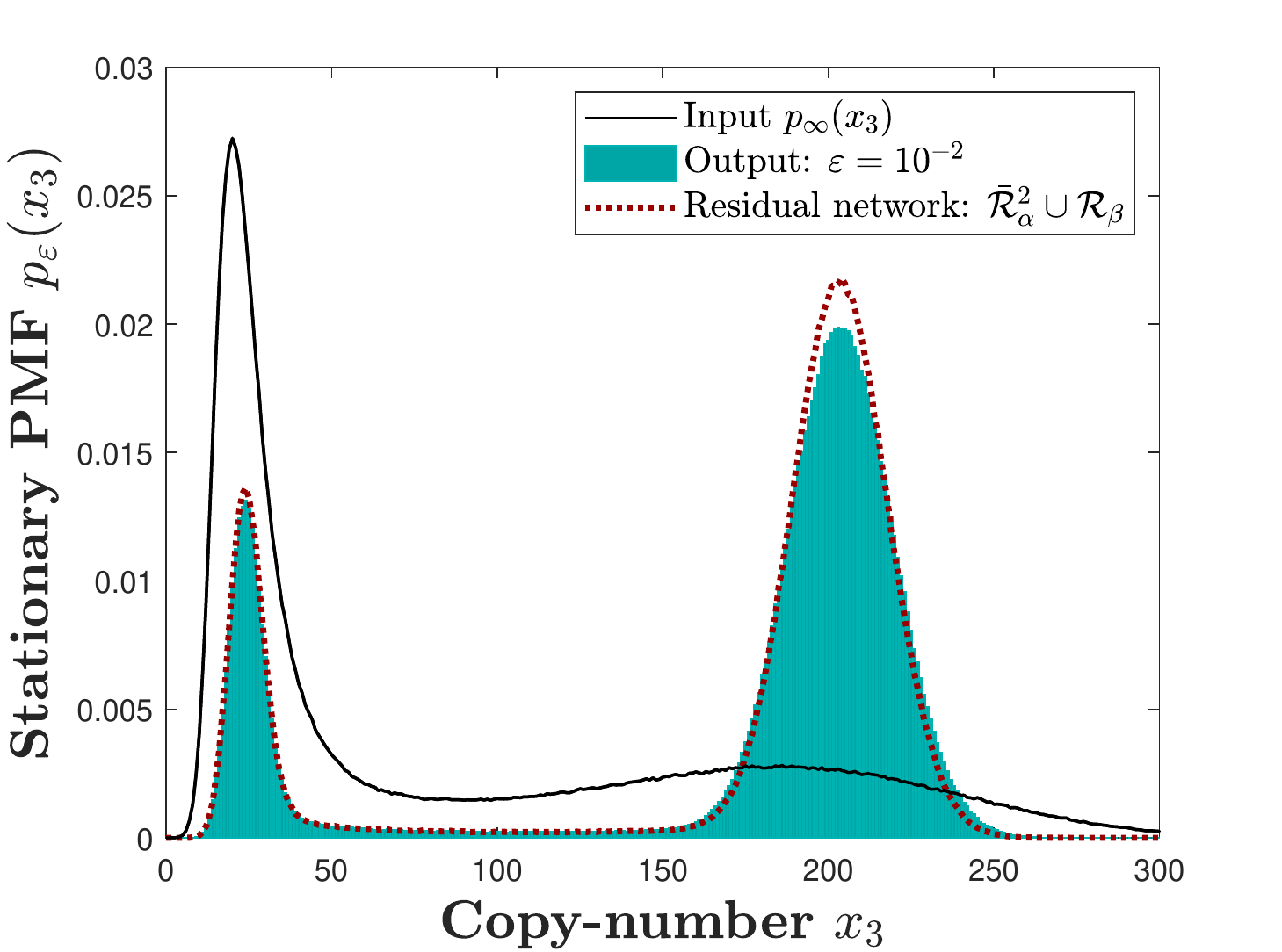}
}
\vskip -4.5cm
\leftline{\hskip 2.2cm (g) \hskip 5.4cm (h)} 
\vskip 3.8cm 
\caption{ 
Application of the lower-resolution control from 
{\rm Algorithm~\ref{alg:SM_algorithm}} on the input network~{\rm (\ref{eq:input_2})}
with $\boldsymbol{\alpha} = 
(\alpha_1, \alpha_2,\alpha_3, \alpha_4,\alpha_5, \alpha_6,\alpha_7, \alpha_8) = 
(2,7/2,2,18,3/2,9/50,1/200,1/48)$.
{\it Panels {\rm (a)}, and {\rm (b)}, show the stationary $(x_1, x_2)$-marginal {\rm PMF}
of the input network~{\rm (\ref{eq:input_2})}, and the underlying
representative sample paths for target species $X_1$ and $X_2$, respectively. 
Analogous plots are shown in panels {\rm (c)}--{\rm (d)}, and 
{\rm (e)}--{\rm (f)}, for the output 
network~{\rm (\ref{eq:input_2})$\cup$(\ref{eq:R_gamma_Poisson_2_2})}
with $\boldsymbol{\beta} = (\beta_{1,1}, \beta_{1,2}, \beta_{2,1}) = 
(1,4/50,1/50)$, $\boldsymbol{\gamma} = (\boldsymbol{\gamma}_1, \boldsymbol{\gamma}_2)
= ((\gamma_{0,1}, \gamma_{1,1}, \gamma_{2,1}), 
(\gamma_{0,2}, \gamma_{1,2}, \gamma_{2,2})) = 
((1,10,40), (1,10,40))$, and different values of the 
asymptotic parameter $\varepsilon$, as indicated in the plots.
Panel {\rm (g)} displays as the black dots, interpolated with the black lines, a
plot of the $l^1$-distance between the target {\rm PMF}~{\rm (\ref{eq:R_gamma_Poisson_2_PMF_2})}
and the long-time {\rm PMF} of output 
network~{\rm (\ref{eq:input_2})$\cup$(\ref{eq:R_gamma_Poisson_2_2})}
as a function of the asymptotic parameter $\varepsilon$. 
Panel {\rm (h)} shows the stationary $x_3$-marginal {\rm PMFs}
of the input network~{\rm (\ref{eq:input_2})}, and the output 
network~{\rm (\ref{eq:input_2})$\cup$(\ref{eq:R_gamma_Poisson_2_2})}
with $\varepsilon = 10^{-2}$, as the black solid curve, and the cyan 
histogram, respectively. Also shown, as the dotted red curve,
 is the stationary $x_3$-marginal 
{\rm PMF}of the residual network~{\rm (\ref{eq:input_2_residual})}.
}
}
\label{fig:Schlogl}
\end{figure}

\emph{Target species}. Let us now apply Algorithm~\ref{alg:SM_algorithm} in order 
to morph the input PMF from Figure~\ref{fig:Schlogl}(a)
into a bi-modal one, with the species $X_1$ and $X_2$
being positively correlated. More precisely, let us morph
the input modes into the target modes given by $(x_1,x_2) = (10,10)$
and $(x_1,x_2) = (40,40)$, where the $(x_1,x_2)$-marginal PMF
takes approximately the same values, and with the switching time
between the two new modes being of the order $\mathcal{O}(10)$ time-units. 
To this end, consider the 
stochastic morpher $\mathcal{R}_{\beta} \cup \mathcal{R}_{\gamma}^\mathcal{P} = 
\mathcal{R}_{\beta}(Y_1, Y_2) \cup \mathcal{R}_{\gamma}^\mathcal{P}(X_1, X_2; \, Y_1, Y_2)$, 
given by
\begin{align}
\mathcal{R}_{\beta}:
& & 2 Y_1 & \xrightarrow[]{\beta_{1,1}} Y_1 \xrightleftharpoons[\beta_{2,1}]{\beta_{1,2}} Y_2,
\nonumber \\
\mathcal{R}_{\gamma}^{\mathcal{P}}: 
& & \ \mathcal{R}_{\gamma_0}^{\varepsilon}:  
\; \; \; \;                        X_1 & \xrightarrow[]{\gamma_{0,1}/\varepsilon} \varnothing, 
\hspace{1.3cm}         X_2  \xrightarrow[]{\gamma_{0,2}/\varepsilon} \varnothing, \nonumber \\
& & \mathcal{R}_{\gamma_1}^{\varepsilon}:  
\; \; \; \;                 Y_1 & \xrightarrow[]{\gamma_{1,1}/\varepsilon} Y_1 + X_1, 
\hspace{0.3cm}  Y_1 \xrightarrow[]{\gamma_{1,2}/\varepsilon} Y_1 + X_2, \nonumber \\
& & \mathcal{R}_{\gamma_2}^{\varepsilon}:  
\; \; \; \;                Y_2 & \xrightarrow[]{\gamma_{2,1}/\varepsilon} Y_2 + X_1, 
\hspace{0.3cm}  Y_2 \xrightarrow[]{\gamma_{2,2}/\varepsilon} Y_2 + X_2, 
\; \; \; \; 0 < \varepsilon \ll 1.
\label{eq:R_gamma_Poisson_2_2}
\end{align}
Controller~(\ref{eq:R_gamma_Poisson_2_2}) is a two-target-species analogue
of the network~(\ref{eq:R_gamma_Poisson_2}), with each of the controlling species
$\mathcal{Y} = \{Y_1, Y_2\}$ now producing both of the target species $\mathcal{X}_{\tau} = \{X_1, X_2\}$. 
The stationary $(x_1, x_2)$-marginal PMF of the output 
network~(\ref{eq:input_2})$\cup$(\ref{eq:R_gamma_Poisson_2_2})
is given, in the limit $\varepsilon \to 0$, by
\begin{align}
p_{0}(x_1, x_2) & = \left(1 + \frac{\beta_{1,2}}{\beta_{2,1}} \right)^{-1}
\mathcal{P} \left(x_1; \,  \frac{\gamma_{1,1}}{\gamma_{0,1}} \right) 
\mathcal{P} \left(x_2; \,  \frac{\gamma_{1,2}}{\gamma_{0,2}} \right) 
+ \left(1 + \frac{\beta_{2,1}}{\beta_{1,2}} \right)^{-1}
\mathcal{P} \left(x_1; \,  \frac{\gamma_{2,1}}{\gamma_{0,1}} \right) 
\mathcal{P} \left(x_2; \,  \frac{\gamma_{2,2}}{\gamma_{0,2}} \right), \label{eq:R_gamma_Poisson_2_PMF_2}
\end{align}
see also Theorem~\ref{theorem:Poisson_Kronecker} from 
Appendix~\ref{app:long_time}.
Note that~(\ref{eq:R_gamma_Poisson_2_PMF_2}) is a linear combination of 
a product of two Poisson distributions, which is
a two-dimensional analogue of~(\ref{eq:R_gamma_Poisson_2_PMF}).
 
In order to achieve the desired modes, we fix
$\boldsymbol{\gamma} = (\boldsymbol{\gamma}_1, \boldsymbol{\gamma}_2)
= ((\gamma_{0,1}, \gamma_{1,1}, \gamma_{2,1}), 
(\gamma_{0,2}, \gamma_{1,2}, \gamma_{2,2})) = 
((1,10,40), (1,10,40))$. One the other hand, 
in order to ensure that the PMF takes 
approximately equal values at the two modes,  and
that the switching time is of the order $\mathcal{O}(10)$ time-units, 
we set $\beta_{1,2}/\beta_{2,1} = 4$, and 
$(\beta_{1,2} + \beta_{2,1}) = 1/10$, respectively, which is achieved
by taking $\boldsymbol{\beta} = (\beta_{1,1}, \beta_{1,2}, \beta_{2,1}) = 
(1,4/50,1/50)$.
In Figure~\ref{fig:Schlogl}(c)--(d), we display the stationary 
$(x_1, x_2)$-marginal PMF of the output 
network~(\ref{eq:input_2})$\cup$(\ref{eq:R_gamma_Poisson_2_2}), 
and the underlying representative sample paths, when 
the asymptotic parameter is fixed to $\varepsilon = 1$.  
One can notice that the two modes from the input PMF, shown
in Figure~\ref{fig:Schlogl}(a), have already largely redistributed across the 
two target modes, concentrating more
near $(x_1, x_2) = (40,40)$. 
Broadly speaking, as the parameter $\varepsilon$ is decreased, 
the input PMF is at first more attracted towards 
the target mode $(x_1, x_2) = (40,40)$, than to
$(x_1, x_2) = (10,10)$, due to the fact that the former mode 
contains significantly more probability mass in the limit
$\varepsilon \to 0$, under the particular choice of $\boldsymbol{\beta}$.
Note that the PMF already displays bi-modality and positive correlation
between the target species for $\varepsilon = 1$.
As the parameter $\varepsilon$ is further decreased, the PMF 
from Figure~\ref{fig:Schlogl}(c) gradually reshapes into the desired form. 
In Figure~\ref{fig:Schlogl}(e)--(f), we take the asymptotic parameter
$\varepsilon = 10^{-2}$, so that the largest rate coefficients 
from the input network~(\ref{eq:input_2}) and 
the controller~(\ref{eq:R_gamma_Poisson_2_2}) are separated by two orders of magnitude,
and one can notice that the 
stationary PMF has converged close to the target~(\ref{eq:R_gamma_Poisson_2_PMF_2}) 
when $\varepsilon = 10^{-2}$.
Comparing Figures~\ref{fig:Schlogl}(a)--(b) and (e)--(f), 
one can also notice that the marginal modes for the 
target species $\mathcal{X}_{\tau} = \{X_1, X_2\}$ have been approximately preserved under the 
stochastic bifurcation induced by controller~(\ref{eq:R_gamma_Poisson_2_2}), 
while the correlation has been reversed from negative to positive, respectively. 
To gain a more quantitative information about the convergence,
in Figure~\ref{fig:Schlogl}(g) we display the distance (error)
between the target stationary $(x_1, x_2)$-marginal 
PMF~\ref{eq:R_gamma_Poisson_2_PMF_2}
and the long-time output PMF for non-zero $\varepsilon$, denoted by $p_{\varepsilon}(x_1, x_2)$,
as a function of the asymptotic parameter $\varepsilon$.
Measuring the error using the $l^1$-norm: $\|p_0 - p_{\varepsilon}\|_1 = 
\sum_{x_1, x_2} |p_0(x_1, x_2) - p_{\varepsilon}(x_1, x_2)|$, 
one can notice that $\|p_0 - p_{\varepsilon}\| = \mathcal{O}(\varepsilon)$
 for sufficiently small $\varepsilon$, 
i.e. the error decreases linearly as a function of $\varepsilon$, 
which also holds for more general input networks, see 
Theorem~\ref{theorem:convergence} in Appendix~\ref{sec:convergence}.

\emph{Residual species}. 
The dynamics of the
target and residual species, $\mathcal{X}_{\tau} = \{X_1, X_2\}$ 
and $\mathcal{X}_{\rho} = \{X_3\}$, respectively, are coupled. 
As a consequence, explicit control of the target species 
implicitly influences the dynamics of the residual species.
In Figure~\ref{fig:Schlogl}(h), we display as the solid black curve, 
and the cyan histogram, the stationary $x_3$-marginal PMFs from 
the input network~(\ref{eq:input_2}), and the output
network~(\ref{eq:input_2})$\cup$(\ref{eq:R_gamma_Poisson_2_2}) with 
$\varepsilon = 10^{-2}$, respectively. 
One can notice that, under the controller~(\ref{eq:R_gamma_Poisson_2_2}), the
PMF of the residual species is redistributed across the two modes, 
which approximately remain fixed in this particular example. 

The dynamics of the residual species $X_3$,
in the limit $\varepsilon \to 0$, is governed by the so-called
\emph{residual} network, denoted by $\bar{\mathcal{R}}_{\alpha}^2 = 
\bar{\mathcal{R}}_{\alpha}^2(X_3; \, Y_1, Y_2)$, which is given by
\begin{align}
\bar{\mathcal{R}}_{\alpha}^2:
& & \varnothing & \xrightleftharpoons[\alpha_{5}]{\alpha_{4}} X_3, \nonumber \\
& & Y_1 & \xrightarrow[]{\alpha_{6} \left(\gamma_{1,1}/\gamma_{0,1} \right)^2} Y_ 1 + X_3, \hspace{0.3cm}
Y_2 \xrightarrow[]{\alpha_{6} \left(\gamma_{2,1}/\gamma_{0,1} \right)^2} Y_2 + X_3.
\label{eq:input_2_residual}
\end{align}
Network~(\ref{eq:input_2_residual}) is obtained by suitably eliminating 
the target species $\mathcal{X}_{\tau} = \{X_1, X_2\}$ from the input network~(\ref{eq:input_2}),
see Appendix~\ref{sec:residual_networks}, and equation~(\ref{eq:residual_Poisson}) in
particular, for more details. Note that the controlling species $\mathcal{Y} = \{Y_1, Y_2\}$, which play
a catalytic role in the sub-network $\mathcal{R}_{\gamma}^{\mathcal{P}}$
from~(\ref{eq:R_gamma_Poisson_2_2}), also play a catalytic role in the
residual network~(\ref{eq:input_2_residual}). 
Note also that, as the black-box input network~(\ref{eq:input_2}) 
is assumed to have an unknown structure, the residual network~(\ref{eq:input_2_residual})
is unknown from the control perspective. 

In Figure~\ref{fig:Schlogl}(h), we show, as the dotted red curve, 
the stationary $x_3$-marginal PMF of the composite residual network
$\bar{\mathcal{R}}_{\alpha}^2 \cup \mathcal{R}_{\beta}$, where
$\mathcal{R}_{\beta}$ is given in~(\ref{eq:R_gamma_Poisson_2_2}), 
which is in good agreement with the cyan histogram, 
verifying the validity of the network~(\ref{eq:input_2_residual}).
In fact, under the parameter choice in this paper, the conversion reactions 
from the network $\mathcal{R}_{\beta}$ fire significantly slower than the residual
network $\bar{\mathcal{R}}_{\alpha}^2$ and, as a consequence, 
the corresponding $x_3$-marginal PMFs is approximately 
given by a linear combination of two Poisson distributions~\cite{Me_Mixing}
centered at $x_3 =(\alpha_4 + \alpha_6 (\gamma_{1,1}/\gamma_{0,1})^2)/\alpha_5 = 24$ and 
$x_3 =(\alpha_4 + \alpha_6 (\gamma_{2,1}/\gamma_{0,1})^2)/\alpha_5 = 204$, 
with weights identical to those from~(\ref{eq:R_gamma_Poisson_2_PMF_2}). 
Let us note that, more generally, 
the residual species do not necessarily inherit multi-modality, nor
Poisson-based PMFs, from the target species. 

\section{Implicit control: Gene-expression input network} \label{sec:example3}
In Section~\ref{sec:example2}, 
we have focused on controlling the target species, while
ignoring the induced implicit effects on the underlying residual species. 
In this section, we shift our focus to an implicit control of the residual species, 
via appropriate manipulations of the target species.
In particular, we exploit the time-scale separation between 
the residual network and the networking governing the dynamics
of the controlling species in order to obtain a control over the residual species.
To this end, consider the two-species uni-molecular
reaction network, denoted by $\mathcal{R}_{\alpha}^3 = 
\mathcal{R}_{\alpha}^3(X_1, X_2)$, and given by
\begin{align}
\mathcal{R}_{\alpha}^3:
& & \varnothing & \xrightleftharpoons[\alpha_{2}]{\alpha_{1}} X_1, \hspace{0.3cm}
X_1 \xrightarrow[]{\alpha_{3}} X_1 + X_2, \hspace{0.3cm}
X_2 \xrightarrow[]{\alpha_{4}} \varnothing. 
\label{eq:input_3}
\end{align}
Network~(\ref{eq:input_3}) may be interpreted 
as a simplified model for gene transcription and translation within 
a biological cell: $X_1$ represents an mRNA species, 
transcribed from a gene, and translated into a suitable protein species $X_2$, with each
of the two species being degradable. For simplicity, 
many biochemical steps necessary for 
gene expression have been omitted~\cite{Kepler}
(e.g. a fixed abundance of genes and ribosomes is assumed, and incorporated as effective rate coefficients).
Let us note that, while we have assumed that the black-box networks~(\ref{eq:input_1}) 
and~(\ref{eq:input_2}) have an unknown structure from the perspective
of control, in this section we assume that the structure of the 
 black-box input network~(\ref{eq:input_3}) is partially known.
More precisely, we assume that the only reactions which 
change the copy-number of $X_2$ are 
$X_1 \to X_1 + X_2$ and $X_2 \to \varnothing$ (whose rate coefficients
may be unknown), which is satisfied if e.g. $X_2$ is involved only as 
a catalyst in the remaining biochemical reactions within the cell, 
while the rest of the structure of~(\ref{eq:input_3}), which may be embedded into a larger 
ambient network, is allowed to be unknown.

Assume one desires to control the protein species $X_2$ 
(and thereby the phenotype of the underlying cells), 
by utilizing a controller which is experimentally realized with RNA 
molecules~\cite{RNAComputing,PostTranscription}. 
In this case, interfacing 
RNA controlling species with the protein $X_2$ may
be a difficult task, as the two types of molecules have different biophysical properties. 
A more natural target for the RNA controlling species
is the mRNA species $X_1$, as they may 
interact via the highly programmable 
toehold-mediated strand-displacement mechanism.
This motivates one to consider the problem of 
explicitly influencing the target species $X_1$, in order to implicitly control the 
residual species $X_2$.
To this end, let us induce bi-modality into the probability distribution
of $X_2$, by considering the stochastic morpher acting on $X_1$, given by
\begin{align}
\mathcal{R}_{\beta}:
& & 2 Y_1 & \xrightarrow[]{\beta_{1,1}} Y_1 \xrightleftharpoons[\beta_{2,1}]{\beta_{1,2}} Y_2,
\nonumber \\
\mathcal{R}_{\gamma}^{\mathcal{P}}: 
& & \mathcal{R}_{\gamma_0}^{\varepsilon}:  \; \; \; \; 
X_1 & \xrightarrow[]{\gamma_{0,1}/\varepsilon} \varnothing, \nonumber \\
& & \mathcal{R}_{\gamma_1}^{\varepsilon}:  \; \; \; \; 
Y_1 & \xrightarrow[]{\gamma_{1,1}/\varepsilon} Y_1 + X_1, \nonumber \\
& & \mathcal{R}_{\gamma_2}^{\varepsilon}:  \; \; \; \;  
Y_2 & \xrightarrow[]{\gamma_{2,1}/\varepsilon} Y_2 + X_1, 
\; \; \; \; 0 < \varepsilon \ll 1.
\label{eq:ex3}
\end{align}
In the limit $\varepsilon \to 0$, the stationary $x_1$-marginal PMF from the output
network~(\ref{eq:input_3})$\cup$(\ref{eq:ex3}), denoted by $p_0(x_1)$, 
has the form~(\ref{eq:R_gamma_Poisson_2_PMF}). 
On the other hand, the residual network, governing the dynamics of the
species $X_2$, denoted by 
$\bar{\mathcal{R}}_{\alpha}^3 = 
\bar{\mathcal{R}}_{\alpha}^3(X_2; \, Y_1, Y_2)$, is given by
\begin{align}
\bar{\mathcal{R}}_{\alpha}^3:
& & X_2 & \xrightarrow[]{\alpha_{4}} \varnothing, \nonumber \\
& & Y_1 & \xrightarrow[]{\alpha_{3} \left(\gamma_{1,1}/\gamma_{0,1} \right)} Y_ 1 + X_2, \hspace{0.3cm}
Y_2 \xrightarrow[]{\alpha_{3} \left(\gamma_{2,1}/\gamma_{0,1} \right)} Y_2 + X_2.
\label{eq:input_3_residual}
\end{align}
Bi-modality in the probability distribution of $X_2$
may be achieved by taking sufficiently slow conversion reactions from the network 
$\mathcal{R}_{\beta}(Y_1, Y_2)$~\cite{Me_Mixing}.
In particular, the stationary $x_2$-marginal PMF, in the limit 
$(\beta_{1,2} + \beta_{2,1}) = \varepsilon_{\beta} \to 0$ with $\beta_{1,2}/\beta_{2,1}$ fixed,
is given by
\begin{align}
p_{0}(x_2) & = \left(1 + \frac{\beta_{1,2}}{\beta_{2,1}} \right)^{-1}
\mathcal{P} \left(x_2; \,  \frac{\alpha_3}{\alpha_4} \left( \frac{\gamma_{1,1}}{\gamma_{0,1}} \right) \right) 
+ \left(1 + \frac{\beta_{2,1}}{\beta_{1,2}} \right)^{-1}
\mathcal{P} \left(x_2; \,  \frac{\alpha_3}{\alpha_4} \left( \frac{\gamma_{2,1}}{\gamma_{0,1}}  \right) \right).
\label{eq:PMF2_ex3}
\end{align}
Note that the assumption that the reactions which change the copy-numbers of $X_2$ are known, 
implies that the structure of the residual network~(\ref{eq:input_3_residual}), and 
the form of the $x_2$-marginal PMF~(\ref{eq:PMF2_ex3}), are also known.

The stationary PMF~(\ref{eq:PMF2_ex3}), describing the long-time dynamics
of the residual species $X_2$, depends on the (generally) unknown input rate coefficients $\alpha_3$ 
and $\alpha_4$, in contrast to the stationary PMF
of the target species, which does not depend on the input coefficients as $\varepsilon \to 0$. 
Despite dependence on the input parameters, 
controllable bi-modality in species $X_2$ may be achieved.  
In particular, the ratio between the centers of the two Poisson distributions
 from~(\ref{eq:PMF2_ex3}) is given by $\gamma_{2,1}/\gamma_{1,1}$, 
which is independent of the rate coefficients $\alpha_3$ and $\alpha_4$. 
Hence, relative distance between the two modes of the species $X_2$ is controllable 
with the ratio of the production rate coefficients of the species $X_1$ 
from the stochastic morpher~(\ref{eq:ex3}). Taking the ratio 
$\gamma_{2,1}/\gamma_{1,1}$ sufficiently large ensures that the
two Poisson distributions from~(\ref{eq:PMF2_ex3}) are well-separated, 
and the probability mass at the two modes
 is controlled by the ratio $\beta_{1,2}/\beta_{2,1}$, as before. 
On the other hand, as opposed to before, in Sections~\ref{sec:example1} and~\ref{sec:example2},
 where $(\beta_{1,2} + \beta_{2,1})$ was chosen to control
the switching time of the underlying sample paths of the target species, in the current setting
we exploit this degree of freedom by sufficiently slowing down the switching time by 
setting $(\beta_{1,2} + \beta_{2,1}) = \varepsilon_{\beta} \ll 1$, 
ensuring validity of~(\ref{eq:PMF2_ex3}).

In Figure~\ref{fig:Genetic}(a) and (b), we display, as the black interpolated dots,  
the stationary PMFs of the target and residual species, $X_1$ and $X_2$, 
respectively, of the input network~(\ref{eq:input_3}) with the rate coefficients fixed to 
$\boldsymbol{\alpha} = (\alpha_1, \alpha_2, \alpha_3, \alpha_4) = (2, 1, 10, 1)$.
Let us morph the input $x_2$-stationary PMF into a bi-modal one, 
of the form~(\ref{eq:PMF2_ex3}), with  the larger mode
being three times further away than the smaller mode, $\gamma_2/\gamma_1 = 3$, 
and with $\beta_{1,2}/\beta_{2,1} = 1$. 
To this end, we consider the output network~(\ref{eq:input_3})$\cup$(\ref{eq:ex3})
with $\boldsymbol{\gamma} = (\gamma_{0,1}, \gamma_{1,1}, \gamma_{2,1}) 
= (1, 1, 3)$, and $\beta_{1,2} + \beta_{2,1} = \varepsilon_{\beta} \ll 1$. 
In Figure~\ref{fig:Genetic}(a), we display the stationary $x_1$-marginal PMF
of the output network~(\ref{eq:input_3})$\cup$(\ref{eq:ex3}) when 
$\varepsilon = 10^{-2}$ as the cyan histogram, which is in an excellent agreement with 
the theoretical prediction~(\ref{eq:R_gamma_Poisson_2_PMF}). 
Note that the $x_1$-marginal PMF is uni-modal, as the two Poisson
distributions from~(\ref{eq:R_gamma_Poisson_2_PMF}) are not sufficiently
well-separated for the chosen parameters, and that it is approximately
independent of the value of $\varepsilon_{\beta}$ (which only influences
appropriate dynamical time-scales of the underlying sample paths). 
On the other hand, in Figure~\ref{fig:Genetic}(b), we display
the stationary $x_2$-marginal PMF
of the output network~(\ref{eq:input_3})$\cup$(\ref{eq:ex3})
as the interpolated purple squares, and the cyan histogram, 
for $\varepsilon = 10^{-2}$, and two different values of
$\varepsilon_{\beta}$, namely $\varepsilon_{\beta} = 1$ and $\varepsilon_{\beta} = 10^{-2}$, 
respectively. One can notice that, 
as the conversion reactions in the sub-network $\mathcal{R}_{\beta}$
fire slower, the stationary $x_2$-marginal PMF converges to the desired
bi-modal form, and an implicit control of the residual species is achieved. 

\begin{figure}[!htbp]
\centerline{
\hskip 0mm
\includegraphics[width=0.35\columnwidth]{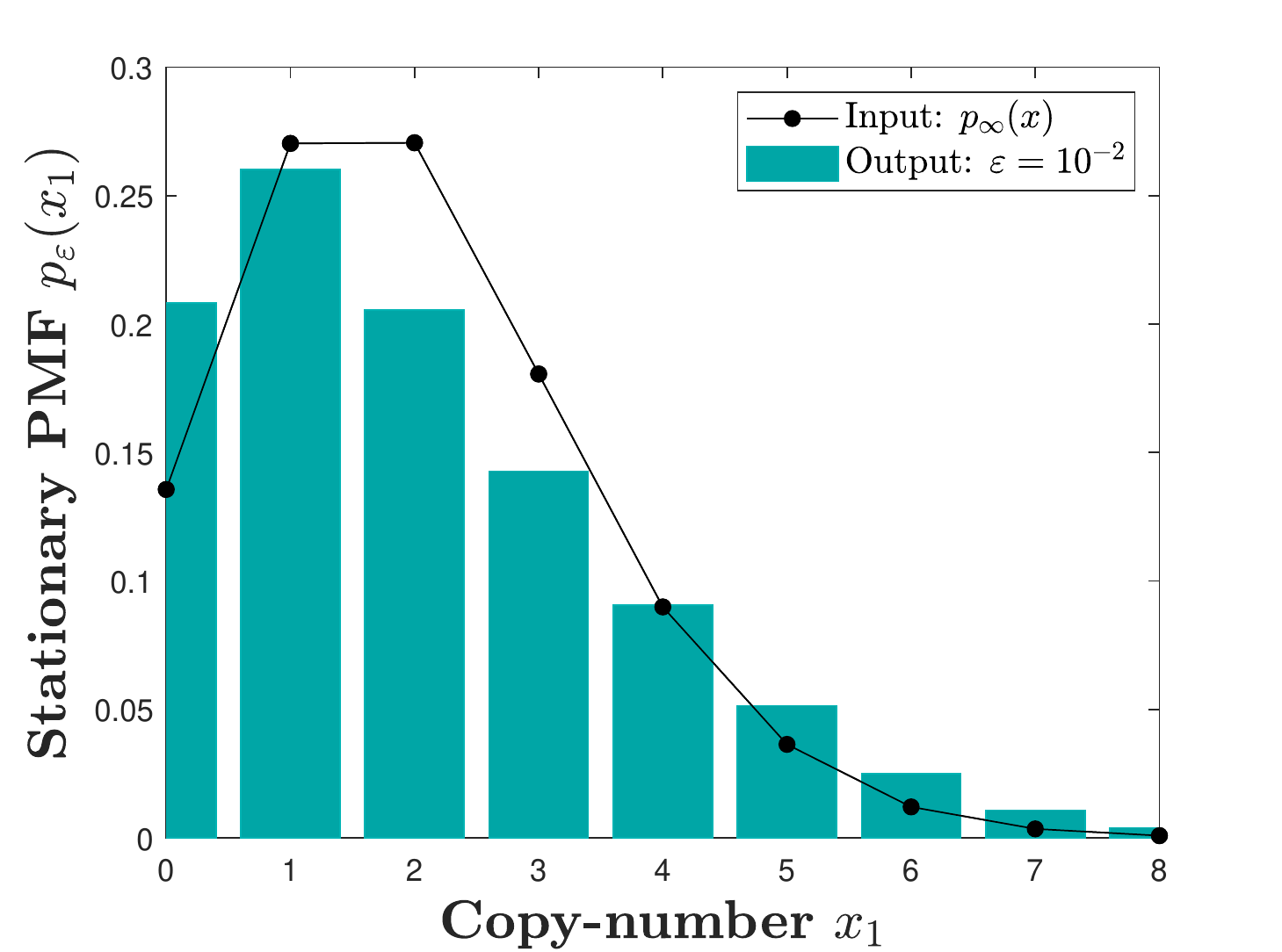}
\hskip 1.5cm
\includegraphics[width=0.35\columnwidth]{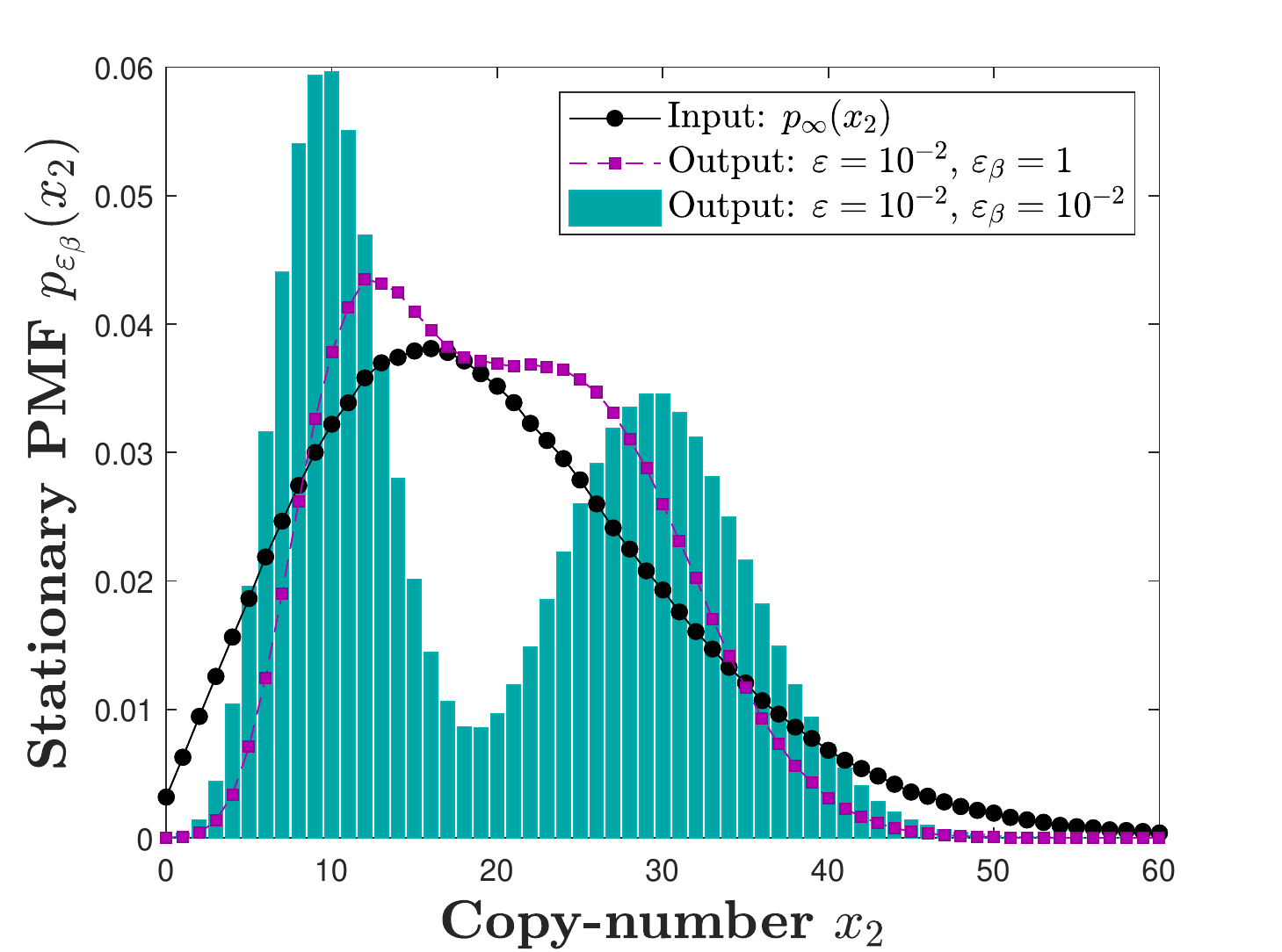}
}
\vskip -4.5cm
\leftline{\hskip 1.5cm (a) \hskip 6.8cm (b)} 
\vskip 3.8cm 
\caption{ 
Application of the lower-resolution control from 
{\rm Algorithm~\ref{alg:SM_algorithm}} on the input network~{\rm (\ref{eq:input_3})}
with $(\alpha_1, \alpha_2, \alpha_3, \alpha_4) = (2, 1, 10, 1)$.
{\it Panel {\rm (a)} shows the stationary $x_1$-marginal {\rm PMF}
of the input network~{\rm (\ref{eq:input_3})} as the interpolated
black dots, and of the output network~{\rm (\ref{eq:input_3})$\cup$(\ref{eq:ex3})},
with $\beta_{1,1} = 1$, $\beta_{1,2}/\beta_{2,1} = 1$, $\beta_{1,2} + \beta_{2,1} = \varepsilon_{\beta} = 1$,
$(\gamma_{0,1}, \gamma_{1,1}, \gamma_{2,1}) = (1, 1, 3)$ and $\varepsilon = 10^{-2}$, 
as the cyan histogram. 
Panel {\rm (b)} shows the stationary $x_2$-marginal {\rm PMF}
of the input network~{\rm (\ref{eq:input_3})} as the interpolated
black dots. Also shown as the interpolated purple squares, 
and the cyan histogram, are the $x_2$-marginal {\rm PMF}s
of the output network~{\rm (\ref{eq:input_3})$\cup$(\ref{eq:ex3})},
when $\varepsilon_{\beta} = 1$ and $\varepsilon_{\beta} = 10^{-2}$, 
respectively, with the rest of the parameters as in panel {\rm (a)}.
}
}
\label{fig:Genetic}
\end{figure}

\section{Proposed experimental implementation} \label{sec:experiments}
The stochastic morpher has been 
constructed to be experimentally implementable.
At the structural level, both the lower- and higher-resolution controller 
networks from Algorithm~\ref{alg:SM_algorithm}
consist of up-to second-order (and not higher-order)
reactions, with any reaction which does not involve an auxiliary species
being up-to first-order, allowing one to readily map the controller 
into nucleic-acid-based physical networks~\cite{DNAComputing1}. At the dynamical level,
as demonstrated in Sections~\ref{sec:example1}--\ref{sec:example3}, 
and established in Appendix~\ref{app:analysis}, 
the operational precision and robustness of the stochastic morpher depend on 
the appropriate orders of magnitude (time-scale separation) and ratios
of the rate coefficients of the underlying reactions, 
rather than specific values. This is suitable for implementation via the 
strand-displacement mechanism, 
where the rate coefficients may be varied over at least 
six orders of magnitude~\cite{Experiment5,Mismatches,RNAComputing},
and where the speed of the overall biochemical dynamics may also be further increased
with the addition of appropriate enzymes~\cite{Enzymes}.

In this section, we put forward a blueprint for an experimental implementation 
of the stochastic morpher from Algorithm~\ref{alg:SM_algorithm}. 
More specifically, as a proof-of-concept, 
we focus on physically realizing the 
controller~(\ref{eq:R_gamma_Poisson_2}) from Section~\ref{sec:example1}, which is 
capable of morphing the probability distribution of an
input network into a desired bi-modal form. 
One way to realize the stochastic morpher~(\ref{eq:R_gamma_Poisson_2}) 
is via a physical network which satisfies the following two properties:
(i) it contains two isomeric molecular species which may convert between themselves 
(realizing the controlling species $Y_1$ and $Y_2$, and the reaction $Y_1 \xrightleftharpoons[]{} Y_2$), 
each triggering a catalytic production of the target biochemical species $X$ at generally different rates 
(realizing the reactions $Y_1 \to Y_1 + X$ and $Y_2 \to Y_2 + X$), 
and
(ii) the network is integrated into an environment ensuring the presence of exactly one copy-number of the two isomeric
 species at a time (realizing $y_1 + y_2 = 1$, and eliminating the need for
the reaction $2 Y_1 \to Y_1$). 
Such conditions may be experimentally achieved via suitable DNA-strand-displacement-based physical 
networks, enclosed in appropriate compartments~\cite{Vesicles1,Vesicles2,Vesicles3,Vesicles5,Vesicles4}, 
allowing for an experimental observation and validation of the stochastic morpher.

More specifically, we put forward a DNA complex, known as a Holliday junction molecule, 
encapsulated in a nano-scale chamber, known as a small unilamellar vesicle (SUV),
 as a realization of the stochastic morpher~(\ref{eq:R_gamma_Poisson_2}), 
schematically displayed in Figure~\ref{fig:Experiment}(a). 
The SUV encapsulation is an experimentally demonstrated method for isolating and observing the 
dynamics of individual molecules, such as a single Holiday junction molecule, 
with minimal effect from the external environment~\cite{Vesicles5}.
The DNA Holliday junction complex consists of four double-stranded arms crossing
at a branch point, which is designed to be fixed (non-migratory) for the purpose of our implementation.
Let us note that the Holliday junction molecule with a migratory branch point is a central intermediate
during genetic recombination process~\cite{Experiment2}.
In the presence of magnesium ions, the Holliday junction can adopt two distinct orientations,
known as stacked conformational isomers~\cite{Experiment3}, allowing one to physically realize the reversible reaction 
$Y_1 \xrightleftharpoons[]{} Y_2$.
In order to control the rate coefficient of the interconversion reactions
$Y_1 \xrightleftharpoons[]{} Y_2$, and the catalytic production reactions
 $Y_1 \to Y_1 + X$ and $Y_2 \to Y_2 + X$, we put forward suitable DNA overhangs
involved in the associative (rather than dissociative)
 toehold activation~\cite{Experiment4}. 
More precisely, we tag three arms of the four-armed Holliday junction 
by extending them with distinct single-stranded DNA molecules (overhangs),
which pairwise associate with each other and activate distinct toeholds, 
shown as the grey regions attached to $Y_1$ and $Y_2$ in Figure~\ref{fig:Experiment}(b). 
The overhangs also partially hybridize with each other, 
forming duplexes next to the toeholds.
By controlling the length, and therefore the binding energy, of the 
duplexes in the associated DNA overhangs,
one can experimentally tune the interconversion rate between the 
two stacked conformational isomers. On the other hand, by controlling 
the length of the association-activated toeholds,
one can independently tune the rates of the two subsequent strand-displacement reactions
which produce the target species $X$ from a suitable precusor molecule, denoted
by $\bar{X}$ in Figure~\ref{fig:Experiment}(b).
Once a molecule of $X$ is produced, 
the isomeric species $Y_1$ and $Y_2$ can be recovered via 
appropriate strand-displacement reactions, 
ensuring an effective catalytic role of $Y_1$ and $Y_2$ 
in the overall reaction cascade.

\begin{figure}[!htbp]
\vskip 0.5cm
\centerline{
\hskip 0mm
\includegraphics[width=0.55\columnwidth]{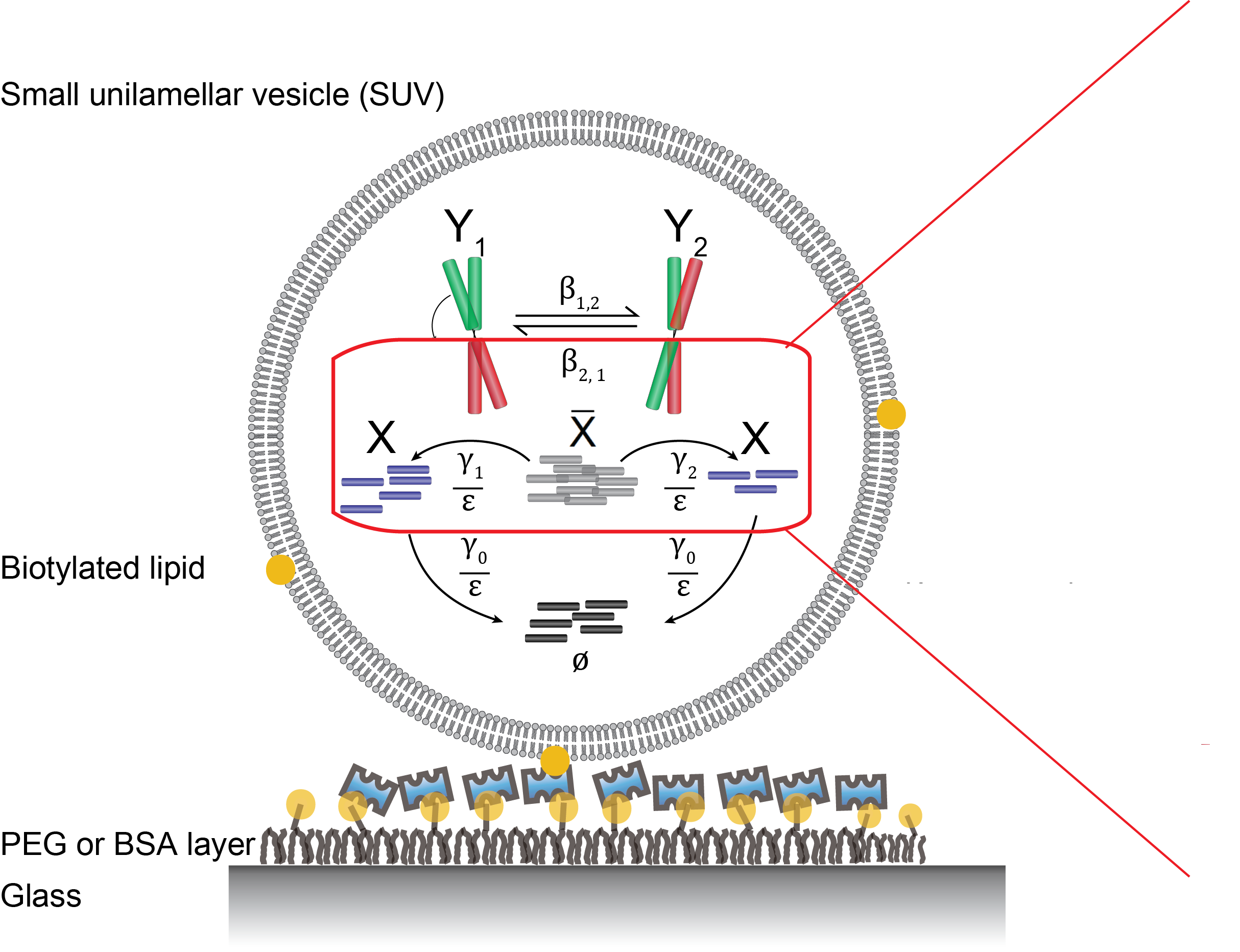}
\hskip 0cm
\includegraphics[width=0.55\columnwidth]{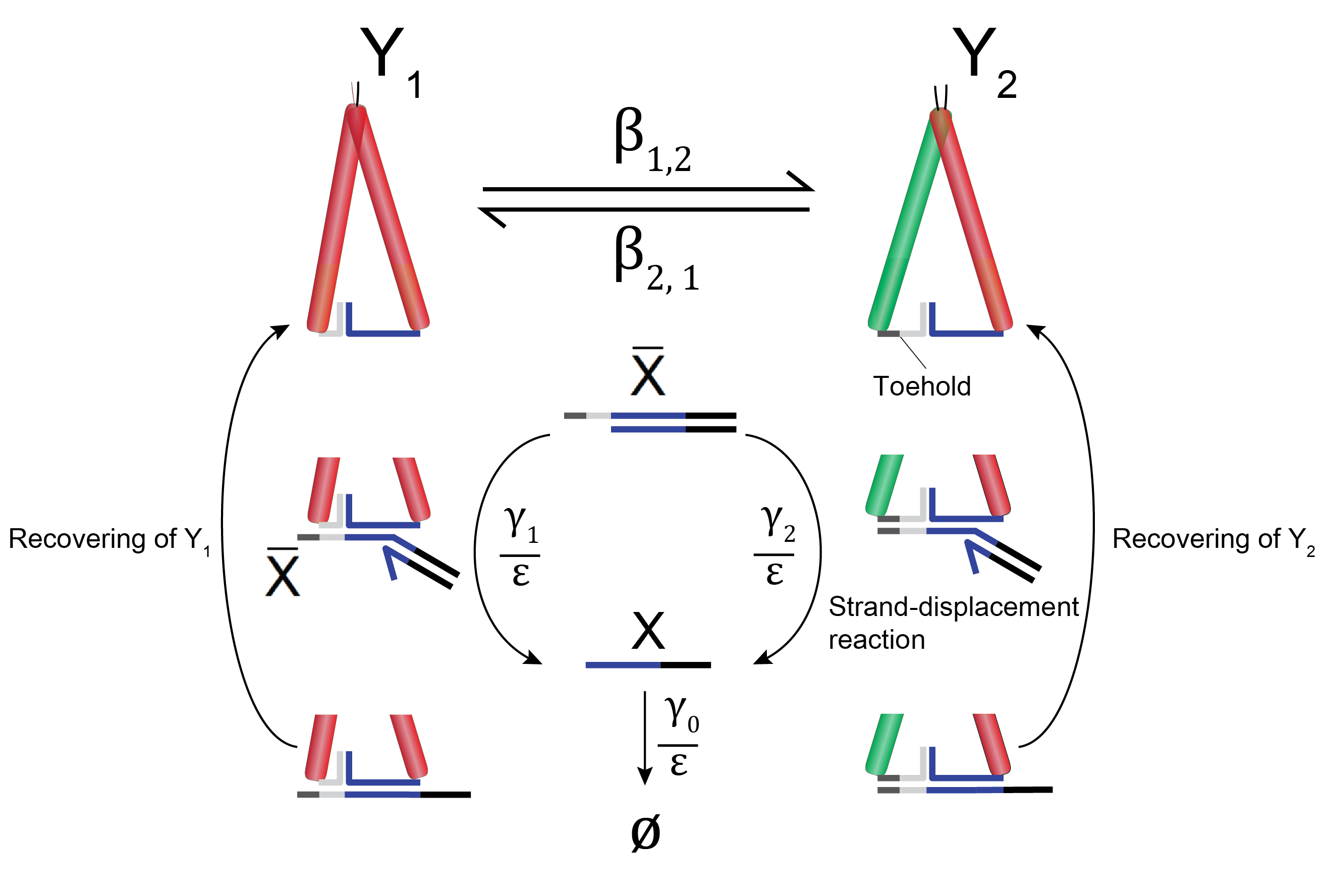}
}
\vskip -7.2cm
\leftline{\hskip 0.3cm (a) \hskip 8.2cm (b)} 
\vskip 6.8cm 
\caption{ 
Proposed experimental scheme for the stochastic morpher~(\ref{eq:R_gamma_Poisson_2}).
{\it Panel {\rm (a)} displays a small unilamellar vesicle 
({\rm SUV}), immobilized onto a {\rm PEG} or {\rm BSA} passivated surface.
The {\rm SUV} encloses a single {\rm DNA} Holliday junction molecule,
which switches between two distinct orientations, denoted by $Y_1$ and $Y_2$,
 and catalytically produces the target species $X$. 
Panel {\rm (b)} displays the underlying strand-displacement reactions, 
enclosed in the red quadrilateral in panel {\rm (a)}, in a greater detail. 
Here, three arms of the four-armed Holliday junction molecule are extended with 
single-stranded {\rm DNA} overhangs (grey and blue strands on $Y_1$ and $Y_2$). 
The DNA overhangs pairwise connect with each other, forming associated DNA overhangs, 
which consist of a toehold and a branch-migration domain, shown 
in grey and blue on $Y_1$ and $Y_2$, respectively, which are separated 
by a duplex. The association-activated toeholds
bind to an auxiliary double-stranded {\rm DNA}, called the precursor species
and denoted by $\bar{X}$, thus initiating the release of 
the target molecule $X$ via suitable strand-displacement reactions, 
whose rates are controlled by the toehold lengths.
Species $Y_1$ and $Y_2$ are then recovered 
using additional suitable strand-displacement reactions.
}
}
\label{fig:Experiment}
\end{figure}

\section{Discussion} \label{sec:discussion}
In this paper, we have introduced a biochemical controller, called 
\emph{stochastic morpher} and presented in Algorithm~\ref{alg:SM_algorithm} in Section~\ref{sec:intro}, 
which, when embedded into a black-box input reaction network with stochastic dynamics, 
overrides the input reactions and gradually transforms (morphs) 
the long-time probability mass function (PMF) of the target 
species into a desired predefined form. 
Morphing an input PMF consists of a sequence of intermediate probability 
distributions which increasingly resemble the desired output form, 
and which are parametrized by suitable time-scale separations between
 the stochastic morpher and the input network.
We have put forward two forms of the stochastic morpher:
the lower- and higher-resolution controllers, given 
by~(\ref{eq:R_beta_alg})$\cup$(\ref{eq:Poisson_control_alg})
and~(\ref{eq:R_beta_alg})$\cup$(\ref{eq:Kronecker_control_alg})
in Algorithm~\ref{alg:SM_algorithm}, respectively. 
The lower-resolution controller morphs a given input PMF into the space spanned by 
non-negative linear combinations of the Poisson distribution basis,
allowing one to achieve multi-modal PMFs (weak control), and
to manipulate average timing and the mode-switching pattern
of the underlying multi-stable sample paths (strong control). 
On the other hand, the higher-resolution controller
allows one to achieve arbitrary PMFs defined on bounded-domains.
General properties of the controller, including asymptotic robust perfect adaptation 
and convergence, are rigorously established in Appendix~\ref{app:analysis}
using singular perturbation theory. 
The stochastic morpher is envisaged for experimental
implementations involving cell-like vesicle in vitro,
and biological cells in vivo, where the lower species copy-numbers
can be exploited for gaining control over the dynamics
of desired biochemical species. 

In Section~\ref{sec:example1}, the lower-resolution stochastic morpher
has been applied on the one-species production-degradation 
test network~(\ref{eq:input_1}), in order to achieve a desired uni-, bi- and tri-modality, 
as well as to control the timing and pattern of the stochastic switching in the underlying sample paths, 
as shown in Figure~\ref{fig:Poisson}. We have also applied the higher-resolution
controller, in order to achieve a PMF concentrated at a single state (Kronecker-delta distribution), 
a uniform PMF, and a hybrid combination of Poisson and Kronecker-delta distributions, 
which is displayed in Figure~\ref{fig:Kronecker}.
In Section~\ref{sec:example2}, we have focused on the lower-resolution control
in greater detail, by considering the three-species test network~(\ref{eq:input_2}),
whose stationary PMF is bi-modal. The stochastic morpher has been utilized
to jointly explicitly control two, out of three, input species. In particular, the   
input marginal-PMF of the target species has been suitably redistributed, 
ensuring that the correlation between the two target species reverses 
from negative to positive, as shown in Figure~\ref{fig:Schlogl}(a)--(f). The error between the output
PMF and its target form has been numerically shown to decrease linearly with 
the underlying asymptotic parameter in Figure~\ref{fig:Schlogl}(g), 
in agreement with the theoretical convergence result established in Appendix~\ref{app:analysis}.
Finally, the dynamics of the remaining, residual, species in the output network
have been shown to match well with the theoretically derived
residual network~(\ref{eq:input_2_residual}), as demonstrated in Figure~\ref{fig:Schlogl}(h).
In Section~\ref{sec:example3}, we have applied Algorithm~\ref{alg:SM_algorithm}
on the two-species test network~(\ref{eq:input_3}), inspired by 
the process of gene expression, in order to obtain an implicit, rather than
explicit, control. More specifically, the stochastic morpher has been
used to explicitly influence a target species (interpreted as a transcribed mRNA),
 in order to implicitly control a residual species (interpreted as a translated protein).
It has been shown that such an approach can be used to induce multi-modality
into the residual species, with controllable relative separation between the modes, 
as displayed in Figure~\ref{fig:Genetic}. 
Finally, in Section~\ref{sec:experiments}, as a proof-of-concept,
we have put forward a blueprint for a DNA-strand-displacement-based experimental
realization of the stochastic morpher~(\ref{eq:R_gamma_Poisson_2}),
involving encapsulation of a suitable DNA complex inside
nano-scale vesicles~\cite{Experiment2,Experiment3,Vesicles5}. The experimental scheme, 
which may be seen as designing a synthetic cell, 
is outlined in Figure~\ref{fig:Experiment}, and will be pursued in a future publication.  
\newpage
Let us complete this section with four remarks. 
Firstly, the stochastic morpher has been inspired by a design principle from systems biology, called
noise-induced mixing~\cite{Me_Mixing}, and its operation requires that the 
interfacing network $\mathcal{R}_{\gamma}$, shown
in red in Figure~\ref{fig:Control_Theory} in Section~\ref{sec:intro},
fires sufficiently fast, with the controlling species, present at a single copy-number, 
being catalysts in $\mathcal{R}_{\gamma}$. 
Such fast-slow dynamics (time-scale separations) are ubiquitous and central in many natural biochemical 
networks from systems biology~\cite{Multiscale,CellCycle,Circadian,Bistability,Kepler,
Me_Homoclinic,Me_Limitcycles,Me_NAA,Me_Mixing}. 
It is then no surprise that time-scale separations
are necessary when abstract biochemical networks
are physically realized using molecules~\cite{DNAComputing1}, 
and that fast-slow dynamics also play a role in other biochemical controllers
developed in the literature~\cite{LeakyControl1,Khammash}.
Furthermore, control achieved via single (or, more generally, low and tightly regulated) 
molecular copy-numbers is also utilized by some natural biochemical systems, 
such as gene-regulatory networks involving expression of a single copy-number of a gene,
and it is then no surprise that such a property is also desirable in stochastic synthetic controllers.
Secondly, the lower- and higher-resolution 
stochastic morphers from Algorithm~\ref{alg:SM_algorithm} 
achieve PMFs that are linear combinations involving suitable bases. 
The bases are determined
 purely by the networks $\mathcal{R}_{\gamma}^{\mathcal{P}}$ and 
$\mathcal{R}_{\gamma}^{\delta}$, while the weights
in the underlying linear combinations are determined purely by
 $\mathcal{R}_{\beta}(\mathcal{Y})$, see also Appendix~\ref{app:long_time}.
The latter network, which is given by~(\ref{eq:R_beta_alg}), 
consists of the reaction $2 Y_1 \to Y_1$, and a periodic chain of
irreversible first-order conversion reactions between the controlling 
species $\mathcal{Y} = \{Y_1, Y_2, \ldots, Y_M\}$.
However, one can replace~(\ref{eq:R_beta_alg}) with 
 more general weakly reversible chains of first-order conversion reactions
between the species $\mathcal{Y}$, such as $2 Y_1 \to Y_1 
\xrightleftharpoons[]{} Y_2 \xrightleftharpoons[]{} \ldots \xrightleftharpoons[]{} Y_M$.
Such different choices of $\mathcal{R}_{\beta}$ only modify the functional form of the weights
in the PMFs achieved by Algorithm~\ref{alg:SM_algorithm}, and not the bases, 
and hence do not qualitatively change the weak control put forward in this paper.
An advantage of the particular choice~(\ref{eq:R_beta_alg}) is when strong control
is desired, since it allows one to deterministically control the mode-switching pattern
of the multi-stable sample paths underlying the achieved PMFs, 
see also Section~\ref{sec:example1} and Appendix~\ref{app:long_time}. 
Thirdly, the interfacing networks 
$\mathcal{R}_{\gamma}^{\mathcal{P}}(\mathcal{X}_{\tau}; \, \mathcal{Y})$ and
$\mathcal{R}_{\gamma}^{\delta}(\mathcal{X}_{\tau}, \mathcal{Z}; \, \mathcal{Y})$ 
from Algorithm~\ref{alg:SM_algorithm}
do not contain a feedback loop between the target species $\mathcal{X}_{\tau}$
and the controlling species $\mathcal{Y}$, making the stochastic morpher
an open-loop controller in the control theory language~\cite{ControlTheoryBook}. 
One can straightforwardly create a feedback loop between $\mathcal{X}_{\tau}$
and $\mathcal{Y}$ by e.g. introducing some of the target species
into the network $\mathcal{R}_{\gamma}$, thus resulting in a closed-loop controller. 
 Such feedback loops change only the weights of the PMFs achieved by the 
stochastic morpher, and not the bases in which the PMFs are expressed. 
However, in this case, the weights do not only depend on the rate 
coefficients from the network $\mathcal{R}_{\beta}$, but also on 
the coefficients from the networks $\mathcal{R}_{\gamma}^{\mathcal{P}}$
and  $\mathcal{R}_{\gamma}^{\delta}$. Such mixed dependence
may be seen as a disadvantage, since then e.g. the distribution
of the modes and the values of the achieved PMF at the modes 
cannot be controlled independently.
Finally, in the absence of an input network,
$\mathcal{R}_{\alpha} = \varnothing$, Algorithm~\ref{alg:SM_algorithm}
may be utilized to design, rather than control, biochemical reaction networks with predefined PMFs. 
Furthermore, noise-induced mixing can also be 
exploited for achieving other dynamical features, such multi-cyclicity (coexistence
of multiple stable oscillations), see~\cite{Me_Mixing}.

\section{Acknowledgements}
This work was supported by the EPSRC grant EP/P02596X/1.
Guy-Bart Stan also gratefully acknowledges the support of the 
UK EPSRC through the EPSRC Fellowship for Growth EP/M002187/1, 
and of the Royal Academy of Engineering 
through the Chair in Emerging Technology programme.
Thomas E. Ouldridge would also like to thank the Royal Society for
a University Research Fellowship.
\appendix

\section{Appendix: Background} \label{app:background}
\emph{Notation}. 
Set $\mathbb{R}$ is the space of real numbers, 
$\mathbb{R}_{\ge}$ the space of nonnegative real numbers,
 and $\mathbb{R}_{>}$ the space of positive real numbers. 
Similarly, $\mathbb{Z}$ is the space of integer numbers, 
$\mathbb{Z}_{\ge}$ the space of nonnegative integer numbers, 
and $\mathbb{Z}_{>}$ the space of positive integer numbers. 
Given two appropriate sequences $p(\cdot) : \mathbb{Z}_{\ge}^N \to \mathbb{R}$
and $u(\cdot) : \mathbb{Z}_{\ge}^N \to \mathbb{R}$,
the $l^1$-norm of $p(\mathbf{x})$ is
given by $ \| p \|_1 = \sum_{\mathbf{x}} |p(\mathbf{x})|$, 
while the $l^2$ (Hilbert sequence space) inner-product
of $p(\mathbf{x})$ and $u(\mathbf{x})$ is given by 
$\langle p, u \rangle = \sum_{\mathbf{x}} p(\mathbf{x}) u(\mathbf{x})$.
Euclidean row-vectors are denoted in boldface, 
$\mathbf{x} = (x_1, x_2, \ldots, x_N) \in \mathbb{R}^{N} = \mathbb{R}^{1 \times N}$.
Given a function $f(\cdot) : \mathbb{Z}_{\ge} \to \mathbb{R}$,
we define the product $\prod_{i = a}^b f(i) 
= f(a) f(a+1) \ldots f(b) \equiv 1$ if $a > b$. 
We also define $0^0 \equiv 1$. 
Given sets $\mathcal{A}_1$ and $\mathcal{A}_2$, 
their union is denoted by $\mathcal{A}_1 \cup \mathcal{A}_2$,
 their difference by $\mathcal{A}_1 \setminus \mathcal{A}_2$,
while their Cartesian-product, abusing the notation
slightly, by $\prod_{j = 1}^2 \mathcal{A}_j \equiv \mathcal{A}_1 \times \mathcal{A}_2$.
The empty set is denoted by $\emptyset$.

\subsection{Biochemical reaction networks} \label{app:CRNs}
In this paper, we consider reaction networks $\mathcal{R}_{\alpha}$
firing in well-mixed unit-volume reactors under mass-action kinetics, involving $N$ biochemical species 
$\mathcal{X} = \{X_1, X_2, \ldots, X_N\}$, and $A$ reactions, given by~\cite{Feinberg}
\begin{align}
\mathcal{R}_{\alpha}(\mathcal{X}): \; \; 
& & \sum_{l = 1}^N \nu_{j, l} X_l & \xrightarrow[]{\alpha_{j}} \sum_{l = 1}^N \bar{\nu}_{j, l}X_l,
\; \; \; \;
j \in \{1, 2, \ldots, A\}.
\label{eq:general_networks}
\end{align}
Here, $\alpha_j \in \mathbb{R}_{\ge}$ 
 is the \emph{rate coefficient} of the $j$th reaction, 
and $\nu_{j, l}, \bar{\nu}_{j, l} \in \mathbb{Z}_{\ge}$ 
are the \emph{reactant} and \emph{product stoichiometric coefficients}
of the species $X_l$ in the $j$th reaction, respectively.
When all of the reactant (product) stoichiometric coefficients
are equal to zero in a reaction, the reactant (product)
is the zero-species, denoted by $\varnothing$, 
which represents species which are not explicitly modelled. 

When convenient, we indicate dependence of a reaction network on
species of interest, e.g. to emphasize that $\mathcal{R}_{\alpha}$
 involves species $\mathcal{X}$, we have written $\mathcal{R}_{\alpha} = 
\mathcal{R}_{\alpha}(\mathcal{X})$ in~(\ref{eq:general_networks}).
We collect all of the rate coefficients into the vector 
$\boldsymbol{\alpha} = (\alpha_1, \alpha_2, \ldots, \alpha_A) \in \mathbb{R}_{\ge}^A$.
We follow the convention of denoting the (vector of the) rate coefficients 
of the reactions underlying a reaction network using the same letter as the network subscript. 
In addition, fixing a reaction coefficient to zero is
defined as deleting the corresponding reaction from
the underlying network. 
Furthermore, we define the \emph{reactant} and \emph{product complexes} 
of the $j$th reaction from~(\ref{eq:general_networks}) as the vectors 
$\boldsymbol{\nu}_j = (\nu_{j,1}, \nu_{j,2}, \ldots, \nu_{j,N}) \in \mathbb{Z}_{\ge}^N$
and $\bar{\boldsymbol{\nu}}_j = (\bar{\nu}_{j,1}, \bar{\nu}_{j,2}, \ldots, \bar{\nu}_{j,N}) 
\in \mathbb{Z}_{\ge}^N$, respectively, and, abusing the notation slightly, 
denote the $j$th reaction by $(\boldsymbol{\nu}_j \to \bar{\boldsymbol{\nu}}_j) 
\in \mathcal{R}_{\alpha}$, when convenient. 
Reactions with the same reactant and product complexes, 
$(\boldsymbol{\nu}_j \to \boldsymbol{\nu}_j)$
(such as the one obtained by taking $M = 1$ in~(\ref{eq:R_beta_alg}) 
from Algorithm~\ref{alg:SM_algorithm}) are redundant, and
 are deleted from reaction networks.
We denote two irreversible reactions $(\boldsymbol{\nu} \to \bar{\boldsymbol{\nu}}) 
\in \mathcal{R}_{\alpha}$
and
$(\bar{\boldsymbol{\nu}} \to \boldsymbol{\nu}) \in \mathcal{R}_{\alpha}$
 jointly as the single reversible reaction
$(\boldsymbol{\nu} \xrightleftharpoons[]{} \bar{\boldsymbol{\nu}}) \in\mathcal{R}_{\alpha}$, 
when convenient.
\emph{Reaction vector} of the $j$th reaction is defined as 
$\Delta \mathbf{x}_j = (\bar{\boldsymbol{\nu}}_j - \boldsymbol{\nu}_j) \in \mathbb{Z}^N$.
The \emph{order of reaction} $(\boldsymbol{\nu}_j \to \bar{\boldsymbol{\nu}}_j) \in \mathcal{R}$
is given by $\langle \mathbf{1}, \boldsymbol{\nu}_j\rangle \in \mathbb{Z}_{\ge}$, 
with $\mathbf{1} = (1, 1, \ldots, 1) \in \mathbb{Z}^N$. 
The \emph{order of reaction network} $\mathcal{R}_{\alpha}$ is given by the order of
its highest-order reaction.

\subsection{The stochastic model of reaction networks} \label{app:stochastic_model}
We consider reaction networks with discrete species counts, and stochastic dynamics. 
Let $\mathbf{x} = (x_1, x_2, \ldots, x_N) \in \mathbb{Z}_{\ge}^N$
 denote the discrete state-vector of the species 
$\mathcal{X} = \{X_1, X_2, \ldots, X_N\}$ from~(\ref{eq:general_networks}), 
where element $x_l \in \mathbb{Z}_{\ge}$ denotes the 
copy-number values of the species $X_l$. Abusing the notation slightly, 
we denote the copy-numbers of the biochemical species $X_l$ as a function of time
using the same symbol, $X_l(t)$, where $t \in \mathbb{R}_{\ge}$ is the time-variable.
A suitable stochastic description models the time-evolution of the species copy-number
vector as a continuous-time discrete-space Markov chain~\cite{GillespieDerivation}.
The underlying \emph{probability-mass function} (PMF)
satisfies the partial difference-differential equation, 
called the \emph{chemical master equation} (CME)~\cite{VanKampen,Gardiner,David1}, given by
\begin{align}
\frac{\partial}{\partial t} p(\mathbf{x},t)   =  \mathcal{L}_{\alpha} p(\mathbf{x},t) & = 
\sum_{j = 1}^A (E_{\mathbf{x}}^{-\Delta \mathbf{x}_j} - 1)
 \big(\lambda_j(\mathbf{x}) p(\mathbf{x},t) \big), \label{eq:CME_general}
\end{align}
where $p(\mathbf{x},t)$ is the PMF,
i.e. the probability that the copy-number vector 
at time $t > 0$ is given by $\mathbf{x}  \in \mathbb{Z}_{\ge}^N$.
Here, the \emph{step operator} $E_{\mathbf{x}}^{-\Delta \mathbf{\mathbf{x}}} = 
\prod_{l = 1}^N E_{x_l}^{-\Delta x_l}$ is such that 
$E_{\mathbf{x}}^{-\Delta \mathbf{\mathbf{x}}} p(\mathbf{x},t)
 = p(\mathbf{x} - \Delta \mathbf{x},t)$.
 The function $\lambda_j(\mathbf{x})$ is the propensity (rate) function of the $j$-th reaction, 
and is given by
\begin{align}
\lambda_{j}(\mathbf{x}) & = \alpha_{j} \mathbf{x}^{\underline{\boldsymbol{\nu}_j}} 
\equiv \alpha_j \prod_{l = 1}^N x_l^{\underline{\nu_{j, l}}}, \; \; \; \; \mathbf{x} \in \mathbb{Z}_{\ge}^N, 
\label{eq:propensity_general}
\end{align}
where $\alpha_{j} \in \mathbb{R}_{\ge}$ is the rate coefficient of the reaction 
$(\boldsymbol{\nu}_j \to \bar{\boldsymbol{\nu}}_j) \in \mathcal{R}_{\alpha}$.
Here,
$x^{\underline{\nu}}  = x (x - 1) (x - 2) \ldots (x - \nu + 1) \in \mathbb{Z}_{\ge}$
for $x, \nu \in \mathbb{Z}_{\ge}$, denotes the $\nu$-th factorial power of $x$,
 with the convention that $x^{\underline{0}} \equiv 1$ for all $x \in \mathbb{Z}_{\ge}$.

The linear operator $\mathcal{L}_{\alpha}$ from~(\ref{eq:CME_general}) 
is called the \emph{forward operator} of the network $\mathcal{R}_{\alpha}(\mathcal{X})$, 
given by~(\ref{eq:general_networks}). 
The $l^2$-adjoint operator of $\mathcal{L}_{\alpha}$,
denoted by $\mathcal{L}_{\alpha}^*$ and called the \emph{backward operator}~\cite{Pavliotis}, 
is given by
\begin{align}
\mathcal{L}_{\alpha}^{*} u(\mathbf{x}) & =  
\sum_{j = 1}^A \lambda_j(\mathbf{x}) (E_{\mathbf{x}}^{+\Delta \mathbf{\mathbf{x}}_j} - 1) 
u(\mathbf{x}). \label{eq:backwardoperator}
\end{align}

A function $p(\mathbf{x})$ satisfying $\mathcal{L}_{\alpha} p(\mathbf{x}) = 0$, 
i.e. $p(\mathbf{x}) \in \mathcal{N}(\mathcal{L}_{\alpha})$, where
$\mathcal{N}(\cdot)$ denotes the null-space of an operator, 
is called a \emph{stationary} PMF of the network $\mathcal{R}_{\alpha}(\mathcal{X})$.

\subsection{Biochemical control} \label{app:biochemical_control}
The aim of biochemical control theory~\cite{Control1,Control2,Control3,
Me_NAA,Me_Homoclinic,Khammash}
 is to suitably modify a given reaction network in order to desirably influence the 
dynamics of a subset of the underlying species.
It is implicitly assumed that the given network, which we wish to control, has at most partially
known structure, and some of its rudimentary 
dynamical features, such as an averaged (mean) behavior or the time-scales at which 
the underlying reactions fire, may also be known.
\begin{definition} [\textbf{Black-box}]  \label{def:blackbox} 
Reaction network $\mathcal{R}_{\alpha}(\mathcal{X})$,
whose fixed structure (the set of reactions underlying the network)
 and the induced dynamics are at most partially known, 
is called a \emph{black-box} network.
\end{definition}
In this paper, it is assumed we are given a 
black-box reaction network under mass-action kinetics,
denoted by $\mathcal{R}_{\alpha} = \mathcal{R}_{\alpha}(\mathcal{X})$, 
called an \emph{input} (uncontrolled) network,
which depends on $N$ biochemical species $\mathcal{X} = \{X_1, X_2, \ldots, X_N\}$, 
and has the form~(\ref{eq:general_networks}).
The input species are partitioned into $\mathcal{X} = \mathcal{X}_{\tau} \cup \mathcal{X}_{\rho}$, 
where $\mathcal{X}_{\tau}= \{X_1, X_2, \ldots, X_n\}$, $1 < n \le N$, are the \emph{target} species,
 whose dynamics we wish to explicitly control.
On the other hand, the remaining species $\mathcal{X}_{\rho} = \mathcal{X} 
\setminus \mathcal{X}_{\tau} = 
\{X_{n+1}, X_{n+2}, \ldots, X_N\} $,
are the \emph{residual} species, whose dynamics may only be implicitly, but
not explicitly, controlled.

In order to control a black-box 
input network $\mathcal{R}_{\alpha}(\mathcal{X})$, 
an auxiliary mass-action reaction network is embedded, 
called a \emph{controller} network,
and denoted by $\mathcal{R}_{\beta,\gamma} = \mathcal{R}_{\beta,\gamma}
(\mathcal{X}_{\tau},\mathcal{Y}, \mathcal{Z})$.
The resulting composite network is denoted by $\mathcal{R}_{\alpha,\beta,\gamma} = 
\mathcal{R}_{\alpha} \cup \mathcal{R}_{\beta,\gamma}$, 
and called an \emph{output} (controlled) network.
Here, generally two sets of auxiliary species are introduced 
by the controller:  $\mathcal{Y} = \{Y_1, Y_2, \ldots, Y_M\}$, called the \emph{controlling} species, 
and $\mathcal{Z}$, called the \emph{mediating} species.
The controller can be decomposed into two
sub-networks, $\mathcal{R}_{\beta,\gamma}(\mathcal{X}_{\tau},\mathcal{Y}, \mathcal{Z})
= \mathcal{R}_{\beta}(\mathcal{Y}) \cup \mathcal{R}_{\gamma}(\mathcal{X}_{\tau},\mathcal{Y}, \mathcal{Z})$, 
where $\mathcal{R}_{\beta} = \mathcal{R}_{\beta}(\mathcal{Y})$ specifies how
the controlling species interact among themselves, while $\mathcal{R}_{\gamma}
= \mathcal{R}_{\gamma}(\mathcal{X}_{\tau},\mathcal{Y}, \mathcal{Z})$ specifies how the controller
is interfaced with the input network. More precisely, the interfacing
network $\mathcal{R}_{\gamma}(\mathcal{X}_{\tau},\mathcal{Y}, \mathcal{Z})$
describes how the controlling species $\mathcal{Y}$
interact with the target species $\mathcal{X}_{\tau}$, either directly ($\mathcal{Z} = \emptyset$),
or indirectly via the mediating species ($\mathcal{Z} \ne \emptyset$). 
Control of an input network in the absence
of the mediating species is schematically depicted
 in Figure~\ref{fig:Control_Theory} in the main text. 

We denote the vectors of the rate coefficients from the networks 
$\mathcal{R}_{\alpha}$, $\mathcal{R}_{\beta}$ and $\mathcal{R}_{\gamma}$
 by $\boldsymbol{\alpha} \in \mathbb{R}_{\ge}^A$,  
$\boldsymbol{\beta} \in \mathbb{R}_{\ge}^B$ and $\boldsymbol{\gamma} \in \mathbb{R}_{\ge}^C$, 
respectively.
It is assumed that $\boldsymbol{\alpha}$ is a given constant (fixed) vector (since the input network
is a black-box), while $\boldsymbol{\beta}$ and $\boldsymbol{\gamma}$
are (variable) parameters (since we assume the kinetics of the controller are tunable). 
Furthermore, for fixed initial conditions, we denote the stationary marginal-PMF
of the target species from the output network $\mathcal{R}_{\alpha,\beta,\gamma}$, 
assumed to exist, by $p(\mathbf{x}_{\tau}) =  
p(\mathbf{x}_{\tau}, \boldsymbol{\beta}, \boldsymbol{\gamma};
 \, \boldsymbol{\alpha})$, where
 $\mathbf{x}_{\tau} = (x_1, x_2, \ldots, x_n) \in \mathbb{Z}_{\ge}^n$ 
are the copy-numbers of the target species $\mathcal{X}_{\tau}$.

\subsubsection*{Controllability and robustness}
The objective of a controller network $\mathcal{R}_{\beta,\gamma}$, which is 
embedded into an input network $\mathcal{R}_{\alpha}$, is to ensure
 that the target species $\mathcal{X}_{\tau}$ 
from the resulting output network $\mathcal{R}_{\alpha,\beta,\gamma}$
have suitably controlled stochastic dynamics. 
Control may be sought at the level of the PMF (which we call \emph{weak control}), 
or at the level of the underlying sample paths (which we call \emph{strong control}).
In this paper, we focus predominantly on the weak control over the stationary (long-time) dynamics, 
which is often of most practical importance, 
and is achieved by manipulating the properties of 
the stationary marginal PMF of the target species,
$p(\mathbf{x}_{\tau}, \boldsymbol{\beta}, \boldsymbol{\gamma}; 
\, \boldsymbol{\alpha})$, 
such as the underlying means and modes.
In particular, we consider linear functionals of the form 
$\mathbb{E} f_{\mathbf{x}_{\tau}} = 
\mathbb{E} f_{\mathbf{x}_{\tau}} (\boldsymbol{\beta}, \boldsymbol{\gamma}; \, \boldsymbol{\alpha})
 = \sum_{\mathbf{x}_{\tau}} f_{\mathbf{x}_{\tau}}(\mathbf{x}_{\tau}) p(\mathbf{x}_{\tau}, \boldsymbol{\beta}, \boldsymbol{\gamma};
 \, \boldsymbol{\alpha})$, where $f_{\mathbf{x}_{\tau}} : \mathbb{Z}_{\ge}^n \to \mathbb{R}$
is a suitable function of the target species, and $\mathbb{E} \cdot$ is the expectation operator
with respect to the PMF $p(\mathbf{x}_{\tau}, \boldsymbol{\beta}, \boldsymbol{\gamma};
 \, \boldsymbol{\alpha})$.
\begin{definition} [\textbf{Controllability}]  \label{def:control} 
Consider an input network $\mathcal{R}_{\alpha}(\mathcal{X})$, 
 a controller $\mathcal{R}_{\beta,\gamma}(\mathcal{X}_{\tau},\mathcal{Y}, \mathcal{Z})$, 
and the corresponding output network 
$\mathcal{R}_{\alpha,\beta,\gamma}(\mathcal{X},\mathcal{Y}, \mathcal{Z})
= \mathcal{R}_{\alpha}(\mathcal{X}) \cup 
\mathcal{R}_{\beta,\gamma}(\mathcal{X}_{\tau},\mathcal{Y}, \mathcal{Z})$, where 
$\mathcal{X}_{\tau} = \{X_1, X_2, \ldots, X_n\} \subseteq \mathcal{X}$ are the target species. 
The range of the stationary
statistic $\mathbb{E} f_{\mathbf{x}_{\tau}} 
(\boldsymbol{\beta}, \boldsymbol{\gamma}; \, \boldsymbol{\alpha})$ for each fixed 
$\boldsymbol{\alpha} \in \mathbb{R}_{\ge}^A$ is denoted by
$\mathcal{S}_{f}^{\alpha} \subseteq \mathbb{R}$, i.e.
$\mathbb{E} f_{\mathbf{x}_{\tau}}( \cdot, \cdot \, ; \, \boldsymbol{\alpha})
 : \mathbb{R}_{\ge}^B \times \mathbb{R}_{\ge}^C
\to \mathcal{S}_{f}^{\alpha} \subseteq \mathbb{R}$, and 
is called the \emph{set of admissible values} of $\mathbb{E} f_{\mathbf{x}_{\tau}} 
(\boldsymbol{\beta}, \boldsymbol{\gamma}; \, \boldsymbol{\alpha})$. 
For a fixed $\boldsymbol{\alpha} \in \mathbb{R}_{\ge}^A$, given a target value $f^* \in \mathbb{R}$, 
the statistic $\mathbb{E} f_{\mathbf{x}_{\tau}}
(\boldsymbol{\beta}, \boldsymbol{\gamma}; \, \boldsymbol{\alpha})$ is said to be \emph{controllable}
if $f^* \in \mathcal{S}_{f}^{\alpha}$.
\end{definition}
\noindent The set of admissible values $\mathcal{S}_{f}^{\alpha}$ 
depends on both the structure and rate coefficients $\boldsymbol{\alpha}$ 
of the input network $\mathcal{R}_{\alpha}$. 
Given a suitable function $f_{\mathbf{x}_{\tau}}(\mathbf{x}_{\tau})$, the goal is to find a controller with an appropriate structure, 
and suitably tuned rate coefficients, which manipulates the stationary $\mathbf{x}_{\tau}$-marginal PMF
so that  the target value $f^*$ lies within the range of the underlying statistic of interest.

We are interested in the output networks which are experimentally
implementable, which 
imposes a set of constraints on the controllers, some of which are captured in the following definition.
\begin{definition} [\textbf{Robustness}]  \label{def:robust} 
A controller $\mathcal{R}_{\beta,\gamma}(\mathcal{X}_{\tau},\mathcal{Y}, \mathcal{Z})$
is said to be \emph{robust}, when embedded into an input network $\mathcal{R}_{\alpha}(\mathcal{X})$,
 if both of the following two conditions are satisfied:
\begin{enumerate}
\item[{\rm (a)}] \textbf{Robustness with respect to the initial conditions}. 
The stationary marginal-{\rm PMF} of the target species $\mathcal{X}_{\tau}$,
from the output network $\mathcal{R}_{\alpha,\beta,\gamma}(\mathcal{X},\mathcal{Y}, \mathcal{Z})$, 
exists and is unique, i.e. it is independent of the initial conditions
for the species $\mathcal{X}$, $\mathcal{Y}$ and $\mathcal{Z}$, 
for a given fixed $\boldsymbol{\alpha}$.
\item[{\rm (b)}] \textbf{Robustness with respect to the input coefficients}. 
The controlled stationary statistic $\mathbb{E} f_{\mathbf{x}_{\tau}}$ of the target species $\mathcal{X}_{\tau}$, 
from the output network $\mathcal{R}_{\alpha,\beta,\gamma}(\mathcal{X},\mathcal{Y}, \mathcal{Z})$,
 does not explicitly depend on the parameters $\boldsymbol{\alpha}$ from 
the input network $\mathcal{R}_{\alpha}(\mathcal{X})$ (possibly only in an asymptotic limit 
of some of the rate coefficients from the controller), i.e. 
$\mathbb{E} f_{\mathbf{x}_{\tau}} = \mathbb{E} f_{\mathbf{x}_{\tau}}
(\boldsymbol{\beta}, \boldsymbol{\gamma})$.
\end{enumerate}
\end{definition}
Definition~\ref{def:robust}(a) demands that the marginal stochastic process, 
underlying the dynamics of the target species, is ergodic.
On the other hand, Definition~\ref{def:robust}(b) ensures that the 
desired stationary statistics of the species $\mathcal{X}_{\tau}$ depend parametrically 
only on the rate coefficients appearing in the controller, which 
are experimentally tunable, allowing one to treat the underlying input network as a black-box. 
Let us note that if the stationary $\mathbf{x}_{\tau}$-marginal PMF 
$p(\mathbf{x}_{\tau}, \boldsymbol{\beta}, \boldsymbol{\gamma};
 \, \boldsymbol{\alpha})$ is independent of
the initial conditions, then the same is true for the stationary statistic 
$\mathbb{E} f_{\mathbf{x}_{\tau}} = \sum_{\mathbf{x}_{\tau}} 
f_{\mathbf{x}_{\tau}}(\mathbf{x}_{\tau}) p(\mathbf{x}_{\tau}, \boldsymbol{\beta}, \boldsymbol{\gamma};
 \, \boldsymbol{\alpha})$, i.e. the stability
condition (a) from Definition~\ref{def:robust} also ensures
stability of the underlying statistics. We allow non-uniqueness of the 
stationary marginal-PMF for the residual species $\mathcal{X}_{\rho}$,
 thereby including a larger class of input networks $\mathcal{R}_{\alpha}$ 
into considerations, see also Example~\ref{ex:ergodicity} 
in Appendix~\ref{sec:residual_networks}.

If condition (b) from Definition~\ref{def:robust} is satisfied independently 
of the values of the rate coefficients from the controller, 
then the corresponding statistic is said to 
display \emph{robust perfect adaptation}~\cite{Adaptation,CellSignal}. 
On the other hand, if condition (b) from Definition~\ref{def:robust}
is satisfied only in an asymptotic limit of some of the rate coefficients from the controller,
then we say that the statistic displays \emph{asymptotic robust perfect adaptation}. 
Note that (asymptotic) robust perfect adaptation is experimentally
implementable, and allows one to treat the input network as a black-box.
This is in contrast to non-robust perfect adaptation, which
requires fine-tuning of the rate coefficients $\boldsymbol{\beta}$ and $\boldsymbol{\gamma}$ of the
controller to specific values, which depend on the unknown rate coefficients $\boldsymbol{\alpha}$ 
of the black-box input network $\mathcal{R}_{\alpha}$.

\section{Appendix: Dynamical analysis of the stochastic morpher} \label{app:analysis}
In this paper, we consider controllers of the form
$\mathcal{R}_{\beta,\gamma} = \mathcal{R}_{\beta,\gamma}
(\mathcal{X}_{\tau},\mathcal{Y}) = \mathcal{R}_{\beta,\gamma}(\mathcal{Y}) \cup 
\mathcal{R}_{\gamma}^{\varepsilon}(\mathcal{X}_{\tau}; \, \mathcal{Y})$, giving
rise the the output networks
\begin{align}
\mathcal{R}_{\alpha,\beta,\gamma}(\mathcal{X}, \mathcal{Y}) & = 
\mathcal{R}_{\alpha}(\mathcal{X})
\cup 
\mathcal{R}_{\beta}(\mathcal{Y})
\cup 
\mathcal{R}_{\gamma}^{\varepsilon}(\mathcal{X}_{\tau}, \mathcal{Y}). \label{eq:R_output}
\end{align}
The interfacing network $\mathcal{R}_{\gamma}^{\varepsilon} = 
\mathcal{R}_{\gamma}^{\varepsilon}(\mathcal{X}_{\tau}; \, \mathcal{Y})$ from~{\rm \ref{eq:R_output}} 
is assumed to take the following separable form
\begin{align}
\mathcal{R}_{\gamma}^{\varepsilon}(\mathcal{X}_{\tau}; \, \mathcal{Y}) & = 
\mathcal{R}_{\gamma_0}^{\varepsilon}(\mathcal{X}_{\tau}; \, \varnothing) 
\mathop{\bigcup}_{i = 1}^M 
\mathcal{R}_{\gamma_i}^{\varepsilon}(\mathcal{X}_{\tau}; \, Y_i),
\label{eq:R_gamma}
\end{align}
where the sub-network $\mathcal{R}_{\gamma_0}^{\varepsilon} = 
\mathcal{R}_{\gamma_0}^{\varepsilon}(\mathcal{X}_{\tau}; \, \varnothing)$ 
depends on the species $\mathcal{X}_{\tau}$ and is independent of $\mathcal{Y}$, while 
each factor $\{\mathcal{R}_{\gamma_i}^{\varepsilon} = 
\mathcal{R}_{\gamma_i}^{\varepsilon}(\mathcal{X}_{\tau}; \, \mathcal{Y}_i)\}_{i = 1}^M$
consists exclusively of reactions which are catalyzed by
the species $Y_i$. More precisely, the $j$th reaction from the sub-network
$\mathcal{R}_{\gamma_i}^{\varepsilon}$, denoted by $r_{i, j}$, is given by
\begin{align}
r_{0, j}: \; & & \sum_{l = 1}^n \nu_{0, j, l} X_l & 
\xrightarrow[]{\gamma_{0, j}/\varepsilon} \sum_{l = 1}^n \bar{\nu}_{0, j, l}X_l,
\nonumber \\
r_{i, j}: \; & & Y_i 
+ \sum_{l = 1}^n \nu_{i, j, l} X_l & 
\xrightarrow[]{\gamma_{i,j}/\varepsilon}  Y_i 
+ \sum_{l = 1}^n \bar{\nu}_{i, j, l} X_l, 
\; \; \; \; 
0 < \varepsilon \ll 1,
\; \; \; \; 
\textrm{for } i \in \{1, 2, \ldots, M\}.
\label{eq:R_gamma_detail}  
\end{align} 
The rate coefficients $\boldsymbol{\alpha}$, 
$\boldsymbol{\beta}$,  and $\varepsilon^{-1} \boldsymbol{\gamma}$ of the sub-networks 
$\mathcal{R}_{\alpha}$, $\mathcal{R}_{\beta}$, 
and $\mathcal{R}_{\gamma}^{\varepsilon}$, respectively, 
are assumed to be of order one, $\boldsymbol{\alpha}, \boldsymbol{\beta}, \boldsymbol{\gamma}
= \mathcal{O}(1)$, with respect to the small asymptotic parameter $0 < \varepsilon \ll 1$.
In other words, we assume that the network $\mathcal{R}_{\gamma}^{\varepsilon}$,
which interfaces the controlling and the target species,
fires much faster than the input network and the network governing the controlling species,
$\mathcal{R}_{\alpha}$ and $\mathcal{R}_{\beta}$, respectively. 

Note that the target species $\mathcal{X}_{\tau}$ interact directly
with each other in the sub-network 
$\mathcal{R}_{\gamma_0}^{\varepsilon}(\mathcal{X}_{\tau}; \, \varnothing)$, 
and that the auxiliary species $\mathcal{Y}$ are interfaced directly with $\mathcal{X}_{\tau}$, 
as specified by the sub-network $\mathop{\bigcup}_{i = 1}^M 
\mathcal{R}_{\gamma_i}^{\varepsilon}(\mathcal{X}_{\tau}; \, Y_i)$. 
As outlined in Appendix~\ref{app:biochemical_control}, we also consider a generalized case
when the mediating species $\mathcal{Z}$ are present.
Such species may play a role of a buffer for indirect interactions between the species $\mathcal{X}_{\tau}$, 
or may serve as intermediate species, propagating the action of the controlling species 
$\mathcal{Y}$ onto the target species $\mathcal{X}_{\tau}$. 
In this paper, we assume that the mediating species, 
when present, are sufficiently fast, and characterized by the dimensionless time-scale $\mu$, 
with $0 < \mu \ll \varepsilon \ll 1$, and with the interfacing network given by
$\mathcal{R}_{\gamma}^{\mu, \varepsilon}
= \mathcal{R}_{\gamma}^{\mu,\varepsilon}(\mathcal{X}_{\tau}, \mathcal{Z}; \, \mathcal{Y})
= \mathcal{R}_{\gamma_0}^{\mu,\varepsilon}(\mathcal{X}_{\tau}, \mathcal{Z}; \, \varnothing) 
\mathop{\bigcup}_{i = 1}^M 
\mathcal{R}_{\gamma_i}^{\mu, \varepsilon}(\mathcal{X}_{\tau}, \mathcal{Z}; \, Y_i)$.
More specifically, we consider the interfacing network~(\ref{eq:Kronecker_control_alg})
from Algorithm~\ref{alg:SM_algorithm}, for which $\mathcal{R}_{\gamma}^{\mu, \varepsilon} \to 
\mathcal{R}_{\gamma}^{\varepsilon}$ as $\mu \to 0$, i.e.
the indirect coupling reduces to 
an effective direct coupling in the limit $\mu \to 0$, 
 as we now establish. 

\begin{theorem} \label{theorem:Kronecker_mu} 
\textit{Consider the network 
$\mathcal{R}_{\gamma}^{\delta}(\mathcal{X}_{\tau},  \mathcal{Z}; \, \mathcal{Y})
= \mathcal{R}_{\gamma_0}^{\mu,\varepsilon,\sigma}(\mathcal{X}_{\tau},  \mathcal{Z}; \, \varnothing) 
\cup_{i = 1}^M \mathcal{R}_{\gamma_i}^{\mu,\varepsilon,\sigma}(\mathcal{Z}; \, Y_i)$
 given by~{\rm \ref{eq:Kronecker_control_alg}} from {\rm Algorithm~\ref{alg:SM_algorithm}}.
Assume that the rate coefficients 
from~{\rm \ref{eq:Kronecker_control_alg}} satisfy the 
\emph{kinetic conditions}, given by 
\begin{align}
\mu^{x_{i, j} + 1} \left( \prod_{m = 1}^{x_{i, j} + 1} \gamma_{0,j,m} \right) \gamma_{i, j} & = (\varepsilon \sigma)^{-1}, 
\, \, \, 
\mu \gamma_{0,j,x_{i, j}}, \, \mu \gamma_{0,j,c_j}, \, \mu \gamma_{i, j} \ll 1, \nonumber \\
&
\textrm{for } i \in \{1, 2, \ldots, M\},  
\; 
j \in \{1, 2, \ldots, n\}, 
\;
 \{x_{i,j}\}_{i = 1}^M \in \{1, 2, \ldots, c_j-1\},
\label{eq:kinetic_conditions_alg}
\end{align} 
with $\mathbf{c} = (c_1, c_2, \ldots, c_n) \in \mathbb{Z}_{>}^{n}$.
Then, as $\mu \to 0$, with $\varepsilon, \sigma = \mathcal{O}(1)$,
the {\rm PMF} of the network 
$\mathcal{R}_{\gamma}^{\delta}(\mathcal{X}_{\tau},  \mathcal{Z}; \, \mathcal{Y})$
 converges to the {\rm PMF} of the network
 $\mathcal{R}_{\gamma_0}^{\varepsilon}(\mathcal{X}_{\tau}; \, \varnothing) 
\cup_{i = 1}^M \mathcal{R}_{\gamma_i}^{\varepsilon,\sigma}(\mathcal{X}_{\tau}; \, Y_i)$, 
given by
\begin{align}
\mathcal{R}_{\gamma_0}^{\varepsilon}: \; 
& & \varnothing & \xrightarrow[]{1/\varepsilon}  X_j,
\; \; \; \; \; \; \; \; \; \; \; \; \; \; \; \; \; \; \; \; \; 
  \textrm{for } j \in \{1, 2, \ldots, n\}, \nonumber \\
\mathcal{R}_{\gamma_i}^{\varepsilon,\sigma}: \; 
& & Y_i + (x_{i, j} + 1) X_j &  \xrightarrow[]{1/( \sigma \varepsilon)}  Y_i + x_{i, j} X_j, 
\; \; \; \;
\textrm{for } i \in \{1, 2, \ldots, M\},  
\, \, j \in \{1, 2, \ldots, n\}, \nonumber \\
&&&
\; \; \; \; \; \; \; \; \; \; \; 
\; \; \; \; \; \; \; \; \; \; \; 
\; \; \; \; \; \; \; \; \; \; \;  
\; \; \; \; \; \; \; \; \; \; \; 
\;  \; \,
 \{x_{i,j}\}_{i = 1}^M \in \{0, 1,\ldots, c_j-1\}.
\label{eq:Kronecker_control_reduced}
\end{align} 
}
\end{theorem} 
\begin{proof} 
See~\cite{Me_Bimolecular}.
\end{proof}

\subsection{Perturbation analysis: Limit $\varepsilon \to 0$} \label{sec:perturbations}
The CME induced by the output network~(\ref{eq:R_output}) is given by
\begin{equation}
\frac{\partial}{\partial t}  p_{\varepsilon}(\mathbf{x},\mathbf{y},t) 
= \mathcal{L}_{\varepsilon} p_{\varepsilon}(\mathbf{x},\mathbf{y},t) =
\left(
\frac{1}{\varepsilon} \mathcal{L}_{\gamma} + (\mathcal{L}_{\alpha} + \mathcal{L}_{\beta})
\right) 
p_{\varepsilon}(\mathbf{x},\mathbf{y},t),
 \label{eq:CME}
\end{equation}
where $\mathcal{L}_{\alpha}$, $\mathcal{L}_{\beta}$, and $\mathcal{L}_{\gamma}$ 
are the forward operators of the sub-networks $\mathcal{R}_{\alpha}$, $\mathcal{R}_{\beta}$, 
and $\mathcal{R}_{\gamma}^1$, respectively.
Here, $\mathbf{x} = (x_1, x_2, \ldots, x_N) \in \mathbb{Z}_{\ge}^{N}$
and $\mathbf{y} = (y_1, y_2, \ldots, y_M) \in \mathbb{Z}_{\ge}^{M}$
are the copy-number vectors of the input species $\mathcal{X} = \{X_1, X_2, \ldots, X_N\}$
and the controlling species $\mathcal{Y} = \{Y_1, Y_2, \ldots, Y_M\}$, respectively. 
Furthermore, we denote the copy-number vectors 
of the target and residual species, 
$\mathcal{X}_{\tau} = \{X_1, X_2, \ldots, X_n\}$ and $\mathcal{X}_{\rho} = \{X_{n+1}, X_{n+2}, \ldots, X_N\}$, 
by $\mathbf{x}_{\tau} = (x_1, x_2, \ldots, x_{n}) \in \mathbb{Z}_{\ge}^{n}$
and $\mathbf{x}_{\rho} = (x_{n+1}, x_{n + 2}, \ldots, x_{N}) \in \mathbb{Z}_{\ge}^{N - n}$, 
respectively. 

The CME~(\ref{eq:CME}) involves a singularly
perturbed forward operator, which we now exploit by considering the following
perturbation series~\cite{Pavliotis}:
\begin{equation}
p_{\varepsilon}(\mathbf{x},\mathbf{y},t) 
= 
p_0(\mathbf{x},\mathbf{y},t) 
+ 
\varepsilon \, p_1(\mathbf{x},\mathbf{y},t) 
+ 
\ldots 
+ 
\varepsilon^i \, p_i(\mathbf{x},\mathbf{y},t) + \ldots, 
\; \; \; \;
\textrm{where } i \ge 2. \label{eq:perturbation_series}
\end{equation}
Here, $p_0(\mathbf{x},\mathbf{y},t)$ is required to be non-negative and normalized, 
and we call it the \emph{zero-order} PMF, while we require $p_i(\mathbf{x},\mathbf{y},t)$, 
called an $i$th-order corrector,  to be centered, $\langle 1, p_i(\mathbf{x},\mathbf{y},t) \rangle_{\mathbf{x},\mathbf{y}} 
= \sum_{\mathbf{x},\mathbf{y}} p_i(\mathbf{x},\mathbf{y},t) = 0$, for $i \in \{1, 2, \ldots \}$.
Substituting~(\ref{eq:perturbation_series}) into~(\ref{eq:CME}), and equating terms of 
equal powers in $\varepsilon$, 
the following system of equations is obtained:
\begin{align}
\mathcal{O} \left(\frac{1}{\varepsilon} \right): \; 
\mathcal{L}_{\gamma} \,
p_0(\mathbf{x},\mathbf{y},t) & = 0, \label{eq:QSAa1}\\
\mathcal{O}(1): \; 
\mathcal{L}_{\gamma} \,
p_1(\mathbf{x},\mathbf{y},t)  & = 
 \left(\frac{\partial}{\partial t} - (\mathcal{L}_{\alpha} + \mathcal{L}_{\beta})  \right) 
p_0(\mathbf{x},\mathbf{y},t). \label{eq:QSAa2}
\end{align}

\emph{Order $1/\varepsilon$ equation} (\ref{eq:QSAa1}). 
It follows from~(\ref{eq:R_gamma})--(\ref{eq:R_gamma_detail}) 
that the operator $\mathcal{L}_{\gamma}$ may be written
as the following linear combination
\begin{align}
\mathcal{L}_{\gamma} = \mathcal{L}_{\gamma_0} + \sum_{i = 1}^M y_i \mathcal{L}_{\gamma_i}, 
\label{eq:operator_gamma}
\end{align}
where $\mathcal{L}_{\gamma_i}$ is the forward operator 
of the sub-network $\mathcal{R}_{\gamma_i}^1(\mathcal{X}_{\tau}; \, \varnothing)$ from~(\ref{eq:R_gamma}),
for $i \in \{0,1, \ldots, M\}$.  Here, $\mathcal{R}_{\gamma_i}^{\varepsilon}(\mathcal{X}_{\tau}; \, \varnothing)$
is obtained by removing the catalyst $Y_i$ from the reactions underlying
 $\mathcal{R}_{\gamma_i}^{\varepsilon}(\mathcal{X}_{\tau}; \, Y_i)$, for $i \in \{1, 2, \ldots, M\}$.
Using the fact that $\mathcal{L}_{\gamma}$ acts only on the copy-numbers of the 
target species $\mathbf{x}_{\tau}$, and depends parametrically on $\mathbf{y}$,
the definition of conditional probability implies that
$p_0(\mathbf{x},\mathbf{y},t) = p_0(\mathbf{x}_{\tau}| \mathbf{y}) p_0(\mathbf{x}_{\rho},\mathbf{y},t)$,
 and~(\ref{eq:QSAa1}) becomes
\begin{align}
\left( \mathcal{L}_{\gamma_0} + \sum_{i = 1}^M y_i \mathcal{L}_{\gamma_i} \right)
p_0(\mathbf{x}_{\tau}| \mathbf{y}) = 0.
\label{eq:QSAa12}
\end{align}
 
\emph{Order $1$ equation} (\ref{eq:QSAa2}). 
Applying the $l^2$ inner-product
$\langle 1, \cdot \rangle_{\mathbf{x}_{\tau}} = (\sum_{\mathbf{x}_{\tau}} \cdot)$
on equation~(\ref{eq:QSAa2}), and using  the fact
 that $\mathcal{L}_{\beta}$ acts and depends only on $\mathbf{y}$,
leads to the solvability condition in a form of 
an \emph{effective} CME, describing the time-evolution
of the $(\mathbf{x}_{\rho}, \mathbf{y})$-marginal PMF, given by
\begin{align}
\frac{\partial}{\partial t} p_0(\mathbf{x}_{\rho},\mathbf{y},t) 
& = 
(\bar{\mathcal{L}}_{\alpha} 
+
\mathcal{L}_{\beta} )
 p_0(\mathbf{x}_{\rho},\mathbf{y},t),
\; \; \; \;
\textrm{where }
\bar{\mathcal{L}}_{\alpha} = 
\left \langle 1, \mathcal{L}_{\alpha} p_0(\mathbf{x}_{\tau}| \mathbf{y}) \right \rangle_{\mathbf{x}_{\tau}}.
\label{eq:effective_CME}
\end{align}
This motivates the following definition.
\begin{definition}[\textbf{Residual network}]  \label{def:residual}
Consider an input network $\mathcal{R}_{\alpha}$ with the forward operator
$\mathcal{L}_{\alpha}$, embedded into 
an output network~{\rm (\ref{eq:R_output})}--{\rm (\ref{eq:R_gamma_detail})}.
 The operator  $\bar{\mathcal{L}}_{\alpha} = \left \langle 1, \mathcal{L}_{\alpha} 
p_0(\mathbf{x}_{\tau}| \mathbf{y}) \right \rangle_{\mathbf{x}_{\tau}}$ 
is called the \emph{residual forward operator}, where
$p_0(\mathbf{x}_{\tau}| \mathbf{y})$ satisfies~{\rm (\ref{eq:QSAa12})}.
The reaction network induced by $\bar{\mathcal{L}}_{\alpha}$
is called the corresponding \emph{residual network}, 
and is denoted by
$\bar{\mathcal{R}}_{\alpha} = 
\bar{\mathcal{R}}_{\alpha}(\mathcal{X}_{\rho}; \, \mathcal{Y})$.
\end{definition}
The residual network $\bar{\mathcal{R}}_{\alpha}(\mathcal{X}_{\rho}; \, \mathcal{Y})$
is obtained by averaging the input network $\mathcal{R}_{\alpha}(\mathcal{X})$
over the faster species $\mathcal{X}_{\tau}$ conditioned on the slower species $\mathcal{Y}$.
Note that the controlling species $\mathcal{Y}$ play a catalytic role in the residual network, 
which we capture with the notation $\bar{\mathcal{R}}_{\alpha} = 
\bar{\mathcal{R}}_{\alpha}(\mathcal{X}_{\rho}; \, \mathcal{Y})$. 
See also Theorem~\ref{theorem:Poisson_Kronecker} and 
Section~\ref{sec:residual_networks} for more details on
the residual networks. 

The main object of interest in this paper is 
the zero-order marginal-PMF of the target species $\mathcal{X}_{\tau}$,
denoted by $p_0(\mathbf{x}_{\tau}, t)$, and given by
\begin{align}
p_0(\mathbf{x}_{\tau}, t) 
& = \sum_{\mathbf{y}} p_0(\mathbf{y}, t) p_0(\mathbf{x}_{\tau}| \mathbf{y}).
\label{eq:pox_general}
\end{align} 
Here, $p_0(\mathbf{x}_{\tau}| \mathbf{y})$ is a solution of~(\ref{eq:QSAa12}), 
while, applying the inner-product
$\langle 1, \cdot \rangle_{\mathbf{x}_{\rho}} = (\sum_{\mathbf{x}_{\rho}} \cdot)$  on the
equation~(\ref{eq:effective_CME}),
and using the fact that $\bar{\mathcal{L}}_{\alpha}$ acts only on $\mathbf{x}_{\rho}$, 
it follows that $p_0(\mathbf{y}, t)$ satisfies
\begin{align}
 \frac{\partial}{\partial t} p_0(\mathbf{y},t) 
& = \mathcal{L}_{\beta}  p_0(\mathbf{y},t).
\label{eq:effective_CME_y}
\end{align}

\subsubsection{Convergence} \label{sec:convergence}
We now provide conditions under which the PMF of the 
output network $\mathcal{R}_{\alpha, \beta, \gamma}$, given by~(\ref{eq:R_output}), converges 
to its zero-order approximation from the perturbation series~(\ref{eq:perturbation_series}),
thereby mathematically delineating the class of input networks $\mathcal{R}_{\alpha}$
which may be controlled with the stochastic morpher $\mathcal{R}_{\beta, \gamma}$.
\begin{theorem} \label{theorem:convergence} 
\textit{Consider the output network $\mathcal{R}_{\alpha,\beta,\gamma} = 
\mathcal{R}_{\alpha} \cup \mathcal{R}_{\beta} \cup \mathcal{R}_{\gamma}^{\varepsilon}$, 
given by~{\rm (\ref{eq:R_output})}--{\rm (\ref{eq:R_gamma_detail})}, on a bounded state-space. 
Let $p_{\varepsilon}$ be the {\rm PMF} of the output network, satisfying~{\rm (\ref{eq:CME})},
and let $p_{0}$ be the zero-order {\rm PMF}, satisfying~{\rm (\ref{eq:QSAa1})}.
Assume there exists a function $p_1$, satisfying~{\rm (\ref{eq:QSAa2})}, 
 which is bounded, and has a bounded time-derivative, for each time $t \ge 0$, 
and assume also that $p_{\varepsilon} = p_{0}$ initially, at time $t = 0$.
Then, $p_{\varepsilon} \to p_{0}$ as $\varepsilon \to 0$ over any finite time-interval, with
\begin{align}
 \|p_{\varepsilon}(\mathbf{x},\mathbf{y},t) - p_{0}(\mathbf{x},\mathbf{y},t) \|_1 & 
\le  c(T) \varepsilon, 
\; \; \; \;
\textrm{for } 0 \le t \le T, 
\; \; \; \;
 \textrm{as } \varepsilon \to 0,
\label{eq:convergence}
\end{align}
where $c(T)$ is a constant independent of $\varepsilon$.
}
\end{theorem} 

\begin{proof}
Let us write the PMF of the output network in the following form:
\begin{equation}
p_{\varepsilon}(\mathbf{x},\mathbf{y},t) 
= 
p_0(\mathbf{x},\mathbf{y},t) 
+ 
\varepsilon \, p_1(\mathbf{x},\mathbf{y},t) 
+ 
r_{\varepsilon}(\mathbf{x},\mathbf{y},t), 
 \label{eq:perturbation_series2}
\end{equation}
where $p_0 = p_0(\mathbf{x},\mathbf{y},t) $ and $p_1 = p_1(\mathbf{x},\mathbf{y},t) $ are the zero-order PMF
and a first-order corrector, respectively, satisfying~(\ref{eq:QSAa1})--(\ref{eq:QSAa2}), 
while $r_{\varepsilon} = r_{\varepsilon}(\mathbf{x},\mathbf{y},t)$ is a residual function.
Substituting~(\ref{eq:perturbation_series2})
into~(\ref{eq:CME}), and using~(\ref{eq:QSAa1})--(\ref{eq:QSAa2}),
one obtains a linear non-homogeneous ordinary-differential equation governing the time-evolution
of the residual function:
\begin{align}
\frac{\mathrm{d}}{\mathrm{d} t} r_{\varepsilon} (t) - \mathcal{L}_{\varepsilon} r_{\varepsilon} (t) & = 
\varepsilon \left( (\mathcal{L}_{\alpha} + \mathcal{L}_{\beta}) - \frac{\mathrm{d}}{\mathrm{d} t} \right) p_1(t),
\label{eq:error}
\end{align}
where $p_1(t) = p_1(\mathbf{x},\mathbf{y},t)$ and
 $r_{\varepsilon}(t) = r_{\varepsilon}(\mathbf{x},\mathbf{y},t)$ are interpreted as column-vectors, 
while $\mathcal{L}_{\varepsilon}$ as a matrix, on a bounded state-space. 
Assuming that $p_{\varepsilon}(0) = p_0(0)$, equation~(\ref{eq:perturbation_series2})
provides an initial condition for the residual function, given by
\begin{align}
r_{\varepsilon}(0) = - \varepsilon p_1(0). \label{eq:IC}
\end{align}
The solution to the initial-value problem~(\ref{eq:error})--(\ref{eq:IC})
is given by
\begin{align}
 r_{\varepsilon} (t) & = - \varepsilon e^{\mathcal{L}_{\varepsilon} t} p_1(0)  + 
\varepsilon \int_{0}^t e^{\mathcal{L}_{\varepsilon} (t - s)}  
\left( (\mathcal{L}_{\alpha} + \mathcal{L}_{\beta}) - \frac{\mathrm{d}}{\mathrm{d} s} \right) p_1(s) \mathrm{d} s.
\label{eq:errorsol}
\end{align}

Let $\| \cdot \|_1$ denote the $l^1$-norm over
the bounded state-spaces of $\mathbf{x}$ and $\mathbf{y}$,
as well as the induced matrix-operator norm. Applying 
$\| \cdot \|_1$ on~(\ref{eq:errorsol}), and using the fact that 
$\| e^{\mathcal{L}_{\varepsilon} t} \|_1 = 1$~\cite{Pavliotis}, 
one obtains
\begin{align}
 \| r_{\varepsilon} (t) \|_1 & \le  \varepsilon \left( \| p_1(0) \|_1 
+ t  \, \textrm{sup}_{0 \le s \le t} \left \| \left((\mathcal{L}_{\alpha} + \mathcal{L}_{\beta})- \frac{\mathrm{d}}{\mathrm{d} s} \right) p_1(s) \right \|_1 \right).
\label{eq:errorbound}
\end{align}
Assuming a first-order corrector $p_1(t)$ exists, which is bounded, with a bounded 
time-derivative $\mathrm{d} p_1(t)/\mathrm{d} t$, for each fixed $t \ge 0$, 
it follows from~(\ref{eq:errorbound}) that 
$r_{\varepsilon} (t) \to 0$ as $\varepsilon \to 0$ for each fixed $t \ge 0$. 
Furthermore, the residual
function is asymptotically given by $\| r_{\varepsilon} (t) \|_1  = \mathcal{O}(\varepsilon)$
for sufficiently small $0 < \varepsilon \ll 1$, which, together with equation~(\ref{eq:perturbation_series2}), 
implies~(\ref{eq:convergence}).
\end{proof}

\subsection{Lower- and higher-resolution control} \label{app:long_time}
In Section~\ref{sec:perturbations}, we have established a weak convergence result:
 under suitable conditions, the time-dependent 
zero-order PMF approximates well the time-dependent PMF of the general output network
 $\mathcal{R}_{\alpha, \beta, \gamma}$, given 
by~(\ref{eq:R_output})--(\ref{eq:R_gamma_detail}), 
for $0 < \varepsilon \ll 1$, over arbitrarily long (but finite) time-intervals.
In what follows, we analyze the stationary (time-independent) zero-order PMF,
which, under suitable conditions, approximates well the stationary 
PMF of the output network for $0 < \varepsilon \ll 1$.
Furthermore, we consider two particular
classes of the controller:
the lower- and higher-resolution stochastic morphers, 
given by $\mathcal{R}_{\beta}(\mathcal{Y})
\cup \mathcal{R}_{\gamma}^{\mathcal{P}}(\mathcal{X}_{\tau}; \, \mathcal{Y})$
and 
$\mathcal{R}_{\beta}(\mathcal{Y})
\cup \mathcal{R}_{\gamma}^{\delta}(\mathcal{X}_{\tau}, \mathcal{Z}; \, \mathcal{Y})$
in Algorithm~\ref{alg:SM_algorithm}, respectively. 
Note that one may also consider other choices
for the sub-networks $\mathcal{R}_{\beta}$ 
and $\mathcal{R}_{\gamma}^{\varepsilon}$, see~\cite{Me_Mixing}.

\emph{Network $\mathcal{R}_{\beta}$}. 
Let us choose the sub-network $\mathcal{R}_{\beta}$ to be
 given by~(\ref{eq:R_beta_alg}) in Algorithm~\ref{alg:SM_algorithm}.
The stationary $\mathbf{y}$-marginal PMF is the normalized solution
of  $\mathcal{L}_{\beta}  p_0(\mathbf{y}) = 0$, obtained by
setting the left-hand side in~(\ref{eq:effective_CME_y}) to zero.
The structure of~(\ref{eq:R_beta_alg})  implies that, in
the long-run, its state-space is given by 
$\mathcal{S}_y = \{\mathbf{y} \in \{\mathbf{e}_i\}_{i = 1}^M \, | \, 
\mathbf{e}_i \in \mathbb{R}^M, \, \textrm{for } i \in \{1, 2, \ldots, M\}\}$, 
where $\mathbf{e}_i$ denotes the $i$th standard Euclidean basis vector, 
whose $i$th element equals one, while the rest are zero.
In other words, the long-time state space of~(\ref{eq:R_beta_alg})
 is constrained by the linear kinetic conservation law
$\sum_{i = 1}^M y_i = 1$. 
As a consequence, the stationary PMF
of~(\ref{eq:R_beta_alg})  is equivalent to the stationary PMF of the first-order
conversion network, with an initial condition being element of $\mathcal{S}_y$,
given by
\begin{align}
Y_1 \xrightarrow[]{\beta_{1,2}} Y_2 
       \xrightarrow[]{\beta_{2,3}} Y_3 \xrightarrow[]{\beta_{3,4}} \ldots 
	   \xrightarrow[]{\beta_{M-1,M}} Y_{M}
	  \xrightarrow[]{\beta_{M,1}} Y_{1},
\; \; \; \; 
\textrm{with } \sum_{i = 1}^M Y_i(0) = 1,
\label{eq:effective_R_beta}
\end{align}
which takes a multinomial product-form~\cite{David1,Me_Mixing}, given by
\begin{align}
p_0(\mathbf{y}) = \left(
\sum_{l = 1}^M 
\left(\beta_{l,l+1} (1 - \delta_{l,M}) + \beta_{M,1} \delta_{l,M} \right)^{-1}
\right)^{-1}
\, 
\prod_{i = 1}^{M} 
\left(\beta_{i,i+1} (1 - \delta_{i,M}) + \beta_{M,1} \delta_{i,M}
\right)^{- y_i},
\label{eq:p0y}
\end{align} 
with $\sum_{i = 1}^M y_i = 1$, and where $\delta_{x, x_0}$ denotes the Kronecker-delta PMF
centered at $x = x_0$.

Note that the convergence of $p_{0}(\mathbf{y},t)$ to the
stationary PMF~(\ref{eq:p0y}) may be sped-up by increasing the rate coefficient $\beta_{1,1}$ from 
the sub-network~(\ref{eq:R_beta_alg}), i.e. by taking $\beta_{i,j} \ll \beta_{1,1}$, 
for $(i, j) \ne (1,1)$. If desired, the convergence rate may be further increased 
by adding suitable (faster) reactions to~(\ref{eq:R_beta_alg}),
which also constrain the state-space of the species $\mathcal{Y}$
to $\mathcal{S}_y$,  such as the reactions 
$2 Y_i \xrightarrow[]{\beta_{i,i}} Y_i$, for $i \in \{2, 3, \ldots, M\}$. 
Note also that the sample paths of the 
network~(\ref{eq:effective_R_beta}) may be readily characterized.
In particular, given that the state of the network~(\ref{eq:effective_R_beta})
is $\mathbf{y} = \mathbf{e}_i$, the corresponding 
holding time (i.e. the time spent in the state
$\mathbf{y} = \mathbf{e}_i$) is an exponentially distributed
random variable with mean $1/\beta_{i,i+1}$ for $i \in \{1, 2, \ldots, M-1\}$
(and $1/\beta_{M,1}$ for $i = M$), after which the system 
jumps with probability one (deterministically) to the
state $\mathbf{y} = \mathbf{e}_{i+1}$ for $i \in \{1, 2, \ldots, M-1\}$
(and $\mathbf{y} = \mathbf{e}_{1}$ for $i = M$). 
Hence, while re-scaling the rate coefficient
vector $\boldsymbol{\beta}$ does not influence the 
stationary PMF of the network~(\ref{eq:effective_R_beta})
(this is equivalent to re-scaling the time-variable in~(\ref{eq:effective_CME_y})), 
it does modify the behavior of the underlying sample paths,
by changing the holding times. 

Substituting~(\ref{eq:p0y}) into~(\ref{eq:pox_general}), one obtains 
\begin{align}
p_0(\mathbf{x}_{\tau}) 
& = \sum_{i = 1}^M p_0(\mathbf{e}_i) p_0(\mathbf{x}_{\tau}| \mathbf{e}_i)
=  \sum_{i = 1}^M a_i(\boldsymbol{\beta}) p_{\gamma_i}(\mathbf{x}_{\tau}),
\label{eq:p0x}
\end{align} 
where $a_i(\boldsymbol{\beta}) \equiv p_0(\mathbf{e}_i)$, i.e.
\begin{align}
a_i(\boldsymbol{\beta}) & =  
\left(
\sum_{l = 1}^M 
\left(\beta_{l,l+1} (1 - \delta_{l,M}) + \beta_{M,1} \delta_{l,M} \right)^{-1}
\right)^{-1}
\, 
\left(\beta_{i,i+1} (1 - \delta_{i,M}) + \beta_{M,1} \delta_{i,M}
\right)^{-1},
\label{eq:PMF_coefficients}
\end{align}
and $p_{\gamma_i}(\mathbf{x}_{\tau}) \equiv p_0(\mathbf{x}_{\tau}| \mathbf{e}_i)$
is a solution of~(\ref{eq:QSAa12}) with $\mathbf{y} = \mathbf{e}_i$, i.e.
\begin{align}
(\mathcal{L}_{\gamma_0} + \mathcal{L}_{\gamma_i}) p_{\gamma_i}(\mathbf{x}_{\tau}) = 0, \; \; \; \; \; \; 
\textrm{for } i \in \{1, 2, \ldots, M\}.
\label{eq:auxiliary}
\end{align}
Let us now consider the two choices for the network $\mathcal{R}_{\gamma}^{\varepsilon}$
from Algorithm~\ref{alg:SM_algorithm}. 

\begin{theorem} \label{theorem:Poisson_Kronecker} 
\textit{Consider the input network $\mathcal{R}_{\alpha}$ with the forward operator 
$\mathcal{L}_{\alpha}$, given by~{\rm (\ref{eq:general_networks})} 
and~{\rm (\ref{eq:CME_general})}, respectively. Consider also a
corresponding output network $\mathcal{R}_{\alpha,\beta,\gamma} = 
\mathcal{R}_{\alpha} \cup \mathcal{R}_{\beta} \cup \mathcal{R}_{\gamma}^{\varepsilon}$, 
given by~{\rm (\ref{eq:R_output})}--{\rm (\ref{eq:R_gamma_detail})},
with the sub-network $\mathcal{R}_{\beta}$ fixed to~{\rm (\ref{eq:R_beta_alg})} 
from {\rm Algorithm~\ref{alg:SM_algorithm}}. 
\begin{enumerate}
\item[{\rm \textbf{(i)}}] \textbf{Lower-resolution control}.
If $\mathcal{R}_{\gamma}^{\varepsilon} = \mathcal{R}_{\gamma}^{\mathcal{P}}(\mathcal{X}_{\tau}; \, \mathcal{Y})$,
where $\mathcal{R}_{\gamma}^{\mathcal{P}}(\mathcal{X}_{\tau}; \, \mathcal{Y}) 
= \mathcal{R}_{\gamma_0}^{\varepsilon}(\mathcal{X}_{\tau}; \, \varnothing) 
\cup_{i = 1}^M \mathcal{R}_{\gamma_i}^{\varepsilon}(\mathcal{X}_{\tau}; \, Y_i)$
is given by~{\rm (\ref{eq:Poisson_control_alg})} in {\rm Algorithm~\ref{alg:SM_algorithm}}, 
then the stationary zero-order $\mathbf{x}_{\tau}$-marginal {\rm PMF}~{\rm (\ref{eq:p0x})}
of the output network is given by
\begin{align}
p_0(\mathbf{x}_{\tau}) & = 
\sum_{i = 1}^M a_i(\boldsymbol{\beta}) 
\prod_{j = 1}^{n} \mathcal{P}(x_j; \, \frac{\gamma_{i,j}}{\gamma_{0,j}}),
\; \; \; \;
\textrm{for } 0 < \varepsilon \ll 1,
\label{eq:p0x_Poisson}
\end{align} 
where $\mathcal{P}(x; \, \Lambda)$ denotes the Poisson distribution
with mean $\Lambda$, and
the coefficients $\{a_i(\boldsymbol{\beta})\}_{i = 1}^M$ are given by~{\rm (\ref{eq:PMF_coefficients})}. 
Furthermore, the residual forward operator $\bar{\mathcal{L}}_{\alpha}^{\mathcal{P}} = \bar{\mathcal{L}}_{\alpha}$, 
defined in {\rm Definition}~{\rm \ref{def:residual}}, is given by
\begin{align}
 \bar{\mathcal{L}}_{\alpha}^{\mathcal{P}}& = 
\sum_{j = 1}^A (E_{\mathbf{x}_u}^{-\Delta \mathbf{x}_{j,\rho}} - 1) \alpha_{j}
\left( \sum_{i = 1}^M y_i  \prod_{l = 1}^n 
\left( \frac{\gamma_{i,l}}{\gamma_{0,l}} \right)^{\nu_{j,l}} \right)
 \prod_{l = n+1}^N x_l^{\underline{\nu_{j, l}}}, \label{eq:effective_Poisson}
\end{align}
where $\Delta \mathbf{x}_{j,\rho} = (\Delta x_{j, n+1}, \Delta x_{j, n+2}, \ldots, \Delta x_{j, N}) 
\in \mathbb{Z}^{N - n}$. 
\item[{\rm \textbf{(ii)}}]  \textbf{Higher-resolution control}.
If $\mathcal{R}_{\gamma}^{\varepsilon} = \mathcal{R}_{\gamma}^{\delta}(\mathcal{X}_{\tau}, \mathcal{Z}; \, \mathcal{Y})$, 
where $\mathcal{R}_{\gamma}^{\delta}(\mathcal{X}_{\tau}, \mathcal{Z}; \, \mathcal{Y})
= \mathcal{R}_{\gamma_0}^{\mu,\varepsilon,\sigma}(\mathcal{X}_{\tau}, \mathcal{Z}; \, \varnothing) 
\cup_{i = 1}^M \mathcal{R}_{\gamma_i}^{\mu,\varepsilon,\sigma}(\mathcal{Z}; \, Y_i)$
 is given by~{\rm (\ref{eq:Kronecker_control_alg})} in {\rm Algorithm~\ref{alg:SM_algorithm}}, 
with the truncation vector $\mathbf{c} = (c_1, c_2, \ldots, c_n) \in \mathbb{Z}_{>}^{n}$,
and if the kinetic conditions~{\rm (\ref{eq:kinetic_conditions_alg})} are satisfied,
then the stationary zero-order $\mathbf{x}_{\tau}$-marginal {\rm PMF}~{\rm (\ref{eq:p0x})}
of the output network is given by
\begin{align}
p_0(\mathbf{x}_{\tau}) & = 
\sum_{i = 1}^M a_i(\boldsymbol{\beta}) 
\prod_{j = 1}^{n} \delta_{x_j, x_{i, j}},
\; \; \; \; 
\textrm{for }  \{x_{i,j}\}_{i = 1}^M \in \{0, 1,\ldots, c_j-1\}, 
\; \; 
0 < \mu \ll \varepsilon, \sigma \ll 1,
\label{eq:p0x_Kronecker}
\end{align} 
where $\delta_{x, x_0}$ denotes the Kronecker-delta distribution
centered at $x = x_0$, and the coefficients $\{a_i(\boldsymbol{\beta})\}_{i = 1}^M$ 
are given by~{\rm (\ref{eq:PMF_coefficients})}. 
Furthermore, the residual forward operator $\bar{\mathcal{L}}_{\alpha}^{\delta} = \bar{\mathcal{L}}_{\alpha}$ 
is given by
\begin{align}
 \bar{\mathcal{L}}_{\alpha}^{\delta}
& = \sum_{j = 1}^A (E_{\mathbf{x}_u}^{-\Delta \mathbf{x}_{j,\rho}} - 1) \alpha_{j}
\left( \sum_{i = 1}^M y_i
\prod_{l = 1}^n x_{i,l}^{\underline{\nu_{j,l}}} \right)
 \prod_{l = n+1}^N x_l^{\underline{\nu_{j, l}}}. \label{eq:effective_delta}
\end{align}
\end{enumerate}
}
\end{theorem} 
\begin{proof}
\textbf{(i) Lower-resolution control.}

If $\mathcal{R}_{\gamma}^{\varepsilon}$ is given by~\ref{eq:Poisson_control_alg},
then the operators $\{\mathcal{L}_{\gamma_i}\}_{i = 0}^M$ 
from~(\ref{eq:operator_gamma}) read
\begin{align}
\mathcal{L}_{\gamma_0} & = \sum_{j = 1}^n (E_{x_j}^{+1} - 1) \gamma_{0,j} x_j, \nonumber \\
\mathcal{L}_{\gamma_i} & = \sum_{j = 1}^n (E_{x_j}^{-1} - 1) \gamma_{i,j}, 
\; \; \; \; 
\textrm{for } i \in \{1, 2, \ldots, M\}. \nonumber
\end{align}
The solution of the equation~(\ref{eq:auxiliary}) 
may be written in the product-form, 
with each factor being a Poisson PMF:
\begin{align}
p_{\gamma_i}(\mathbf{x}_{\tau}) & = \prod_{j = 1}^n \mathcal{P}(x_j; \, \frac{\gamma_{i,j}}{\gamma_{0,j}}),
\; \; \; \; 
\textrm{for } i \in \{1, 2, \ldots, M\}, \label{eq:Poisson_conditional}
\end{align}
which, upon substitution into~(\ref{eq:p0x}), leads to~(\ref{eq:p0x_Poisson}).
Substituting~(\ref{eq:Poisson_conditional}) into the expression for 
$\bar{\mathcal{L}}_{\alpha}$, given in Definition~\ref{def:residual}, 
and using the fact that $y_i y_j = \delta_{i,j}$,
one obtains the residual operator~(\ref{eq:effective_Poisson}).

\textbf{(ii) Higher-resolution control}. 

\emph{Limit $\mu \to 0$}. Taking first the limit $\mu \to 0$, with $\varepsilon, \sigma = \mathcal{O}(1)$,
the network $\mathcal{R}_{\gamma}^{\delta}(\mathcal{X}_{\tau}, \mathcal{Z}; \, \mathcal{Y})$
reduces to the effective network $\mathcal{R}_{\gamma}^{\varepsilon,\sigma}(\mathcal{X}_{\tau}; \, \mathcal{Y})
= \mathcal{R}_{\gamma_0}^{\varepsilon}(\mathcal{X}_{\tau}; \, \varnothing) 
\cup_{i = 1}^M \mathcal{R}_{\gamma_i}^{\varepsilon,\sigma}(\mathcal{X}_{\tau}; \, Y_i)$, 
given by~(\ref{eq:Kronecker_control_reduced}), see Theorem~\ref{theorem:Kronecker_mu}.

\emph{Limit $\varepsilon \to 0$}. Taking the limit $\varepsilon \to 0$,  with 
$\sigma = \mathcal{O}(1)$, the $\mathbf{x}_{\tau}$-marginal PMF of the effective output 
network $\mathcal{R}_{\alpha} \cup \mathcal{R}_{\beta} \cup 
\mathcal{R}_{\gamma}^{\varepsilon,\sigma}$ is given by~(\ref{eq:p0x}). 
Equation~(\ref{eq:auxiliary}) may be written in the following form
\begin{align}
(\frac{1}{\sigma} \mathcal{L}_{\gamma_i} + \mathcal{L}_{\gamma_0}) 
p_{\gamma_i}(\mathbf{x}_{\tau}; \, \sigma) = 0, 
\; \; \; \; 
\textrm{for } i \in \{1, 2, \ldots, M\},
\label{eq:auxiliary_Kronecker}
\end{align}
with the operators $\{\mathcal{L}_{\gamma_i}\}_{i = 0}^M$ given by
\begin{align}
\mathcal{L}_{\gamma_0} & = \sum_{j = 1}^n (E_{x_j}^{-1} - 1), \nonumber \\
\mathcal{L}_{\gamma_i} & = \sum_{j = 1}^n (E_{x_j}^{+1} - 1) x_j^{\underline{(x_{i,j} + 1)}},
\; \; \; \; 
\textrm{for } i \in \{1, 2, \ldots, M\}. \nonumber
\end{align}

\emph{Limit $\sigma \to 0$}.
Substituting the perturbation series:
\begin{equation}
p_{\gamma_i}(\mathbf{x}_{\tau}; \, \sigma) 
= 
p_{\gamma_i}^0(\mathbf{x}_{\tau}) 
+ 
\sigma p_{\gamma_i}^1(\mathbf{x}_{\tau}) 
+ 
\ldots 
+ 
\sigma^j \, p_{\gamma_i}^j(\mathbf{x}_{\tau}) + \ldots, 
\; \; \; \;
\textrm{where } j \ge 2, 
\nonumber
\end{equation}
 into~(\ref{eq:auxiliary_Kronecker}), and equating terms of 
equal powers in $\sigma$, one obtains
\begin{align}
\mathcal{O} \left(\frac{1}{\sigma} \right): \; 
\mathcal{L}_{\gamma_i} \,
p_{\gamma_i}^0(\mathbf{x}_{\tau}) & = 0, \label{eq:QSAa1_Kronecker}\\
\mathcal{O}(1): \; 
\mathcal{L}_{\gamma_i} \,
p_{\gamma_i}^1(\mathbf{x}_{\tau})   & = 
  - \mathcal{L}_{\gamma_0} p_{\gamma_i}^0(\mathbf{x}_{\tau}). 
\label{eq:QSAa2_Kronecker}
\end{align}

\emph{Order $1/\sigma$ equation} (\ref{eq:QSAa1_Kronecker}). 
It follows from the structure of the forward operator $\mathcal{L}_{\gamma_i}$
that the PMF $p_{\gamma_i}^0(\mathbf{x}_{\tau})$ takes the product-form
\begin{align}
p_{\gamma_i}^0(\mathbf{x}_{\tau}) & = \prod_{j = 1}^n p_{\gamma_i}^0(x_j; \, x_{i, j}), 
\; \; \; \; 
\textrm{where }
p_{\gamma_i}^0(x_j; \, x_{i, j}) = 0, 
\; 
\textrm{for } x_j \ge (x_{i, j} + 1), 
\; 
i \in \{1, 2, \ldots, M\}.
\label{eq:Kronecker_condition_1}
\end{align}

\emph{Order $1$ equation} (\ref{eq:QSAa2_Kronecker}). 
The null-space of the backward operator $\mathcal{L}_{\gamma_i}^*$
 is given by $\mathcal{N}(\mathcal{L}_{\gamma_i}^*) = 
\prod_{j = 1}^n \{1_{x_j}, \delta_{x_j, 0}, \delta_{x_j, 1} \ldots, \delta_{x_j, x_{i,j}-1}\}$,
where $1_{x_j}$ denotes functions independent of $x_j$.
The Fredholm alternative theorem~\cite{Pavliotis} implies that the solvability conditions
are given by $0 = \langle u_j(\mathbf{x}_{\tau}),  \mathcal{L}_{\gamma_0} 
p_{\gamma_i}^0(\mathbf{x}_{\tau}) \rangle_{\mathbf{x}_{\tau}}$, 
where $u_j(\mathbf{x}_{\tau}) = u_j(x_1) u_j(x_2) \ldots u_j(x_n) \in \mathcal{N}(\mathcal{L}_{\gamma_i}^*)$.
Taking $u(x_j) \in \{1_{x_j}, \delta_{x_j, 0}, \delta_{x_j, 1} \ldots, \delta_{x_j, x_{i,j}-1}\}$
and $\{u(x_m) = 1_{x_m} \}_{m = 1, m \ne j}^n$ implies that 
\begin{align}
p_{\gamma_i}^0(x_j; \, x_{i, j}) = 0, 
\; \; \; \;
\textrm{for } x_j \le (x_{i, j} - 1), 
\; 
i \in \{1, 2, \ldots, M\}.
\label{eq:Kronecker_condition_2}
\end{align}
Conditions~(\ref{eq:Kronecker_condition_1})--(\ref{eq:Kronecker_condition_2})
jointly imply that $p_{\gamma_i}^0(x_j; \, x_{i, j})  = \delta_{x_j, x_{i,j}}$,
and $p_{\gamma_i}^0(\mathbf{x}_{\tau}) = \prod_{j = 1}^n \delta_{x_j, x_{i,j}}$, 
which, upon substitution into~(\ref{eq:p0x}), leads to~(\ref{eq:p0x_Kronecker}), 
while substituting into the definition
of $\bar{\mathcal{L}}_{\alpha}$ leads to~(\ref{eq:effective_delta}).
\end{proof}
\noindent Under the assumption of suitably well-behaved
residual networks, Theorem~\ref{theorem:Poisson_Kronecker}
implies that the stochastic morphers $\mathcal{R}_{\beta}(\mathcal{Y})
\cup \mathcal{R}_{\gamma}^{\mathcal{P}}(\mathcal{X}_{\tau}; \, \mathcal{Y})$
and
$\mathcal{R}_{\beta}(\mathcal{Y})
\cup \mathcal{R}_{\gamma}^{\delta}(\mathcal{X}_{\tau}, \mathcal{Z}; \, \mathcal{Y})$
are robust, according to Definition~\ref{def:robust}. 
In particular, the $\mathbf{x}_{\tau}$-marginal PMFs~(\ref{eq:p0x_Poisson}) and~(\ref{eq:p0x_Kronecker}) 
depend only on the parameters $\boldsymbol{\beta}$ and $\boldsymbol{\gamma}$
from the controller in the asymptotic limit $\varepsilon \to 0$
(i.e. the $\mathbf{x}_{\tau}$-marginal PMFs are independent 
of the parameters $\boldsymbol{\alpha}$ from the input network),
so that the same is true for \emph{all} of the
stationary statistics of the target species $\mathcal{X}_{\tau}$.

Assuming convergence, Theorem~\ref{theorem:Poisson_Kronecker} ensures that, 
under the action of the stochastic morpher, 
the stationary marginal-PMF of the target species $\mathcal{X}_{\tau}$ 
morphs into a PMF which is a linear combination 
of appropriate basis functions, 
whose form and centers depend on the stoichiometry and the rate coefficients $\boldsymbol{\gamma}$ from 
the interfacing sub-network $\mathcal{R}_{\gamma}^{\varepsilon}$, 
while the weights depend on the stoichiometry and the rate coefficients 
$\boldsymbol{\beta}$ from the sub-network $\mathcal{R}_{\beta}$. 
Furthermore, note that by choosing $\mathcal{R}_{\beta}$
to be given by~(\ref{eq:R_beta_alg}), one also gains a control
over the underlying long-time sample paths (strong control), which switch between 
the modes of the stationary PMF at exponentially distributed random times (whose average
is controllable via the rate coefficients $\boldsymbol{\beta})$, but
in a predefined deterministic order, owning to the fact that each of the 
first-order conversion reaction from~(\ref{eq:R_beta_alg}) is irreversible. 

More specifically, choosing $\mathcal{R}_{\gamma}^{\varepsilon}$ to be
 the first-order (uni-molecular) network 
$\mathcal{R}_{\gamma}^{\mathcal{P}}(\mathcal{X}_{\tau}; \, \mathcal{Y})$,
given by~(\ref{eq:Poisson_control_alg}), 
allows one to design PMFs in the space spanned
by the non-negative linear combinations of the Poisson-product functions. 
In this space of functions, one may construct
PMFs with predefined modes (maxima of the PMFs): 
the $i$th summand from~(\ref{eq:p0x_Poisson}) 
peaks at $(x_1, x_2, \ldots, x_n) = (\gamma_{i,1}/\gamma_{0,1}, 
\gamma_{i,2}/\gamma_{0,2}, \ldots, \gamma_{i,n}/\gamma_{0,n})$
with the amplitude given by $a_i(\boldsymbol{\beta})$, as defined in~(\ref{eq:PMF_coefficients}), 
allowing for the design of reaction networks displaying multi-modality/multi-stability.
On the other hand, choosing $\mathcal{R}_{\gamma}^{\varepsilon}$ to be
 the second-order (bi-molecular) network 
$\mathcal{R}_{\gamma}^{\delta}(\mathcal{X}_{\tau}, \mathcal{Z}; \, \mathcal{Y})$,
given by~{\rm (\ref{eq:Kronecker_control_alg})},
allows one to construct arbitrary PMFs defined on a bounded-domain,
 $P(\cdot) : \prod_{j = 1}^n [0, c_j-1] \to [0,1]$, with the state-space truncation
vector $\mathbf{c} = (c_1, c_2, \ldots, c_n) \in \mathbb{Z}_{>}^{n}$.
More precisely, $P(\mathbf{x})$ may be realized with
$M = \prod_{j = 1}^n c_j$ species $\{\{Z_{j, l} \}_{j = 1}^n\}_{l = 1}^{c_j}$  
and $\{Y_i\}_{i = 1}^{M}$, and choosing $\boldsymbol{\beta}$
such that $a_i(\boldsymbol{\beta}) = P(\mathbf{x}_i)$
for all $\mathbf{x}_i = (x_{i, 1}, x_{i, 2}, \ldots, x_{i, n}) \in  \prod_{j = 1}^n [0, c_j-1]$.
Let us note that the centers of the Poisson distributions are encoded kinetically, i.e.
 they are determined by the rate coefficients from~(\ref{eq:Poisson_control_alg}),
while the centers of the Kronecker-delta distributions are
encoded stoichiometrically, i.e. they are determined 
by which species $\mathcal{Z}$ is catalysed by $\mathcal{Y}$ in~(\ref{eq:Kronecker_control_alg}). 

\subsubsection{Hybrid control} \label{app:hybrid}
The lower- and higher-resolution
networks, $\mathcal{R}_{\gamma}^{\mathcal{P}}$ and $\mathcal{R}_{\gamma}^{\delta}$, 
respectively, may be combined into a composite hybrid scheme, capable of morphing
input PMFs into a mixture of Kronecker-delta and Poisson distributions.
In particular, consider the interfacing network $\mathcal{R}_{\gamma}^{\varepsilon} = 
\mathcal{R}_{\gamma}^{\mathcal{P},\delta}(\mathcal{X}_{\tau}, \mathcal{Z}; \, \mathcal{Y})
= \mathcal{R}_{\gamma_0}^{\mu,\varepsilon,\sigma}(\mathcal{X}_{\tau}, \mathcal{Z}; \, \varnothing) 
\cup_{i = 1}^{M_{\mathcal{P}}} \mathcal{R}_{\gamma_i}^{\varepsilon}(\mathcal{X}_{\tau}; \, Y_i)
\cup_{i = M_{\mathcal{P}}+1}^{M_{\mathcal{P}} + M_{\delta}} 
\mathcal{R}_{\gamma_i}^{\mu,\varepsilon,\sigma}(\mathcal{Z}; \, Y_i)$,
with $M_{\mathcal{P}}, M_{\delta} \in \mathbb{Z}_{\ge}$, and
\begin{align}
\mathcal{R}_{\gamma_0}^{\mu,\varepsilon,\sigma}: \; 
& & Z_{j,1} &  \xrightarrow[]{1/\mu} X_j, \nonumber \\
& & X_j + Z_{j,l} &  \xrightleftharpoons[1/\mu]{\gamma_{0,j,l+1}^{\delta}} Z_{j,l+1},
\; \; \; \; \; \; \; 
\textrm{for } j \in \{1, 2, \ldots, n\}, 
\; 
l \in \{1, 2, \ldots, c_j-1\}, \nonumber \\
\mathcal{R}_{\gamma_i}^{\varepsilon}: \; 
& & Y_i &  \xrightarrow[]{\gamma_{i,j}^{\mathcal{P}}/\varepsilon}  Y_i + X_j, \nonumber \\
& & Y_i + X_j & \xrightarrow[]{\tilde{\gamma}_{i,j}^{\mathcal{P}}/\varepsilon} Y_i, 
\; \; \; \; \; \; \; \; \;
\; \; \; \; \; \; 
\textrm{for } 
i \in \{1, 2, \ldots, M_{\mathcal{P}}\},  
\;
 j \in \{1, 2, \ldots, n\}, \nonumber \\
\mathcal{R}_{\gamma_i}^{\mu,\varepsilon,\sigma}: \; 
& & Y_i + X_j &  \xrightarrow[]{\gamma_{0,j,1}^{\delta}} Y_i + Z_{j,1}, \nonumber \\
& & Y_i & \xrightarrow[]{1/\varepsilon}  Y_i + X_j, \nonumber \\
& & Y_i + Z_{j,x_{i,j} + 1} &  \xrightarrow[]{\gamma_{i,j}^{\delta}}  Y_i + Z_{j,x_{i, j}}, 
\; \; \;  \;
\textrm{for } i \in \{M_{\mathcal{P}} + 1, M_{\mathcal{P}} + 2, \ldots, M_{\mathcal{P}} + M_{\delta}\}, 
\; j \in \{1, 2, \ldots, n\}, \nonumber \\
& & &
\; \; \; \; \; \; \; \; \;
\; \; \; \; \; \; \; \; \;
\; \; \; \; \; \; \; \; \;
\; \; \; \; \; \; \; \; \;
\; \; \; \; \; \; 
\textrm{for }  \{x_{i,j}\}_{i = M_{\mathcal{P}} + 1}^{M_{\mathcal{P}} + M_{\delta}}  \in \{0, 1,\ldots, c_j-1\}.
\label{eq:hybrid_control}
\end{align} 
Assume the kinetic conditions~(\ref{eq:kinetic_conditions_alg})
are satisfied by the rate coefficient with the superscript $\delta$ from~(\ref{eq:hybrid_control}).
Then, one can readily show that the stationary zero-order $\mathbf{x}_{\tau}$-marginal PMF,
on the domain bounded by the truncation vector $\mathbf{c} = (c_1, c_2, \ldots, c_n)$,
under the hybrid stochastic morpher $\mathcal{R}_{\beta}(\mathcal{Y}) \cup 
\mathcal{R}_{\gamma}^{\mathcal{P},\delta}(\mathcal{X}_{\tau}, \mathcal{Z}; \, \mathcal{Y})$, 
with $\mathcal{R}_{\beta}$ and $\mathcal{R}_{\gamma}^{\mathcal{P},\delta}$ given 
by~(\ref{eq:R_beta_alg}) and~(\ref{eq:hybrid_control}), respectively, 
is given by
\begin{align}
p_0(\mathbf{x}_{\tau}) & = 
\sum_{i = 1}^{M_{\mathcal{P}}} a_i(\boldsymbol{\beta}) 
\prod_{j = 1}^{n} \mathcal{P}(x_j; \, \frac{\gamma^{\mathcal{P}}_{i,j}}{\tilde{\gamma}^{\mathcal{P}}_{i,j}})
+
\sum_{i = M_{\mathcal{P}} + 1}^{M_{\mathcal{P}} + M_{\delta}} a_i(\boldsymbol{\beta}) 
\prod_{j = 1}^{n}\delta_{x_j, x_{i,j}},
\; \; \; \; 
\textrm{for }  \{x_{i,j}\}_{i = M_{\mathcal{P}} + 1}^{M_{\mathcal{P}} + M_{\delta}} \in \{0, 1,\ldots, c_j-1\}, \nonumber \\
& 
\; \; \; \; \; \; \; \; \; \; \; \; \; \; \; \; \; \; \; \; \; \; \; \; 
\; \; \; \; \; \; \; \; \; \; \; \; \; \; \; \; \; \; \; \; \; \; \; \; 
\; \; \; \; \; \; \; \; \; \; \; \; \; \; \; \; \; \; \; \; \; \; \; \; 
\; \; \; \; \; \; \; \; \; \; \; \; \; \; \; \; \; \; \; \; \; \; \; \; 
\; \; \; \; \; 
0 < \mu \ll \varepsilon, \sigma \ll 1,
\label{eq:p0x_hybrid}
\end{align} 
with  the coefficients 
$\{a_i(\boldsymbol{\beta})\}_{i = 1}^{M_{\mathcal{P}} + M_{\delta}}$  
given by~(\ref{eq:PMF_coefficients}).
The hybrid controller $\mathcal{R}_{\beta}(\mathcal{Y}) \cup 
\mathcal{R}_{\gamma}^{\mathcal{P},\delta}(\mathcal{X}_{\tau}, \mathcal{Z}; \, \mathcal{Y})$
designs Poisson PMFs centered at the $M_{\mathcal{P}}$ points
$(x_1, x_2, \ldots, x_n) = (\gamma_{i,1}/\tilde{\gamma}_{i,1},  \gamma_{i,2}/\tilde{\gamma}_{i,2},$ 
$\ldots, \gamma_{i,n}/\tilde{\gamma}_{i,n}) \in  \prod_{j = 1}^n [0, c_j-1]$ 
for $i \in \{1, 2, \ldots, M_{\mathcal{P}} \}$, 
and Kronecker-delta PMFs centered at the $M_{\delta}$ points 
$(x_1, x_2, \ldots, x_n) = (x_{i,1}, x_{i,2}, \ldots, x_{i,n}) \in  \prod_{j = 1}^n [0, c_j-1]$ 
for $i \in \{M_{\mathcal{P}} + 1, M_{\mathcal{P}} + 2, \ldots, M_{\mathcal{P}} + M_{\delta} \}$. 

\begin{example}
Consider the hybrid stochastic morpher
$\mathcal{R}_{\beta} \cup \mathcal{R}_{\gamma}^{\mathcal{P},\delta} = 
\mathcal{R}_{\beta}(Y_1, Y_2) \cup 
\mathcal{R}_{\gamma}^{\mathcal{P},\delta}(X, Z_1, Z_2; \, Y_1, Y_2)$, given by
\begin{align}
\mathcal{R}_{\beta}:
& & 2 Y_1 & \xrightarrow[]{\beta_{1,1}} Y_1 \xrightleftharpoons[\beta_{2,1}]{\beta_{1,2}} Y_2,
\nonumber \\
\mathcal{R}_{\gamma}^{\mathcal{P},\delta}: 
 & \hspace{1cm} \mathcal{R}_{\gamma_0}^{\mu,\varepsilon,\sigma}:  \hspace{-1.5cm}
&  Z_1 & \xrightarrow[]{1/\mu} X, \nonumber \\
&& X + Z_1 & \xrightleftharpoons[1/\mu]{\gamma_{0,2}^{\delta}} Z_2, \nonumber \\
& \hspace{1cm} \mathcal{R}_{\gamma_1}^{\varepsilon}: \hspace{-1.5cm}
& Y_1 & \xrightarrow[]{\gamma_{1}^{\mathcal{P}}/\varepsilon} Y_1 + X, \nonumber \\
& & Y_1 + X & \xrightarrow[]{\tilde{\gamma}_{1}^{\mathcal{P}}/\varepsilon} Y_1, \nonumber \\
& \hspace{1cm} \mathcal{R}_{\gamma_2}^{\mu,\varepsilon,\sigma}: \hspace{-1.5cm}
& Y_2 + X & \xrightarrow[]{\gamma_{0,1}^{\delta}} Y_2 + Z_1, \nonumber \\
 & & Y_2 & \xrightarrow[]{1/\varepsilon} Y_2 + X, \nonumber \\
 & & Y_2 + Z_2 & \xrightarrow[]{\gamma_{2}^{\delta}} Y_2 + Z_1, 
\; \; \; \; 0 < \mu \ll \varepsilon, \sigma \ll 1,
\label{eq:R_gamma_Kronecker_3}
\end{align}
which is obtained from~{\rm (\ref{eq:hybrid_control})} by taking
one target species $X \equiv X_1$, two controlling species
$Y_1$ and $Y_2$, and two mediating species $Z_1$ and $Z_2$. 
The species $Y_1$ is responsible for creating a Poisson
distribution centered at $x = (\gamma_{1}^{\mathcal{P}}/\tilde{\gamma}_{1}^{\mathcal{P}})$, 
while $Y_2$ generates a Kronecker-delta distribution centered
at $x = 1$. In particular, under the kinetic conditions, the controller~{\rm (\ref{eq:R_gamma_Kronecker_3})}
morphs an input {\rm PMF} into a bi-modal one, given by
\begin{align}
p_{0}(x) & = \left(1 + \frac{\beta_{1,2}}{\beta_{2,1}} \right)^{-1}
\mathcal{P} \left(x; \,  \frac{\gamma_{1}^{\mathcal{P}}}{\tilde{\gamma}_{1}^{\mathcal{P}}} \right) 
+ \left(1 + \frac{\beta_{2,1}}{\beta_{1,2}} \right)^{-1}
\delta_{x,1}, 
\; \; \; \; \textrm{for } 0 < \mu \ll \varepsilon, \sigma \ll 1. \label{eq:PMF_hybrid}
\end{align}
See {\rm Figure~\ref{fig:Kronecker}(c)--(d)} in the main text for the
 plots of the stochastic morphing induced when the hybrid 
controller~{\rm (\ref{eq:R_gamma_Kronecker_3})} is embedded
into the input network~{\rm (\ref{eq:input_1})}.
\end{example}

\subsection{Residual networks} \label{sec:residual_networks}
The residual network under the lower-resolution control,
denoted by 
$\bar{\mathcal{R}}_{\alpha}^{\mathcal{P}}(\mathcal{X}_{\rho}; \, \mathcal{Y}) = 
\bar{\mathcal{R}}_{\alpha}^{\mathcal{P}}$,
is induced by the effective forward operator~(\ref{eq:effective_Poisson}), 
and given by
\begin{align}
\bar{\mathcal{R}}_{\alpha}^{\mathcal{P}}(\mathcal{X}_{\rho}; \, \mathcal{Y}):
& & \sum_{l = n+1}^N \nu_{j, l} X_l + (1 - \delta_{\sum_{l = 1}^n \nu_{j,l},0}) Y_i & \xrightarrow[]
{\alpha_{j} \prod_{l = 1}^n \left(\gamma_{i,l}/\gamma_{0,l} \right)^{\nu_{j,l}}} 
\sum_{l = n+1}^N \bar{\nu}_{j, l}X_l + (1 - \delta_{\sum_{l = 1}^n \nu_{j,l},0}) Y_i, \nonumber \\
&&& \; \; 
\textrm{for } 
i \in \{1, 2, \ldots, M\}, 
\; \; 
j \in \{1, 2, \ldots, A\}.
\label{eq:residual_Poisson}
\end{align}
In other words, if the reactant complex in the $j$th reaction from the input network $\mathcal{R}_{\alpha}$
contains no target species $\mathcal{X}_{\tau}$ (i.e. $\delta_{\sum_{l = 1}^n \nu_{j,l},0} = 1$, 
and $\prod_{l = 1}^n \left(\gamma_{i,l}/\gamma_{0,l} \right)^{\nu_{j,l}} = 1$ for each $i \in \{1, 2, \ldots, M\}$), 
then such a reaction becomes the $j$th reaction in the residual network, without any modifications. 
Otherwise, the $j$th reaction from the input network 
gives rise to a family of $M$ reactions in the corresponding residual network, 
as given by~(\ref{eq:residual_Poisson}).

Analogously, the residual network under the higher-resolution control,
denoted by $\bar{\mathcal{R}}_{\alpha}^{\delta}(\mathcal{X}_{\rho}; \, \mathcal{Y}) = 
\bar{\mathcal{R}}_{\alpha}^{\delta}$, is induced by~(\ref{eq:effective_delta}), 
and reads
\begin{align}
\bar{\mathcal{R}}_{\alpha}^{\delta}(\mathcal{X}_{\rho}; \, \mathcal{Y}):
& & \sum_{l = n+1}^N \nu_{j, l} X_l + (1 - \delta_{\sum_{l = 1}^n \nu_{j,l},0}) Y_i & \xrightarrow[]
{\alpha_{j} \prod_{l = 1}^n x_{i,l}^{\underline{\nu_{j,l}}}} 
\sum_{l = n+1}^N \bar{\nu}_{j, l}X_l + (1 - \delta_{\sum_{l = 1}^n \nu_{j,l},0}) Y_i, \nonumber \\
&&& \; \; 
\textrm{for } 
i \in \{1, 2, \ldots, M\}, 
\; \; 
j \in \{1, 2, \ldots, A\}.
\label{eq:residual_delta}
\end{align}

When going from the input to the residual network,
the dynamics of the underlying residual species may undergo
qualitative changes.
For example, note that the $j$th reaction from the input network 
may be switched off in the corresponding residual network~(\ref{eq:residual_delta}),
and this occurs if there exist indices $l \in \{1,2, \ldots, n\}$ such that 
$x_{i,l} < \nu_{j,l}$ for each $i \in \{1, 2, \ldots, M\}$. Such 
changes in the network structure may induce bifurcations
in the dynamics of the underlying  residual species, 
which may have biochemical significance. 
Let us note that the residual networks may display blow-ups,  
in which case the assumptions made in Theorem~(\ref{theorem:convergence})
from Appendix~(\ref{sec:convergence}) may fail, 
so that the control imposed by the stochastic morpher may also fail.
Explosivity of residual networks can be studied using 
the methods put forward in e.g.~\cite{Engblom,Khammash3}.
On the other hand, the control remains successful
if the residual networks display multiple stationary PMFs, 
as outlined in the following example. 

\begin{example} \label{ex:ergodicity}
Consider the input network
\begin{align}
\mathcal{R}_{\alpha}: 
& & \varnothing & \xrightleftharpoons[\alpha_{2}]{\alpha_{1}} X_1, \hspace{0.3cm}
X_1 \xrightarrow[]{\alpha_{3}} X_1 + X_2, \hspace{0.3cm}
X_1 + X_2 \xrightarrow[]{\alpha_{4}} X_1, \hspace{0.3cm}
2 X_2 \xrightleftharpoons[\alpha_{6}]{\alpha_{5}} 3 X_2,
\label{eq:ex_ergodic}
\end{align}
with positive rate coefficients, $\boldsymbol{\alpha} \in \mathbb{R}_{>}^6$. 
Let us control the target species $X_1$, by embedding into~{\rm (\ref{eq:ex_ergodic})}
the stochastic morpher
\begin{align}
\mathcal{R}_{\beta}:
& & 2 Y_1 & \xrightarrow[]{\beta_{1,1}} Y_1, \nonumber \\
\mathcal{R}_{\gamma}^{\mathcal{P}}: 
& & \mathcal{R}_{\gamma_0}^{\varepsilon}: \; \; \; \; 
X & \xrightarrow[]{\gamma_{0}/\varepsilon} \varnothing, \nonumber \\
& & \mathcal{R}_{\gamma_1}^{\varepsilon}:  \; \; \; \; 
Y_1 & \xrightarrow[]{\gamma_{1}/\varepsilon} Y_1 + X, 
\; \; \; \; 0 < \varepsilon \ll 1,
\label{eq:R_gamma_Poisson_1_ex}
\end{align}
with $\gamma_1 = 0$. In this case, the stationary
$x_1$-marginal {\rm PMF} is given by $p_0(x_1) = 
\mathcal{P}(x_1; \, 0) = \delta_{x_1,0}$. On the other hand, 
it follows from~{\rm (\ref{eq:residual_Poisson})} that
the dynamics of the species $X_2$ is goverened
by the residual network
\begin{align}
\bar{\mathcal{R}}_{\alpha}: \; 
& & 2 X_2 & 
\xrightleftharpoons[\alpha_{6}]
{\alpha_{5}} 3 X_2.
\label{eq:ex_ergodic_residual}
\end{align}
The input network~{\rm (\ref{eq:ex_ergodic})} is jointly 
ergodic in both of the species $X_1$ and $X_2$, with the
third and fourth reactions, which are catalyzed by $X_1$, 
allowing $X_2$ to enter and exit the states $x_2 \in \{0,1\}$. 
On the other hand, 
the output network~{\rm (\ref{eq:ex_ergodic})}$\cup${\rm (\ref{eq:R_gamma_Poisson_1_ex})}, 
in the limit $\varepsilon \to 0$, 
is only marginally ergodic in the target species $X_1$. 
In particular, the residual network~{\rm (\ref{eq:ex_ergodic_residual})} is non-ergodic, 
as the third and fourth reactions from~{\rm (\ref{eq:ex_ergodic})}
are switched off, resulting in a reducible state-space for $X_2$, with the three irreducible
components given by $\{0\}$, $\{1\}$, and $\{x_2 | x_2 \ge 2\}$.
\end{example}

\label{lastpage}

\end{document}